\documentclass[aos]{imsart}

\RequirePackage{amsthm,amsmath,amsfonts,amssymb}
\RequirePackage[numbers,sort&compress]{natbib}
\RequirePackage[colorlinks,citecolor=blue,urlcolor=blue]{hyperref}
\RequirePackage{graphicx}
\usepackage{subfig}
\usepackage{subfloat}

\usepackage{times}
\usepackage{bm}
\usepackage{multirow}
\usepackage{booktabs}
\usepackage{threeparttable}
\usepackage{xcolor}
\RequirePackage{algorithm,algorithmic,tabularx,url,float,booktabs,xypic}

\startlocaldefs
\theoremstyle{plain}

\newtheorem{theorem}{Theorem}[section]
\newtheorem{lemma}[theorem]{Lemma}
\newtheorem{assumption}{Assumption}

\newtheorem{proposition}{Proposition}
\theoremstyle{definition}

\newtheorem{remark}{Remark}
\newtheorem{example}{Example}

\newcommand{\ve}[1]{\mbox{\boldmath ${#1}$}}
\newcommand{\vesub}[2]{\mbox{{\boldmath ${#1}$}$_{#2}$}}
\newcommand{\vesup}[2]{\mbox{{\boldmath ${#1}$}$^{#2}$}}

\newcommand{\BLUE}[1]{{\color{blue} #1}}

\newcommand{\pr}{{\rm pr}}
\newcommand{\var}{{\rm var}}
\newcommand{\cov}{{\rm Rolin}}
\newcommand{\CC}[2]{C_{#2}^{#1}}
\newcommand{\PP}[2]{P_{#2}^{#1}}
\newcommand{\rolin}{{\rm Rolin}}
\newcommand{\T}{\!\top\!}
\usepackage{xr-hyper}
\newcommand{\ConditionAuthor}[1]{%
  \ifcase#1
    {
    \begin{aug}
        \author[A]{\fnms{FIRST }~\snm{AUTHOR}\ead[label=e1]{first@somewhere.com}},
        \author[C]{\fnms{SECOND}~\snm{AUTHOR}\ead[label=e2]{second@somewhere.com}},
        \author[A]{\fnms{THIRD}~\snm{AUTHOR}\ead[label=e3]{third@somewhere.com}}
        \and
        \author[D]{\fnms{FOURTH}~\snm{AUTHOR}\ead[label=e4]{fourth@somewhere.com}}
        \address[A]{Department, University or Company name\printead[presep={,\ }]{e1,e3}}
        
      \address[C]{Department, University or Company name\printead[presep={,\ }]{e2}}

        \address[D]{Department, University or Company name \printead[presep={,\ }]{e4}}
    \end{aug}
    }% 0对应的情况
  \else
    {
    \begin{aug}
        \author[A]{\fnms{Zhe}~\snm{Gao}\ead[label=e1]{gaozh8@mail.ustc.edu.cn}},
        \author[B]{\fnms{Roulin}~\snm{Wang}\ead[label=e2]{rlwang@sfs.ecnu.edu.cn}},
        \author[A]{\fnms{Xueqin}~\snm{Wang}\ead[label=e3]{wangxq20@ustc.edu.cn}}
        \and
        \author[C]{\fnms{Heping}~\snm{Zhang}\ead[label=e4]
        {heping.zhang@yale.edu}\orcid{0000-0000-0000-0000}}
        \address[A]{
        University of Science and Technology of China\printead[presep={,\ }]{e1,e3}}
        \address[B]{
        East China Normal University \printead[presep={,\ }]{e2}}
        \address[C]{
        Yale University \printead[presep={,\ }]{e4}}
    \end{aug}
    }% 1对应的情况
  \fi
}

\endlocaldefs

\begin{document}

\begin{frontmatter}
\title{Studentized Tests of Independence: Random-Lifter approach} 
%\title{A sample article title with some additional note\thanksref{t1}}
\runtitle{Random-Lifter}
%\thankstext{T1}{A sample additional note to the title.}

\ConditionAuthor{1} % Option 1 Display author information, 0 anonymous

\begin{abstract}

%The exploration of associations between random objects exhibiting complex geometric structures has catalyzed the development of various novel testing methodologies, encompassing distance-based and kernel-based techniques. Predominantly, these methods have favorable characteristics. However, a common limitation is that their test statistics tend to converge to asymptotic null distributions described by second-order Wiener chaos. This necessitates using computationally intensive approximation or permutation techniques for delineating rejection regions. This work introduces an innovative approach, leveraging the association measures defined on Gaussian random fields to reformulate these testing procedures.  We present ``Random-Lifter'', a method engineered to ensure the resultant test statistics adhere to standard normal limit null distributions by harnessing the Central Limit Theorems (CLTs) applicable to degenerate U-statistics. This approach not only diminishes the computational load but also obviates the need for permutation testing methods. Our proposal for the minimax property of the test statistics is illustrated with the HISC test. We further substantiate that our method maintains competitive power against existing methods with minimal adjustments in constant factors. Both numerical simulations and real-data analysis corroborate the efficacy of the Random-Lifter method.
The exploration of associations between random objects with complex geometric structures 
has catalyzed the development of various novel statistical tests encompassing distance-based and kernel-based statistics. These methods have various strengths and limitations. One problem is that their test statistics tend to converge to asymptotic null distributions involving second-order Wiener chaos, which are hard to compute and need approximation or permutation techniques that use much computing power to build rejection regions. In this work, we take an entirely different and novel strategy by using the so-called ``Random-Lifter''. This method is engineered to yield test statistics with the standard normal limit under null distributions without the need for sample splitting. In other words, we set our sights on having simple limiting distributions and finding the proper statistics through reverse engineering. We use the Central Limit Theorems (CLTs) for degenerate U-statistics derived from our novel association measures to do this. As a result, the asymptotic distributions of our proposed tests are straightforward to compute. Our test statistics also have the minimax property. We further substantiate that our method maintains competitive power against existing methods with minimal adjustments to constant factors. Both numerical simulations and real-data analysis corroborate the efficacy of the Random-Lifter method.

\end{abstract}



\end{frontmatter}
%%%%%%%%%%%%%%%%%%%%%%%%%%%%%%%%%%%%%%%%%%%%%%
%Please use \tableofcontents for articles.
%% with 50 pages and more
%%%%%%%%%%%%%%%%%%%%%%%%%%%%%%%%%%%%%%%%%%%%%%
%\tableofcontents

\section{INTRODUCTION}

In the dynamic landscape of statistical research, exploring relationships between variables is pivotal and becomes more challenging because scientific and technological advancements give rise to data objects with complex structures.
At the heart of this evolving field lies the investigation of interactions among random objects. This difficult but important field has led to the creation of many new statistical methods that can find dependencies in non-Euclidean data spaces, such as Hilbert spaces, Banach spaces, and Metric spaces.  Distance covariance \cite{szekely2007measuring, lyons2013, SZEKELY2013193, Gao2021}, Hilbert-Schmidt Independence Criterion (HSIC) \cite{gretton2005measuring, gretton2005kernel, gretton2007kernel, pfister2018kernel, Melisande2022}, Grothendieck's covariance \cite{jiang2023},  Ball covariance \cite{pan2019ball,wang2021nonparametric}, copulas-based test \cite{schweizer1981nonparametric, lopez2013randomized,zhang2019bet}, and rank-based test \citep{kim2020,shi2022,deb2023} are some of the most important new methods.

These approaches have resonated across a wide spectrum of applications, from theoretical statistics to practical machine learning, thanks to their nuanced capability to capture dependencies in unconventional data configurations.

These dependency measures stand out not only for their practical relevance but also for their mathematical sophistication and rigor. A notable attribute is the ``independence-zero equivalence'' property, which elegantly captures the essence of statistical independence by indicating that a measure's zero value uniquely signifies the independence of the random objects in question \citep{pmlr-vR5-gretton05a,NIPS2007_Fukumizu,JMLR:v11:sriperumbudur10a, lyons2013,jiang2023,wang2021nonparametric}. Moreover, these measures can be articulated as complex expected values of random objects, a formulation amenable to estimation via U(V)-statistics from sample data, thus providing an effective mechanism for testing independence hypotheses.

However, a challenge emerges with the asymptotic distributions of these test statistics under the null hypothesis because they generally follow certain second-order Wiener chaos distributions. This characteristic complicates our ability to use the asymptotic distribution, for example, to compute the quantiles and handicap the practical use of those test statistics. Typically, the null distributions are approximated through either asymptotic Gamma distributions or permutation techniques that are computationally intensive \cite{gretton2007kernel,li2019optimality}. 
Some rank-based methods can determine the distribution-free parameters under the null hypothesis, but they still have some computational complexity issues.

To deal with the conundrum of non-standard normal null limiting distributions in statistical tests, various strategies have been devised for particular scenarios. For instance, the block-averaged HISC involves segmenting data into smaller blocks, computing HSIC for each, and then averaging these calculations, leading to a limiting Gaussian distribution under the null hypothesis for a suitable choice of block size \cite{zhang2018large}. Similarly, the cross HSIC method, through sample splitting and studentization, achieves asymptotic normality under the null hypothesis \cite{shekhar2023permutation}. Studies into the asymptotic normality of kernel and distance-based statistics, including but not limited to Grothendieck's covariance\citep{jiang2023}, maximum mean discrepancy (MMD) \citep{gao2023two}, and distance covariance \citep{Gao2021}, have been conducted in high-dimensional contexts. These studies typically employ a martingale structure to the studentized test statistics and then apply the martingale central limit theorem. These results are contingent upon the simultaneous expansion of both dimension and sample size towards infinity, coupled with specific constraints on the eigenvalues of certain integral operators, which may be challenging to verify in practical settings.

To circumvent the challenges of non-standard normal limiting distributions, we aim to construct test statistics with the standard normal limit under the null hypothesis  without the need for sample splitting. The key idea is to introduce a novel methodology, termed "random-lifter," which imparts a random weight to the HISC.  This method is engineered to yield test statistics with the standard normal limit under null distributions without the need for sample splitting. Specifically, we use the Central Limit Theorems (CLTs) to degenerate U-statistics derived from novel HISC.

Leveraging this newly formulated HSIC, we propose novel dependence measures and corresponding independence test statistics. These measures and statistics exhibit several key features:
\begin{itemize}
\item They preserve the independence-zero equivalence property under mild conditions. (cf. Theorem \ref{Thm::iff})
\item The test statistics, while being degenerate U(V)-statistics, have their studentizations approaching an asymptotic standard normal null limit distribution, simplifying p-value computation and enhancing testing power without the need for sample splitting. (cf. Theorem \ref{Thm::randomH0}).
\item They demonstrate asymptotic normality under the alternative hypothesis, with a detailed comparison of their power functions relative to the existing test statistics. (cf. Theorem \ref{random:H1} and Theorem \ref{random:H12}).
\item They elucidate the impact of random-lifter technique on the power function and establish the separate rate of the minimax property. (cf. Theorem \ref{random:H12}, Theorem \ref{thm::minimax1} and Theorem \ref{thm::minimax2}).
\end{itemize}

The remainder of this paper is organized as follows: Section 2 introduces the random-lifter technique and elaborates on novel dependence measures and their corresponding independence test statistics.  Section 3 delves into the theoretical efficacy of these tests, particularly highlighting the minimax property.  Section 4 showcases case studies to validate and demonstrate the practical utility of the proposed method.  Section 5 offers concluding remarks, encapsulating our main contributions.

\section{Methodology} 

Consider two random objects,  $\ve{X}$ and $\ve{Y}$, which are defined in topological spaces $\mathcal{X}$ and $\mathcal{Y}$, respectively. These objects are associated with their corresponding Borel probability measures, $\mu$ for $\ve{X}$ and $\nu$ for $\ve{Y}$. Additionally, we have a joint Borel probability measure, denoted as $\omega$, defined on the Cartesian product space $\mathcal{X} \times \mathcal{Y}$, representing the pair $(\ve{X}, \ve{Y})$ as $(\ve{X}, \ve{Y}) \sim \omega$.

The goal is to perform an independence test to determine whether there exists a dependency between these two random objects, $\ve{X}$ and $\ve{Y}$, utilizing a sample of independent realizations ${(\vesub{X}{i}, \vesub{Y}{i})}_{i = 1}^n$. This hypothesis-testing scenario is formally articulated as follows:

\begin{equation}
 H_0 : \omega = \mu \otimes \nu  \quad \text{vs.} \quad H_1 : \omega \neq \mu \otimes \nu,
    \label{indep}
\end{equation}
where $\otimes$ symbolizes the tensor product of the measures.

Various test procedures have been developed to assess dependencies between $\ve{X}$ and $\ve{Y}$. While these procedures share some properties with earlier methods, such as consistency and convergence of the null limiting distribution to second-order Wiener chaos, accurately determining distributions under the null hypothesis poses challenges. Computing these procedures often requires computationally intensive permutation or bootstrap tests, which become particularly burdensome with larger sample sizes.

To address these computational difficulties, we propose a novel random lifter method that leverages the CLTs to degenerate U-statistics and enhance computational efficiency. In this paper, we demonstrate the effectiveness of our method using HISC, a dependence measure, as an example. However, it is important to note that our approach can also be applied to other dependence measures.

We first recall the  HISC. Given a reproducing kernel Hilbert spaces (RKHS)
$\mathcal{H}_{\zeta}$ with $\zeta: \ve{X} \times \ve{X} \rightarrow \mathbf{R}$  as the reproducing kernel, which is symmetric
and positive definite, and an RKHS
$\mathcal{H}_{\eta}$ with $\eta: \ve{Y} \times \ve{Y} \rightarrow \mathbf{R}$  as the reproducing kernel, the HISC is defined as
\begin{align*}
&\operatorname{HSIC}(\omega, \zeta, \eta)\\
=& E\left[\zeta\left(\vesub{X}{1}, \vesub{X}{2}\right)\eta\left(\vesub{Y}{1}, \vesub{Y}{2}\right)\right]+E\left[\zeta\left(\vesub{X}{1}, \vesub{X}{2}\right)\right]E\left[\eta\left(\vesub{Y}{1}, \vesub{Y}{2}\right)\right] \\&- 2E\left[\zeta\left(\vesub{X}{1}, \vesub{X}{2}\right)\eta\left(\vesub{Y}{1}, \vesub{Y}{3}\right)\right]\\= &  E\left\{\left[\zeta\left(\vesub{X}{1}, \vesub{X}{2}\right)+\zeta\left(\vesub{X}{3}, \vesub{X}{4}\right)-\zeta\left(\vesub{X}{1}, \vesub{X}{3}\right)-\zeta\left(\vesub{X}{2}, \vesub{X}{4}\right)\right]\right. \\
& \left.\cdot\left[\eta\left(\vesub{Y}{1}, \vesub{Y}{2}\right)+\eta\left(\vesub{Y}{3}, \vesub{Y}{4}\right)-\eta\left(\vesub{Y}{1}, \vesub{Y}{3}\right)-\eta\left(\vesub{Y}{2}, \vesub{Y}{4}\right)\right]\right\} .
\end{align*}
 
\subsection{Random-lifter technique}

Given a random lifter, which is indeed a random variable, say $Z$, is independent of $(\ve{X}, \ve{Y})$, then testing the independence of $\ve{X}$ and $\ve{Y}$ is equivalent to testing the independence of $\ve{X}$ and $(\ve{Y}, Z)$. Let's assume that $\kappa_b$ is a positive-definite kernel defined on $\mathbf{R} \times \mathbf{R}$ with a tuning parameter $b$, associated with an RKHS $\mathcal{H}_{\kappa_b}$. Since the tensor product of the positive-definite kernel is also positive-definite, $\eta\kappa_b$ becomes a new positive-definite kernel on the product space $\mathcal{Y} \times \mathbf{R}$. Then our new new random-lifter dependence %\underline{r}and\underline{o}m-\underline{li}fter depe\underline{n}dence (Rolin) 
measure is defined as 

\begin{align*}
    &T_{\zeta, \eta, \kappa_b}(\ve{X}, \ve{Y}) \\
    =& E\left[\zeta\left(\vesub{X}{1}, \vesub{X}{2}\right)\eta\left(\vesub{Y}{1}, \vesub{Y}{2}\right)\kappa_b(Z_1, Z_2)\right]+E\left[\zeta\left(\vesub{X}{1}, \vesub{X}{2}\right)\right]E\left[\eta\left(\vesub{Y}{1}, \vesub{Y}{2}\right)\kappa_b(Z_1, Z_2)\right] \\
& - 2E\left[\zeta\left(\vesub{X}{1}, \vesub{X}{2}\right)\eta\left(\vesub{Y}{1}, \vesub{Y}{3}\right)\kappa_b(Z_1, Z_3)\right].
\end{align*}
Since $Z$ is independent of $\ve{X}, \ve{Y}$, it allows for separate analysis, leading to the connection between random-lifter dependence measure and HSIC as
$$T_{\zeta, \eta, \kappa_b}(\ve{X}, \ve{Y}) = \operatorname{HSIC}(\omega, \zeta, \eta)E\left[\kappa_b(Z_1, Z_2)\right].$$

A consistent moment estimator of $T_{\zeta, \eta, \kappa_b}(\ve{X}, \ve{Y})$, denoted by $T_{n, b}$ is obtained as follows:
\begin{align*}
    &T_{n, b}\\ = & \frac{1}{ \PP{n}{2}}\sum_{i \neq j} \zeta(\vesub{X}{i}, \vesub{X}{j})\eta(\vesub{Y}{i}, \vesub{Y}{j})\kappa_b(Z_i, Z_j) -  \frac{2}{\PP{n}{3}} \sum_{i \neq j \neq s}  \zeta(\vesub{X}{i}, \vesub{X}{j})\eta(\vesub{Y}{i}, \vesub{Y}{s}) \kappa_b(Z_i, Z_s)\\ & + \frac{1}{\PP{n}{4}} \sum_{i \neq j\neq s \neq t} \zeta(\vesub{X}{i}, \vesub{X}{j})\eta(\vesub{Y}{s}, \vesub{Y}{t}) \kappa_b(Z_s, Z_t),
\end{align*}
where $\PP{n}{r}$ denotes the number of such $r$-permutation of $n$-set, $\PP{n}{r} = \frac{n!}{(n-r)!}$.

However, unlike the HSIC above, the null limit distribution of $T_{n, b}$ follows a normal distribution under weak conditions. Moreover, let $\zeta$ and $\varrho$ be the $n \times n$ Gram matrices with entries $\zeta(\vesub{X}{i}, \vesub{X}{j})$ and $\eta(\vesub{Y}{i}, \vesub{Y}{j})\kappa_b(Z_i, Z_j)$. Define $M = ((H\zeta H) \circ (H \varrho H))^{\cdot 2}$, where $ \circ$ is the entrywise matrix product and $M^{\cdot 2}$ is the entrywise matrix power. Let
\begin{align} \label{varest}
S_{n, b}^2 = \frac{2(n-4)(n-5)}{n^2(n-1)^2(n-2)(n-3)}\sum_{i\neq j}M_{ij}.
\end{align}
Then, $S_{n, b}^2$ serves as an estimate of the variance of $T_{n, b}$ according to Proposition \ref{Prop::var} in the next subsection, which will be discussed in the next subsection. Thus, we introduce the \underline{r}and\underline{o}m-\underline{l}ifter \underline{in}dependence (Rolin) test statistic as follows:
\begin{align}\label{ROLIN}
\rolin_{n,b}=\frac{T_{n, b}}{S_{n, b}}.
\end{align}
In the subsequent theorem, we will demonstrate that the test statistic $\rolin_{n,b}$ follows a standard normal distribution under the null hypothesis. Hence, for a given significance level $\alpha$, we only need to determine the $\alpha$-quantile $\Phi(\alpha)$ of the standard normal distribution. If $\rolin_{n,b} > \Phi(\alpha)$, we reject the null hypothesis and conclude that $\ve{X}$ and $\ve{Y}$ are dependent. Conversely, we do not reject the null hypothesis if $\rolin_{n,b} \leq \Phi(\alpha)$.

%--------------------------------- Section 2.4
\subsection{Theoretical results}

In this subsection, we delve into the theoretical properties of our novel random-lifter dependence measure and its associated independence test statistic. We begin by discussing the properties associated with the random-lifter dependence measure. The following theorem establishes the fundamental property of our random-lifter dependence measure. 

\begin{theorem}[Independence-zero equivalence property]
The random-lifter dependence measure is non-negative, i.e.,
$T_{\zeta, \eta, \kappa_b}(\ve{X}, \ve{Y}) \ge 0$. Furthermore, if $\zeta, \eta, \kappa_b$ are positive-definite, the equality holds if and only if $\ve{X}$ and $\ve{Y}$ are independent.
\label{Thm::iff}
\end{theorem}

Next, we turn to the asymptotic properties of the test statistic $\rolin_{n,b}$. Since $\kappa_b$ can be constructed from some kernels, that is, $\kappa_b(z_1, z_2) =k\left(\frac{z_1 - z_2}{b}\right)$, we begin by introducing some assumptions on the kernel $\kappa$.
\begin{assumption}\label{ass:kerz} 
    %Assume $\kappa_b(Z_1, Z_2) = k\left(\frac{Z_1 - Z_2}{b}\right)$ 
     Assume $\kappa$ satisfies
    
    (i) $\kappa \geq 0$; 

    (ii) $\int k(z) dz = 1$, $\int k^{2 + \delta}(z) dz < \infty$,  for some $\delta > 0$;

    (iii) $k(z)$ is symmetric with $z = 0$.
\end{assumption}
These assumptions on the kernel $k$ are widely employed in various methods, including those using Gaussian and Laplace kernels. Now, we define some constants and functions for further analysis:
  
Let $Z$ be a random variable defined on $\mathbf{R}$ with support $\mathcal{D}$, and let $g$ be the density function of $Z$. We define the set $\mathcal{D}(z) = \{u | z + bu \in \mathcal{D}\}$ and the function $F_i(z) = \int k^i(u) \mathcal{I}\{u \in \mathcal{D}(z)\} d u$. Additionally, we define the following constants:

\begin{align*}
A_1 = \int F_1(z) g^2(z) d z, \quad A_2 = \int F_1^2(z) g^3(z) d z, \quad A_3 = \int F_2(z) g^2(z) d z.
\end{align*}

Now, let us discuss the properties of the variance estimates:

\begin{proposition}\label{Prop::var}
    Under the null hypothesis, the variance of $T_{n, b}$ is 
    \begin{align*}
       & \mathrm{var}(T_{n, b}) \\
       = & \frac{2b \PP{n-4}{2}}{ \PP{n}{4}} \big\{E[\zeta^2(\vesub{X}{1}, \vesub{X}{2})] -2 E\big[\zeta(\vesub{X}{1}, \vesub{X}{2})\zeta(\vesub{X}{1}, \vesub{X}{3}) \big] + \big[E(\zeta(\vesub{X}{1}, \vesub{X}{2}))\big]^2 \big\} \\
        &\cdot \big\{ A_3 E [\eta^2(\vesub{Y}{1}, \vesub{Y}{2})]  - 2A_2 b E\big[\eta(\vesub{Y}{1}, \vesub{Y}{2})\eta(\vesub{Y}{1}, \vesub{Y}{3}) \big] + A_1^2 b \big[ E(\eta(\vesub{Y}{1}, \vesub{Y}{2})) \big]^2  \big\}\\& + O_p\left( \frac{b^2}{n^2} \right).
    \end{align*}
   Moreover, assuming that $E \big[\zeta(\vesub{X}{1}, \vesub{X}{2}) \big]^{2 + \delta} < \infty$, $E \big[\eta(\vesub{Y}{1}, \vesub{Y}{2}) \big]^{2 + \delta} < \infty $  hold for $\delta > 0$, and $b \to 0$, $nb \to \infty$, $S_{n, b}^2$ in Equation (\ref{varest}) is an estimate of $\mathrm{var}(T_{n, b})$.   
\end{proposition}

Under the null hypothesis, the variance of $T_{n, b}$ is given by a complex expression. However, a key insight is that the variance exhibits the first order of $b,$ regardless of specific choices for $Z$ or $\kappa_b$. This observation signifies the role of $Z$ and $\kappa_b$ in achieving the asymptotic normality of the proposed test statistic. Additionally, $S_{n, b}^2$ can consistently estimate the variance of $T_{n, b}$.

We will now analyze the properties of the statistic $T_{n, b}$ and utilize U-statistics theory as our primary tool in this analysis.  The sample version of the random-lifter dependence measure corresponds to a U-statistic with the following kernel function:
\begin{align*}
    &h_b(\vesub{W}{1}, \vesub{W}{2}, \vesub{W}{3}, \vesub{W}{4}) \\
    = & \frac{1}{4} \big[\zeta(\vesub{X}{1}, \vesub{X}{2}) + \zeta(\vesub{X}{3}, \vesub{X}{4}) - \zeta(\vesub{X}{1}, \vesub{X}{3})- \zeta(\vesub{X}{2}, \vesub{X}{4})  \big]  \\
   & \cdot \big[\eta(\vesub{Y}{1}, \vesub{Y}{2})\kappa_b(Z_1, Z_2) + \eta(\vesub{Y}{3}, \vesub{Y}{4})\kappa_b(Z_3, Z_4)  - \eta(\vesub{Y}{1}, \vesub{Y}{3})\kappa_b(Z_1, Z_3)- \eta(\vesub{Y}{2}, \vesub{Y}{4})\kappa_b(Z_2, Z_4)  \big].
\end{align*}
The symmetric kernel is given by
\begin{align*}
    \Bar{h}_b (\vesub{W}{1}, \vesub{W}{2}, \vesub{W}{3}, \vesub{W}{4}) = & \frac{1}{3}\big[h(\vesub{W}{1}, \vesub{W}{2}, \vesub{W}{3}, \vesub{W}{4}) + h(\vesub{W}{3}, \vesub{W}{2}, \vesub{W}{1}, \vesub{W}{4}) \\
    &+ h(\vesub{W}{1}, \vesub{W}{2}, \vesub{W}{4}, \vesub{W}{3}) \big].
\end{align*}
With these kernels, we can express $T_{n, b}$ as: 
$$T_{n, b} = \frac{1}{\CC{n}{4}} \sum_{i < j < k < l} \Bar{h}_b (\vesub{W}{i}, \vesub{W}{j}, \vesub{W}{k}, \vesub{W}{l}), $$
where $\CC{n}{r}$ the number of $r$-combinations of an $n$-set, i.e., $\CC{n}{r} = \frac{n!}{(n-r)!r!}$.

In line with U-statistics theory, we find that $T_{n, b}$ serves as an unbiased estimator of $\cov^2_{\zeta, \eta, \kappa_b}(\ve{X}, \ve{Y})$. Importantly, due to the independence of $\ve{X}$ from $(\ve{Y}, Z)$, $T_{n, b}$ is a degenerate U-statistic, aligning with the conclusions drawn by most existing methods. However, in contrast to the majority of existing methods, the asymptotic distribution of $\rolin_{n,b}$ does not follow second-order Wiener chaos. In the Appendix, we demonstrate that by introducing the random-lifter, $\rolin_{n,b}$ can be expressed in terms of a martingale, enabling us to derive its asymptotic normality using the martingale central limit theorem. The key findings are summarized below:

\begin{theorem}[Null limit distribution]\label{Thm::randomH0}
 Under the null hypothesis, assuming that $E \big[\zeta(\vesub{X}{1}, \vesub{X}{2}) \big]^{2 + \delta} < \infty$, $E \big[\eta(\vesub{Y}{1}, \vesub{Y}{2}) \big]^{2 + \delta} < \infty $  hold for $\delta > 0$, and $b \to 0$, $nb \to \infty$, we have
  \begin{align*}
      \rolin_{n,b} \to N(0, 1).
  \end{align*}
\end{theorem}

The asymptotic normality of our random-lifter is a generalization of the martingale central limit theorem (MCLT). The theorem is stated as follows:
    Let $(R_i, \mathcal{F}_i)$ be a martingale difference sequence (MDS), if
    \begin{itemize}
        \item [(i)] $\sum_{i = 2}^n E[R_i^2 \mid \mathcal{F}_{i - 1}] \rightarrow \sigma^2$,
        \item [(ii)] $\forall \epsilon > 0$, $\frac{1}{\sigma^2}\sum_{i = 2}^n E[R_i^2 I\{R_i > \sigma\epsilon \}] \rightarrow 0$,
    \end{itemize}
    then $\sum_{ i = 1}^n R_i \rightarrow N(0, \sigma^2)$. 
     The first condition for applying the MCLT stipulates that the quadratic variation of the martingale must be bounded and converge to a specified constant. The second condition is analogous to the conditional Lindeberg condition.

    The primary limitation of the traditional HSIC method in achieving normality lies in its failure to satisfy the first condition of MCLT, which requires that the quadratic variation of the martingale converge to a certain constant. To address this issue, the random-lifter approach introduces a modification to the kernel based on a neighborhood of $Z$. This adjustment incorporates a random weight into the original kernel function.
    The random weight has the order $b$ and will approach $0$ when $n \rightarrow \infty$. This difference compresses the original kernel function. Consequently, under this modification, the quadratic variation of the martingale becomes bounded. Moreover, the condition $nb \rightarrow \infty$ ensures that the Lindeberg condition is still satisfied, even with the inclusion of the random-lifter term.

    In general, to reach asymptotic normality, we only need to impose weaker moment conditions on the kernel function $\zeta, \eta$ while ensuring that the order of the bandwidth $b$ falls between a constant and $\frac{1}{n}$.

To determine the convergence rate of $T_{n, b}/S_{n, b}$, we define the following functionals:
\begin{align*}
    g_x(\vesub{X}{1}, \vesub{X}{2}, \vesub{X}{3}, \vesub{X}{4}) &= \Tilde{\zeta}(\vesub{X}{1}, \vesub{X}{3})\Tilde{\zeta}(\vesub{X}{2}, \vesub{X}{3})\Tilde{\zeta}(\vesub{X}{1}, \vesub{X}{4})\Tilde{\zeta}(\vesub{X}{2}, \vesub{X}{4}), \\
    g_y(\vesub{Y}{1}, \vesub{Y}{2}, \vesub{Y}{3}, \vesub{Y}{4}) &= \eta(\vesub{Y}{1}, \vesub{Y}{3})\eta(\vesub{Y}{2}, \vesub{Y}{4})\eta(\vesub{Y}{1}, \vesub{Y}{4})\eta(\vesub{Y}{2}, \vesub{Y}{4}),
\end{align*}
 where 
 $$\Tilde{\zeta}(\vesub{X}{1}, \vesub{X}{2}) = \zeta(\vesub{X}{1}, \vesub{X}{2}) - E\big[\zeta(\vesub{X}{1}, \vesub{X}{3}) \mid \vesub{X}{1}\big] - E\big[\zeta(\vesub{X}{2}, \vesub{X}{4}) \mid \vesub{X}{2}\big] + E\big[\zeta(\vesub{X}{3}, \vesub{X}{4})\big].$$
  In Theorem~\ref{Thm::randomH0}, we treat $T_{n, b}/S_{n, b}$ as a martingale, then we can use Berry-Esseen bound to find an upper bound of $\sup _{t \in \mathbf{R}}  \left|\pr\left(T_{n, b}/S_{n, b}\le t\right)-\Phi(t)\right|$.
\begin{theorem}\label{random:rate} If $E\left[\zeta(\vesub{X}{1}, \vesub{X}{2}) \right]^{2 + 2\delta} < \infty$, $ E\left[\eta(\vesub{Y}{1}, \vesub{Y}{2}) \right]^{2 + 2\delta} < \infty$ for some $ \delta \in (0,1)$, $b \rightarrow 0$, $nb \rightarrow \infty$, then, under the null hypothesis,
\begin{align*}
      &\sup _{t \in \mathbf{R}}  \left|\pr\left(\rolin_{n,b}\le t\right)-\Phi(t)\right| \\ \lesssim & \Big\{\frac{1}{(nb)^{\delta}} \frac{E [\Tilde{\zeta}(\vesub{X}{1}, \vesub{X}{2}) \eta(\vesub{Y}{1}, \vesub{Y}{2})]^{2+2\delta} }{ \big[ E [\Tilde{\zeta}(\vesub{X}{1}, \vesub{X}{2}) \eta(\vesub{Y}{1}, \vesub{Y}{2})]^{2} \big]^{1+\delta}}\\
      & +b^{\frac{1+\delta}{2}} \frac{\left\{E \left[g_x(\vesub{X}{1}, \vesub{X}{2}, \vesub{X}{3}, \vesub{X}{4}) g_y(\vesub{Y}{1}, \vesub{Y}{2}, \vesub{Y}{3}, \vesub{Y}{4})\right]  \right\} ^{\frac{1+\delta}{2}}}{\big[ E [\Tilde{\zeta}(\vesub{X}{1}, \vesub{X}{2}) \eta(\vesub{Y}{1}, \vesub{Y}{2})]^{2} \big]^{1+\delta}}  \Big\}^{\frac{1}{3+2\delta}}.
  \end{align*}
\end{theorem}
The results in Theorem~\ref{random:rate} give a non-asymptotic Berry-Esseen bound of our statistic.
Unlike the high-dimensional setting mentioned in \citep{Gao2021}, we can treat the three terms in the bound 
\begin{align*}
    E [\Tilde{\zeta}(\vesub{X}{1}, \vesub{X}{2}) \eta(\vesub{Y}{1}, \vesub{Y}{2})]^{2+2\delta}, \big[ E [\Tilde{\zeta}(\vesub{X}{1}, \vesub{X}{2}) \eta(\vesub{Y}{1}, \vesub{Y}{2})]^{2} \big]^{1+\delta}, \\
    \left\{E \left[g_x(\vesub{X}{1}, \vesub{X}{2}, \vesub{X}{3}, \vesub{X}{4}) g_y(\vesub{Y}{1}, \vesub{Y}{2}, \vesub{Y}{3}, \vesub{Y}{4})\right]  \right\} ^{\frac{1+\delta}{2}}
\end{align*}
as constants or $O(1)$. Then the convergence rate of normal approximation is determined by $(nb)^{-\delta}, b^{(1 + \delta)/2}$. The bandwidth $b$ should satisfy $nb \rightarrow \infty$, which can guarantee $(nb)^{-\delta}, b^{(1 + \delta)/2}$ all tend to zero and ensure that the normal approximation is valid. Let $(nb)^{-\delta} = b^{(1 + \delta)/2}$ and optimize with $\delta$. We can obtain the optimal bandwidth as $b = n^{-1/2}$.

We then delve into the asymptotic properties of the test statistic under the alternative hypothesis. Drawing on the H-decomposition of the U-statistic, we reveal the asymptotic distribution of the proposed test statistic, as summarized below:

\begin{theorem}[Asymptotic distribution under alternative hypothesis]\label{random:H1}
  Under the alternative hypothesis, we have
  \begin{align*}
      \frac{\sqrt{n}(T_{n, b} - T_{\zeta, \eta, \kappa_b}(\ve{X}, \ve{Y}))}{4 \sigma_{b1}} \to N(0, 1),
  \end{align*}
  where $\sigma_{b1}^2 = \mathrm{var}(E\big[\Bar{h}_b(\vesub{W}{1}, \vesub{W}{2}, \vesub{W}{3}, \vesub{W}{4})|\vesub{W}{1} \big])$.

\end{theorem}

Under the assumption of employing the same kernels $\zeta$ and $\eta$ for both HSIC and our proposed method, we denote the respective powers of HSIC and our method as $K_{HSIC}$ and $K_{ROLIN}$. We can get the following result about the test powers.
\begin{theorem}[Power functions]\label{random:H12}
    Under the alternative hypothesis, we have
    \begin{align*}
        & K_{ROLIN} =  \Phi\left(\frac{ \sqrt{\tilde{H}_1 + \tilde{H}_2}  }{\sqrt{\frac{A_2}{A_1^2}  \cdot \tilde{H}_1 +   \tilde{H}_2}} \Phi^{-1}(K_{HSIC}) + O_p(\frac{1}{\sqrt{nb}}) \right),
    \end{align*}
    where $\tilde{H}_1$ and $\tilde{H}_2$ are some quantities that we omit the exact form here due to their complex expressions, but we will provide them in the Appendix.
\end{theorem}

\begin{remark}
    The constant $A_2 / A_1^2$ reflects the impact of different choices of the kernel $\kappa_b$ or $Z$ on the random-lifter method. The key properties of $A_1, A_2, A_3$ are described in LEMMA 1.1. In order to give a more in-depth understanding of these constants, we consider the case where the domain of $Z$ is $\mathbf{R}$, as we mentioned in Remark 1. For example, we can choose $Z$ to be a Gaussian distribution. Then, we know from Assumptions (ii) and (iii) of $\kappa_b$ that the choice of $\kappa_b$ does not affect the value
    \begin{align*}
        A_1 = \int g^2(z) dz = \frac{1}{2\sqrt{\pi}} \approx 0.2821, \quad A_2 = \int g^3(z) dz  = \frac{1}{2\sqrt{3}\pi} \approx 0.0919.
    \end{align*}
 Now, we can obtain $A_2 / A_1^2 = \frac{2}{\sqrt{3}} \approx 1.1547$. If we choose the t-distribution with degree $3$ as an example, the value of the constant $A_2 / A_1^2$ is $1.2624$. We see that a normal distribution is a better choice than a t-distribution.

    If $Z$ is a bounded random variable, the computation of $A_1, A_2$ becomes more complicated. 
    For example, we choose $Z$ as a beta distribution with shape parameters $\alpha = 2, \beta = 1$ and $\kappa_b$ as the Laplace kernel. $A_2/A_1^2$ ranges from 1.08 to 1.09, which is better than the above situation. For the simulation studies in our paper, we also choose the beta distribution and Laplace kernel.
    
\end{remark}

Moving on to the alternative hypothesis related to the sample size $n$, we determine the minimax rate of the proposed test method. We begin by introducing assumptions concerning the kernels:

\begin{assumption}\label{ass1}
 $\zeta$, $\eta$ are shift-invariant, i.e.  $\zeta(\vesub{X}{1}, \vesub{X}{2}) = \zeta(\vesub{X}{1} - \vesub{X}{2})$, $\eta(\vesub{Y}{1}, \vesub{Y}{2}) = \eta(\vesub{Y}{1} - \vesub{Y}{2})$, and $\int \zeta\left(\vesub{X}{1}, \vesub{X}{2}\right) d\vesub{X}{1} = 1, \int \eta\left(\vesub{Y}{1}, \vesub{Y}{2}\right) d \vesub{Y}{1} = 1$.
\end{assumption}

With this assumption, we establish the following result:

\begin{theorem}[Minimax Rate I]\label{thm::minimax1}
  Under the alternative hypothesis and assumption \ref{ass1}, for $0 < \beta < 1$, there exists a constant $B(\beta)$ depending only on $\beta$, such that if the condition
$$
T_{\zeta, \eta, \kappa_b}(\ve{X}, \ve{Y}) \geq B(\beta)  \left(\frac{b}{\sqrt{n}} + \frac{\sqrt{b\left\|\zeta\right\|_{\infty}\left\|\eta\right\|_{\infty}}}{n}\right)
$$
holds for sufficiently large $n$, then
$\pr_{H_1}\left(\rolin_{n,b} \leq \Phi(\alpha)\right) \leq \beta$.
\end{theorem}

In this inequality, $T_{\zeta, \eta, \kappa_b}(\ve{X}, \ve{Y})$ is equipped with the random-lifter kernel $\kappa_b$, so its order is at least $O(b)$. 
If we divide $b$ on both sides, the right hand side in Theorem~\ref{thm::minimax1} becomes 
\begin{align*}
    B(\beta)  \left(\frac{1}{\sqrt{n}} + \frac{\sqrt{\left\|\zeta\right\|_{\infty}\left\|\eta\right\|_{\infty}}}{n\sqrt{b}}\right).
\end{align*}

The first term is consistent with the result of HSIC \cite{Melisande2022}. The second term involves an additional $b^{-1/2}$, reflecting the influence of the random lifter. However, explicitly determining the order of $T_{\zeta, \eta, \kappa_b}(\ve{X}, \ve{Y})$ in practical implementations can be challenging.

To obtain a more general upper bound of the separation rate, we assume the joint probability density function (PDF) of $\ve{X}, \ve{Y}$ exists, denoted as $f(\ve{X}, \ve{Y})$ with marginal PDFs $f_x(\ve{X})$ and $f_y(\ve{Y})$, respectively. $\|f\|_{\infty},\left\|f_x\right\|_{\infty},\left\|f_y\right\|_{\infty}$ are all finite. We further consider the $L_2$ norm of the difference as $\psi(x, y) = f(x, y) - f_x(x)f_y(y)$. 

To obtain a more general upper bound on the separation rate, we assume the joint probability density function (PDF) of $\ve{X}, \ve{Y}$ exists, denoted as $f(\ve{X}, \ve{Y})$ with marginal PDFs $f_x(\ve{X})$, $f_y(\ve{Y})$ respectively. These PDFs satisfy $\|f\|_{\infty},\left\|f_x\right\|_{\infty}$, and$\left\|f_y\right\|_{\infty}$ , all finite. Additionally, we consider the $L_2$ norm of the difference, defined as $\psi(x, y) = f(x, y) - f_x(x)f_y(y)$.

Based on these assumptions, we derive the following theorem, which provides a sufficient condition on $\|\psi(x, y)\|_2$ for bounding the Type II error of the test by $\beta$:

\begin{theorem}[Minimax Rate II]\label{thm::minimax2}
  Under the alternative hypothesis and assumption \ref{ass1}, for $0 < \beta < 1$, there exists a constant $B(\beta) > 0$ depending only on $\beta$, such that if the condition
    $$\|\psi\|_2^2 \geq \|\psi - \psi*\left(\zeta \otimes \eta\right)\|_2^2 + \frac{B(\beta)}{n\sqrt{b}} \sqrt{\|\zeta \|_{\infty} \|\eta\|_{\infty}}$$
is satisfied for sufficiently large $n$, then
$\pr_{H_1}\left(\rolin_{n,b} \leq \Phi(\alpha)\right) \leq \beta$.
\end{theorem}

In Theorem \ref{thm::minimax2}, the right-hand side of the inequality contains a bias term $\|\psi - \psi*\left(\zeta \otimes \eta\right)\|_2^2$, which arises from the variance of $T_{n, b}$. This bias term results from using two different criteria, the $L_2$ distance and the random-lifter dependence measure, to evaluate the disparity between the joint PDF and the product of the marginal PDFs. The presence of this bias term is due to the use of the random lifter, which introduces a trade-off between the deviation and the variance.

The crucial point is to verify the following conditions:
$$E(T_{n, b}) > \sqrt{\frac{\var(T_{n, b})}{\beta}} + S_{n, b} \Phi(\alpha).$$
Unlike the conclusion in \cite{Melisande2022,kim2022}, we do not need to consider the threshold of permutation since the threshold we use is just the quantile of the standard normal distribution.
    
\begin{remark}
    We introduce specific definitions for the functions $\zeta$ and $\eta$, which are central to our analysis of the optimal minimax rate. These functions are defined as follows:

    For $\vesub{X}{i} = (x_i^{(1)}, \dots, x_i^{(p)}) \in \mathbf{R}^p$ and $\vesub{Y}{i} = (y_i^{(1)}, \dots, y_i^{(p)}) \in \mathbf{R}^q$, where $i = 1, 2$, we define:
    
    \begin{align*}
        & \zeta(\vesub{X}{1}, \vesub{X}{2})=\frac{1}{(2\pi)^{p/2}\theta^x_1 \ldots \theta^x_p} \exp\left( -\frac{1}{2} \sum_{i = 1}^p \left(\frac{x_1^{(i)} - x_2^{(i)}}{\theta^x_i}\right)\right), \\
        & \eta(\vesub{Y}{1}, \vesub{Y}{2})=\frac{1}{(2\pi)^{q/2}\theta^y_1 \ldots \theta^y_q} \exp\left( -\frac{1}{2} \sum_{i = 1}^q \left(\frac{y_1^{(i)} - y_2^{(i)}}{\theta^y_i}\right)\right).
    \end{align*}
 It is important to note that the success of our analysis hinges on certain conditions being met. Specifically, we assume that:
    $$\max \left\{\prod_{i=1}^p \theta_i^x, \prod_{j=1}^q \theta_j^y \right\}<1, \quad  n \sqrt{\theta^x_1 \ldots \theta^x_p \theta^y_1 \ldots \theta^y_q}>1.$$
    Based on the results of Gaussian MMD \citep{kim2022} and HSIC \citep{Melisande2022}, the condition in Theorem~\ref{thm::minimax2} can be modified as follows.
With these assumptions in place and using the results of Gaussian MMD \citep{kim2022} and HSIC \citep{Melisande2022}, we can adapt the condition in Theorem~\ref{thm::minimax2} as:
    
    $$\|\psi\|_2^2 \geq \|\psi - \psi*\left(\zeta \otimes \eta\right)\|_2^2  + \frac{B_{\beta}}{n\sqrt{b \theta^x_1 \ldots \theta^x_p \theta^y_1 \ldots \theta^y_q}}.$$

   Moreover, by imposing restrictions on the function $\psi$ within specific spaces, we  can control the bias term $\|\psi - \psi*\left(\zeta \otimes \eta\right)\|_2^2$. For instance, if we consider the Sobolev ball with a regularity parameter $\delta$ in the range $(0, 2)$ and a radius $R > 0$, defined as:
    $$\mathcal{S}_d^\delta(R) = \left\{s: \mathbf{R}^d \rightarrow \mathbf{R}; s \in \mathbf{L}_1(\mathbf{R}^d) \cap \mathbf{L}_2(\mathbf{R}^d), \int_{\mathbf{R}^d} \| u \|^{2\delta} |\hat{s}(u)|^2 du \leq (2 \pi)^d R^2 \right\},$$
    where $\| \cdot \|$  represents the Euclidean norm in $\mathbf{R}^d$, and $\hat{s}$ is the Fourier transform of $s$, defined on $\mathbf{R}^d$ by $\hat{s}(u) = \int_{\mathbf{R}^d} s(x) \exp(i x^{\top}u) dx$. Based on Lemma 3 of \citet{Melisande2022}, we can derive an upper bound for the bias term:
    \begin{align*}
        \|\psi - \psi*\left(\zeta \otimes \eta\right)\|_2^2 \leq B(p, q, \delta, R) \left[\sum_{i = 1}^p (\theta_{i}^x)^{2 \delta} + \sum_{j = 1}^q (\theta_{j}^y)^{2 \delta}\right]
    \end{align*}
 with a constant $B(p, q, \delta, R)$. Let $b = n^{-\gamma}$, and by selecting optimal bandwidths for the Gaussian kernel, given by
    $\theta_i^x = \theta_j^y = n^{-\frac{2 - \gamma}{4 \delta + p + q}}$,
    we can express the condition in Theorem~\ref{thm::minimax2} as:
    $$\|\psi\|_2 \geq B(p, q, \delta, R, \alpha, \beta) n^{-\frac{2 \delta - \delta \gamma}{4 \delta + p + q}}.$$

  It is worth noting that the optimal minimax rate in the Sobolev ball is $n^{-2 \delta/(4 \delta + p + q)}$, a rate that has been shown to be attainable by HSIC in previous studies \cite{rigollet2007linear,kim2022,Melisande2022}. The random-lifter independence test, however, achieves a rate of $n^{-\delta \gamma/(4 \delta + p + q)}$, which introduces a gap of at most half of the optimal minimax rate.  This gap is a consequence of the techniques involved in employing the random-lifter and introducing the bandwidth parameter $b$, and it does not exceed half of the optimal minimax rate due to the choice of $\gamma$ in the range $(0, 1)$.
\end{remark}

\subsection{Random-lifter independence test algorithm}
In this subsection, we elucidate the practical computational steps involved in our proposed method. Given the observed realizations, we compute the positive kernel matrices associated with the samples $\ve{X}$ and the product of $\ve{Y}$ and $Z$, denoted by $\zeta_{n \times n}$ and $\varrho_{n \times n}$, respectively. %\RED{(The font for $X$ and $Y$ differs from the above.)}

The procedure of our algorithm is summarized as Algorithm~\ref{alg:cap}.
To obtain the numerator of the proposed test statistic, we provide a more efficient estimator via the U-centered centered function shown in Algorithm~\ref{alg:cap2}. The V-centered function in Algorithm~\ref{alg:cap3} is used to compute the variance term. 
The entire test procedure for the random-lifter method has a computational complexity of order $O(n^2)$. This implies that the method can perform the test using a standard normal quantile within the same complexity as one-step traditional permutation test methods, resulting in a significant improvement in computational efficiency.

\begin{algorithm}[h!]
\renewcommand{\algorithmicrequire}{\textbf{Input:}}
\renewcommand{\algorithmicensure}{\textbf{Output:}}
\footnotesize
\caption{\quad The random-lifter test statistic}
\label{alg:cap}
%\vspace*{-12pt}
\begin{algorithmic}[1]
    \REQUIRE  $\{(\vesub{X}{i}, \vesub{Y}{i}, Z_i)\}_{i = 1}^n$, positive-definite functions $\zeta, \eta$ and $\kappa$;
    \STATE  Set $Z \sim N(0,1)$;
    \FOR{$i, j = 1$ \TO $n$} 
        \STATE $\zeta[i, j] \gets \zeta(\mathbf{X}_i, \mathbf{X}_j)$;
        \STATE $\varrho[i, j] \gets \eta(\mathbf{Y}_i, \mathbf{Y}_j)\kappa_b(Z_i, Z_j)$;
    \ENDFOR
    \STATE $U_{\zeta} \gets Ucenter(\zeta)$;
    \STATE $U_{\varrho} \gets Ucenter(\varrho)$;
    \STATE $T_{n, b} \gets \frac{1}{n}\sum_{i \neq j} U_{\zeta}[i, j] * U_{\varrho}[i, j]$;
    \STATE $V_{\zeta} \gets Vcenter(\zeta)$
    \STATE  $V_{\varrho} \gets Vcenter(\varrho)$ 
     \STATE $S_{n,b}^2 \gets \frac{2(n-4)(n-5)}{n^2(n-1)^2(n-2)(n-3)} \sum_{i \neq j} (V_{\zeta}[i, j] * V_{\varrho}[i, j])^2$;
    \ENSURE $T_{n, b} / S_{n,b}$.
\end{algorithmic}
\end{algorithm}

%====================== U center function

\begin{algorithm}[h!]
\renewcommand{\algorithmicrequire}{\textbf{Input:}}
\renewcommand{\algorithmicensure}{\textbf{Output:}}
\footnotesize
\caption{\quad U-centered function}
\label{alg:cap2}
%\vspace*{-12pt}
\begin{algorithmic}[1]
    \REQUIRE  $K_{ij} = \zeta(\mathbf{X}_i, \mathbf{X}_j)$;
    \FOR{$i = 1$ \TO $n$} 
        \STATE $K_{i\cdot} \gets \frac{1}{n - 2}\sum_{k \neq i} K_{ik}$;
        \STATE $K_{\cdot i } \gets \frac{1}{n - 2}\sum_{l \neq i} K_{li}$;
    \ENDFOR
    \STATE $K_{\cdot\cdot} \gets \frac{1}{(n - 1)(n - 2)}\sum_{k \neq l} K_{kl}$;
    \ENSURE the U-centered matrix $(U_{ij}) = K_{ij} - K_{i\cdot} - K_{\cdot j } + K_{\cdot\cdot}$.
\end{algorithmic}
\end{algorithm}

%====================== V center function

\begin{algorithm}[h!]
\renewcommand{\algorithmicrequire}{\textbf{Input:}}
\renewcommand{\algorithmicensure}{\textbf{Output:}}
\footnotesize
\caption{\quad V-centered function}
\label{alg:cap3}
%\vspace*{-12pt}
\begin{algorithmic}[1]
    \REQUIRE  Matrix $H_{ij} = \zeta(\mathbf{X}_i, \mathbf{X}_j)$;
    \FOR{$i = 1$ \TO $n$} 
        \STATE $H_{i\cdot} \gets \frac{1}{n}\sum_{s = 1}^n H_{is}$;
        \STATE $H_{\cdot i} \gets \frac{1}{n}\sum_{s = 1}^n H_{si}$;
    \ENDFOR
    \STATE $H_{\cdot\cdot} \gets \frac{1}{n^2}\sum_{i, j} H_{ij}$;
    \ENSURE the V-centered matrix $(V_{ij}) = H_{ij} - H_{i\cdot} - H_{\cdot j } + H_{\cdot\cdot}$.
\end{algorithmic}
\end{algorithm}

\section{Numerical Studies}

In this section, we will present the results of simulation studies and real-data analysis to validate the effectiveness of the Random-Lifter method.

\subsection{Bandwidth selection}
The bandwidth parameter $b$ is crucial in the random-lifter method, making the selection of an appropriate bandwidth a critical issue. According to Theorem~\ref{Thm::randomH0}, the optimal bandwidth order is $1/\sqrt{n}$. However, as Theorems~\ref{random:H12} and \ref{thm::minimax2} suggest, a bandwidth closer to the order of 1 enhances both the power and the separation rate. This presents a trade-off in the choice of bandwidth between achieving a normal approximation and maximizing power.

In practical applications, we observe that while a normal approximation with $b=1/\sqrt{n}$ is satisfactory, the resultant power is often suboptimal. 
We have improved the findings from Theorem~\ref{Thm::randomH0} in light of recent developments in high-dimensional studies \citep{gao2023two, Gao2021, jiang2023}. Our findings suggest an optimal bandwidth selection of $b = \sqrt{pq/n}$. This formula is particularly effective in scenarios characterized by high dimensions and weak serial correlations, such as those involving $m$-dependent structures. Additionally, we have observed that this bandwidth performs well even in low-dimensional settings, where the factor $\sqrt{pq}$ serves as a constant adjustment to the conventional $1/\sqrt{n}$.

As a practical matter, we suggest using $b = \sqrt{pq/n}$ or empirical selection methods like Scott's rule \citep{scott1979optimal} and Silverman's rule \citep{silverman2018density}, which usually suggest a bandwidth order of $n^{-0.2}$. These approaches balance the trade-offs between accuracy in normal approximation and power effectively.

\subsection{Simulation Studies}

We present the numerical performance of the proposed random-lifter method (\textbf{Rolin}) for linear and nonlinear cases and compare it with two other tests, including Hilbert Schmidt Independence Criterion and its gamma approximation\cite{gretton2005kernel} (\textbf{HSIC} and \textbf{HSIC.gamma}), respectively, and
cross HSIC \cite{shekhar2023permutation}(\textbf{Cross}).
Rolin employs a Gaussian kernel for the two random objects under consideration, with the tuning parameter set to the lower quantile but taking one-half of the median if it is 0. The bandwidth of the random-lifter kernel is determined as $\sqrt{pq/n}$, where $n$ represents the sample size, and $p$ and $q$ denote the dimension of $\ve{X}$ and $\ve{Y}$, respectively.
The P-value is calculated using the Z-score table of the standard normal distribution. For the remaining two tests, the default settings utilize a Gaussian kernel with a bandwidth determined by the median heuristic.  The P-value for Cross is also determined from the Z-score table, while a permutation procedure calculates the P-value for the HSIC, and the number of replications is set to $B = 399$. 

In the first 12 cases, we fix the dimension at $p=5$, and for the remaining cases, we set $p=10$ to explore if varying dimensions affect the power performance. Each dimension is independently and identically generated according to the specified model. The initial sample size is set to $n = 50$ and is doubled each time until the sample size reaches 400. We conducted 1000 repetitions for each experiment, resulting in a total of 16 cases studied. By comparing these 16 cases, we evaluate the Type I error and power of the proposed method. Specifically, the first four cases are utilized to assess the Type I error, while the remaining cases are dedicated to comparing the power performance.

\begin{enumerate}
     \item[\BLUE{Case 1:}] Standard normal distribution, i.e. $X, Y \sim N(0, 1)$.
     \item[\BLUE{Case 2:}] (Conditional) standard normal distribution, i.e. $X, Y_1 \sim N(0, 1)$, and $\epsilon$ follows a uniform distribution $U[0, 1]$. Then
    $$Y = Y_1 \epsilon.$$
     \item[\BLUE{Case 3:}] Gamma distribution for which both the location and scale parameters are set to be $1$, i.e., $X, Y \sim \Gamma(1, 1)$.  
     
    \item[\BLUE{Case 4:}] (Conditional) Gamma distribution, i.e. $X, Y_1 \sim \Gamma(1, 1)$, and $\epsilon \sim N(0, 1)$. Then
    $$Y = Y_1 \epsilon.$$
\end{enumerate}

For the power analysis, we investigated 12 cases. In the first four models, we consider a linear association. In the remaining models, we explore various forms of nonlinear associations, including logarithmic, reciprocal, conditional linear, and polynomial relationships.

\begin{enumerate}

\item[\BLUE{Case 5:}] Consider that $X, \epsilon \sim U[0, 1]$, we set
   $$Y = 0.2 X + 0.8 \epsilon.$$
   
\item[\BLUE{Case 6:}] Consider that $X, \epsilon \sim \mathrm{Weibull}(1, 1)$. In this case, the Weibull distribution has both the shape and scale parameters set at $1$. We obtain
   $$Y = 0.2 X + 0.8 \epsilon.$$
   
\item[\BLUE{Case 7:}] Consider that $\log(X), \log(\epsilon) \sim N(0, 1)$. Let $a = (1, 0, \cdots, 0)^{\top}$ and $b = (0, 1, \cdots, 1)^{\top}$. We set
   $$Y = aX + b\epsilon.$$

\item[\BLUE{Case 8:}] Consider that $X, \epsilon \sim \mathrm{Weibull}(0.8, 1)$, where the shape and the scale parameter are set to be $0.8$ and $1$, respectively.  We set
   $$Y = aX + b\epsilon.$$
Here $a$ and $b$ are defined as above.

\item[\BLUE{Case 9:}] Consider two independent variables, $X,\epsilon \sim N(0,1)$.  We set
   $$Y = \log{(1+X^2)} + \epsilon.$$
    
\item[\BLUE{Case 10:}] Here $X \sim t_2$ and $\epsilon \sim \mathrm{Weibull}(5/3, 1)$ are independent, and
   $$Y = \frac{\epsilon}{1 + X^2}.$$

\item[\BLUE{Case 11:}] Here $X \sim N(0,1)$ and $\epsilon \sim t_{2}$ are independent, and we take
$$Y = X\epsilon + \epsilon.$$

\item[\BLUE{Case 12:}] Consider that $X \sim N(0,1)$ and $\epsilon \sim t_2$ are independent. Then
$$Y = X^2 + \epsilon.$$

\item[\BLUE{Case 13:}] Consider two independent variables: $X \sim U[0,1]$ and $\epsilon \sim t_2$.  We set
   $$Y = \log{(1+X^2)} + \epsilon.$$

\item[\BLUE{Case 14:}] Consider two independent variables: $X \sim N(0,1)$ and $\epsilon \sim \chi^2_{5/3}$. We take
   $$Y = \frac{\epsilon}{1 + X^2}.$$

\item[\BLUE{Case 15:}] Here $X \sim \exp{(1)}$ and $\epsilon \sim t_{8/3}$ are independent, and we take
$$Y = X\epsilon + \epsilon.$$

\item[\BLUE{Case 16:}] Consider that $X \sim U[0,1]$ and $\epsilon \sim N(0,1)$ are independent. Then
$$Y = X^2 + \epsilon.$$

\end{enumerate}

\begin{figure}[h!]
    \centerline{\includegraphics[width = 1.0\textwidth, height = 0.59\textwidth]{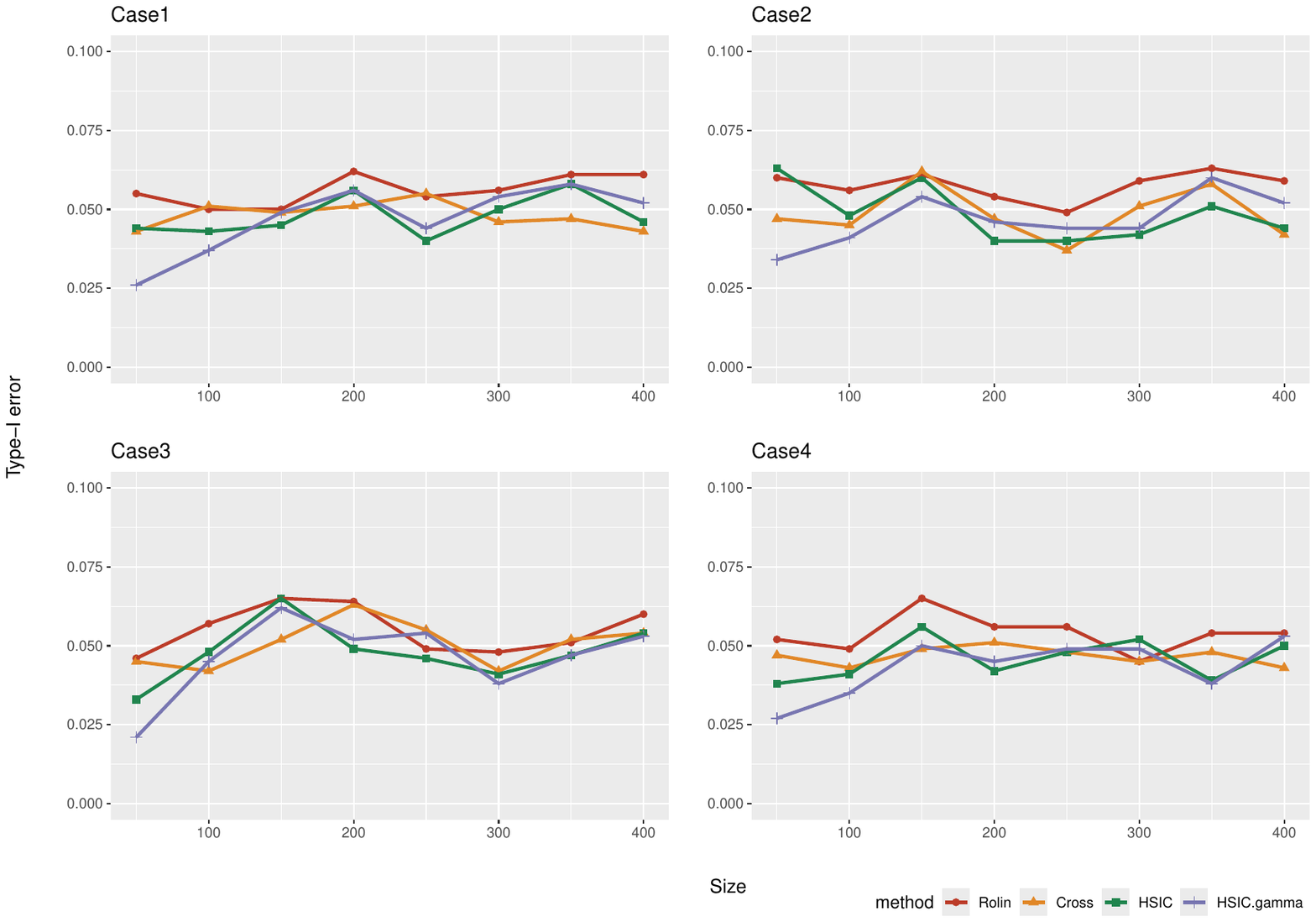}}
 %{pic/beta_case1-4.pdf}}
    \caption{Type I error in Models 1-4.}
    \label{figp4.1}
\end{figure}

\begin{figure}[h!]
    \centerline{\includegraphics[width = 1.0\textwidth, height = 0.6\textwidth]{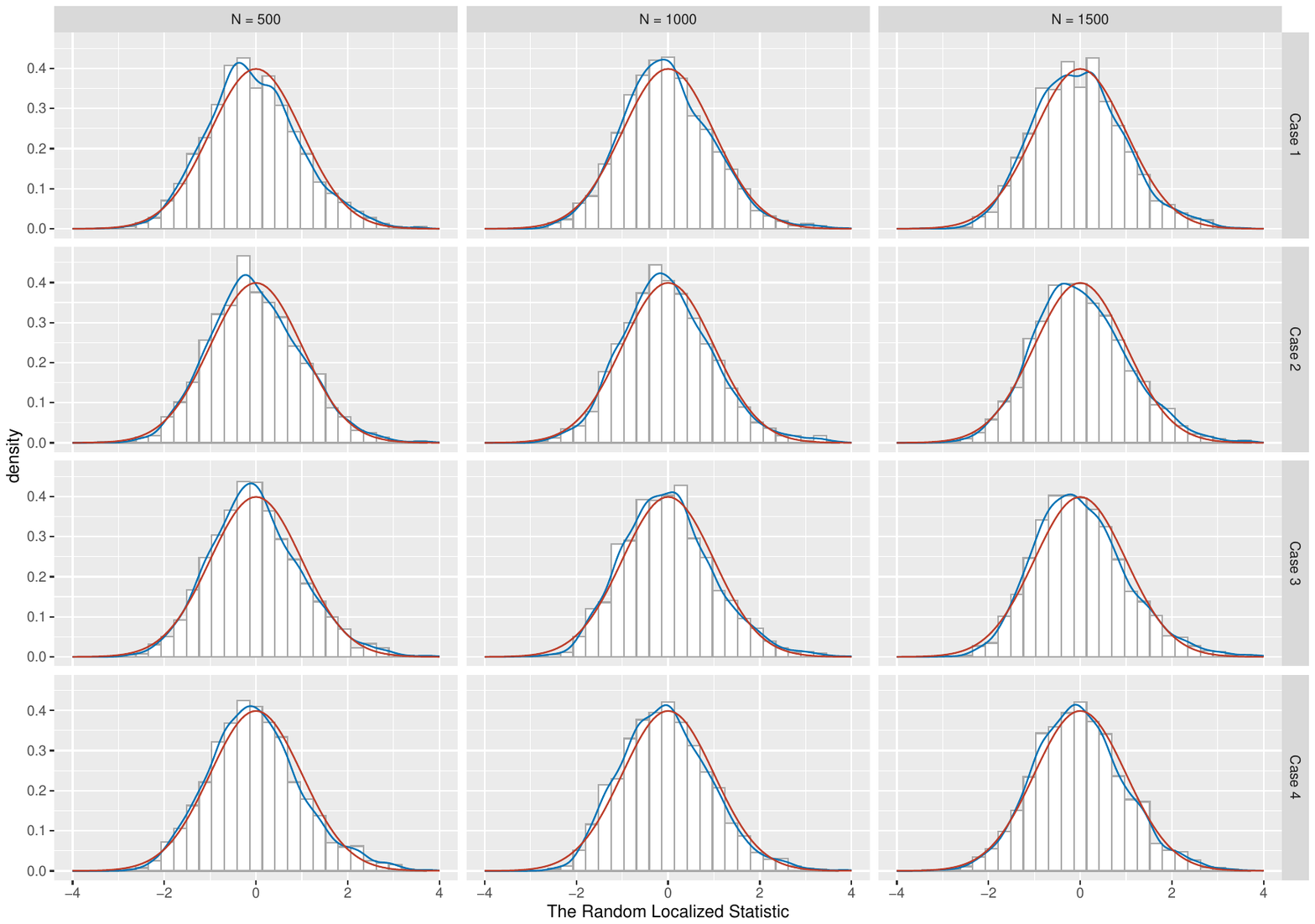}}
 %{pic/null_beta.pdf}}
    \caption{The histogram and density plot in Models 1-4, where the red curve represents the density of the standard normal distribution, and the blue curve is derived from the proposed method.}
    \label{figp4.5}
\end{figure}

Figure~\ref{figp4.1} presents the Type-I error in the aforementioned first four cases. As can be seen, the proposed method and the other two tests are well-controlled for Type-I errors.
We also depict the histogram and density plot with size $n = 500, 1000, 1500$ in Figure~\ref{figp4.5}. It can be seen that the density profiles of the proposed method are close to the standard normal distribution in all the cases considered. It is worth mentioning that at $n$ of 500, the density already has a good normal approximation performance.

\begin{figure}[h!]
    \centerline{\includegraphics[width = 1.0\textwidth, height = 0.61\textwidth]{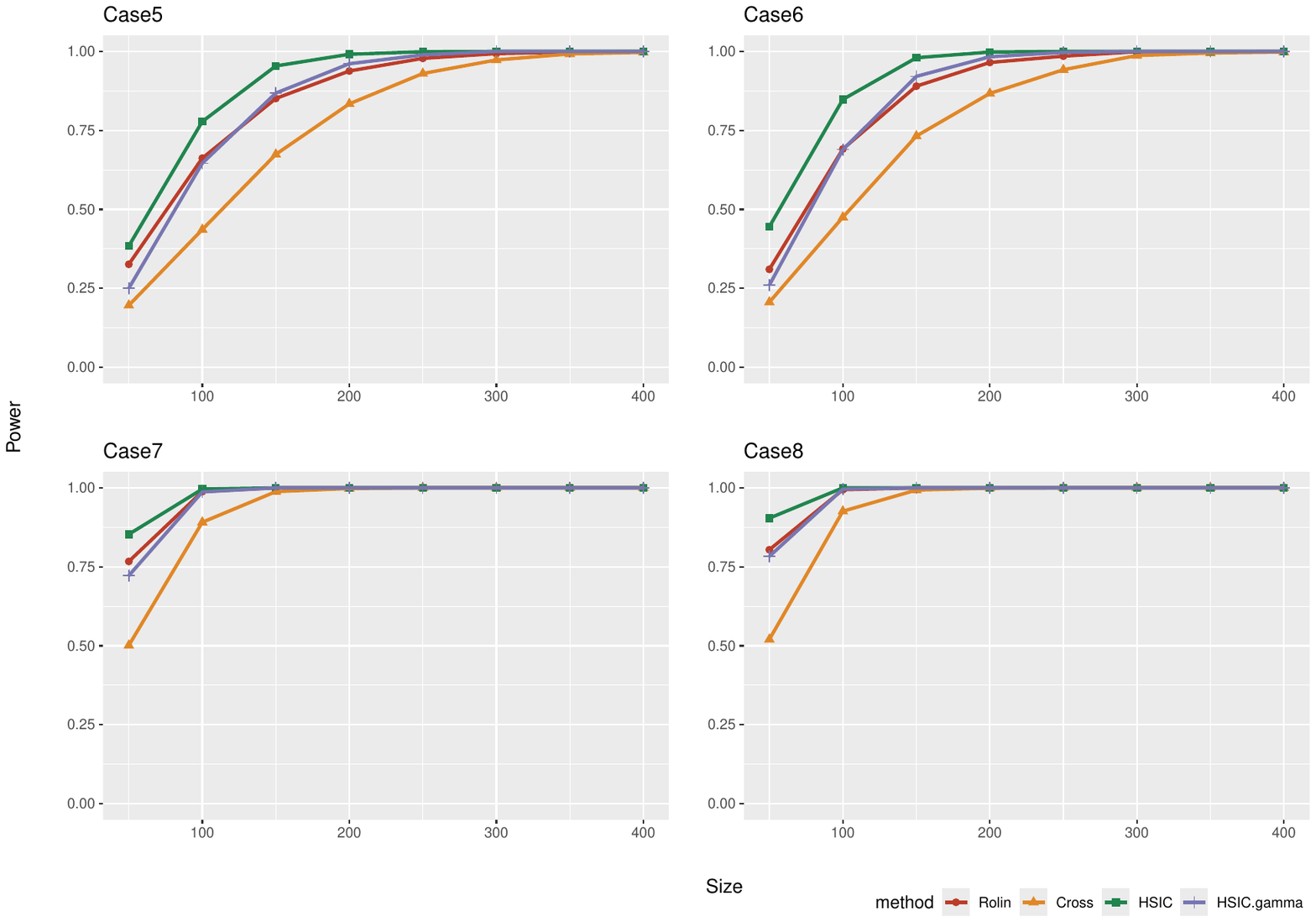}}
 %{pic/beta_case5-8.pdf}}
    \caption{The power performance in the linear case (Models 5-8).}
    \label{figp4.2}
\end{figure}

Figure~\ref{figp4.2} compares the power of the proposed method with the other two tests under the linear case.
From the figure depicted herein, it becomes apparent that, for a small sample size, our proposed test exhibits a somewhat reduced power compared to HISC. This outcome arises due to adopting a random-lifter method, which effectively curtails sample utilization. 
For the gamma approximation, it exhibits competitive performance under large sample sizes. However, with small sample sizes, its power slightly lags behind the proposed method. This discrepancy arises because when the dimension is large relative to the sample size, the approximation becomes less accurate, leading to inconsistent power results.
Nonetheless, in comparison to an alternative regularization approach (\textbf{Cross}), our method manifests superior power characteristics. Notably, our test demonstrates commendable power as the sample size is augmented to 200.

\begin{figure}[h!]
    \centerline{\includegraphics[width = 1.0\textwidth, height = 0.59\textwidth]{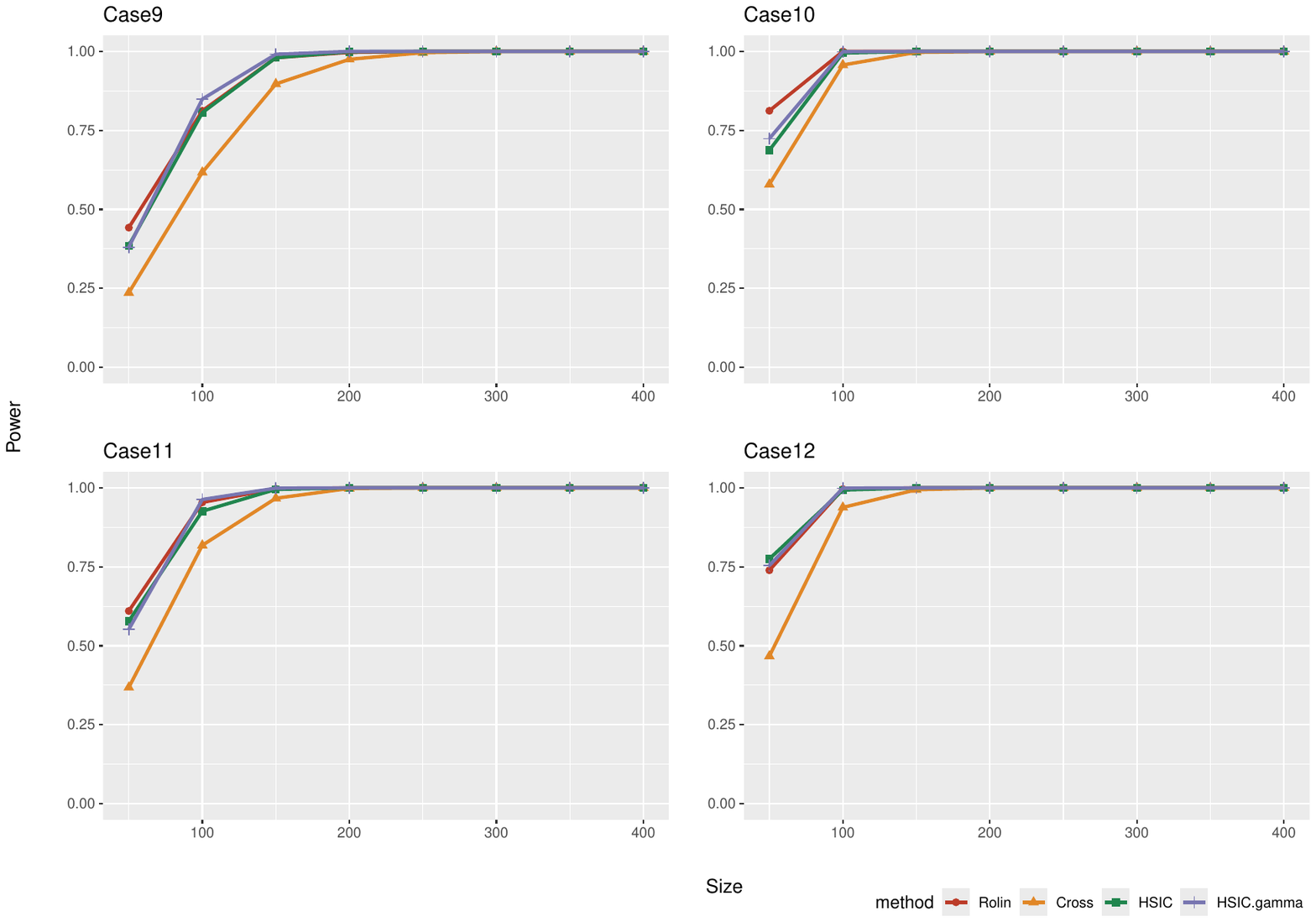}}
 %{pic/beta_case9-12.pdf}}
    \caption{The power performance in the nonlinear case (Models 9-12).}
    \label{figp4.3}
\end{figure}

Cases 9 to 16 feature four nonlinear scenarios under two different dimension settings. Figure~\ref{figp4.3} and \ref{figp4.4} summarize the results. Specifically, our method exhibits a slightly lower or higher power than HSIC when the small sample sizes are small. However, as the sample size slightly increases (n = 200), the power of our proposed approach becomes comparable to that of HSIC. In contrast, when compared to the alternative method (\textbf{Cross}), our method consistently demonstrates favorable performance characteristics. 
Intriguingly, when the first four cases where $p=5$ and the subsequent four cases where $p=10$ are examined, the performance gap between our proposed method and the gamma-approximation method widens in the latter set, especially when the sample size is small. This observation indicates that the gamma-approximation method may experience significant power loss as dimensionality increases.

\begin{figure}[h!]
    \centerline{\includegraphics[width = 1.0\textwidth, height = 0.59\textwidth]{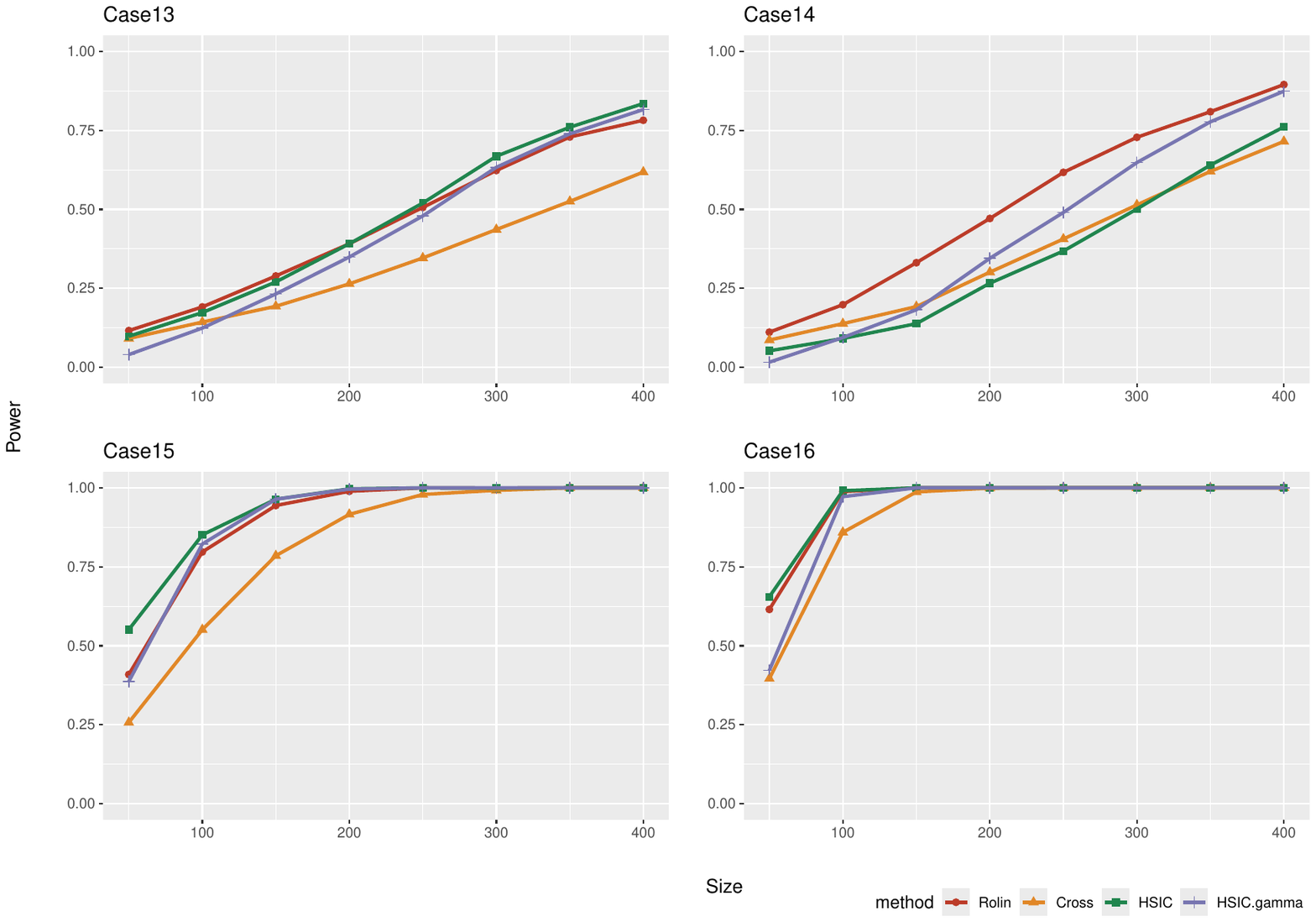}}
 %{pic/beta_case13-16.pdf}}
    \caption{The power performance in the nonlinear case (Models 13-16).}
    \label{figp4.4}
\end{figure}

For Cases 4-16, we can see that in most cases, the power of Rolin is slightly worse than the original HSIC, because Rolin adds weight to one kernel matrix, and the weight becomes very small for points with longer distances between $Z$, resulting in a decrease in sample utilization. Overall, the empirical results confirm that Rolin has asymptotic normality, its Type I error is under control,  and its power is comparable to or better than the competitors that require much more intensive computation.

To assess the runtime of the proposed method, we consider when every component of $X, Y$ is standard normal with $p = 5$, and the sample size varies as $n \in \{ 200, 400, 600, 800,$ $1000, 1200, 1400, 1600, 1800, 2000 \}$. The results are summarised in Figures ~\ref{figp4.6} and \ref{figp4.6.1}.

\begin{figure}[htb]
    \centerline{\includegraphics[width = 1.0\textwidth, height = 0.52\textwidth]{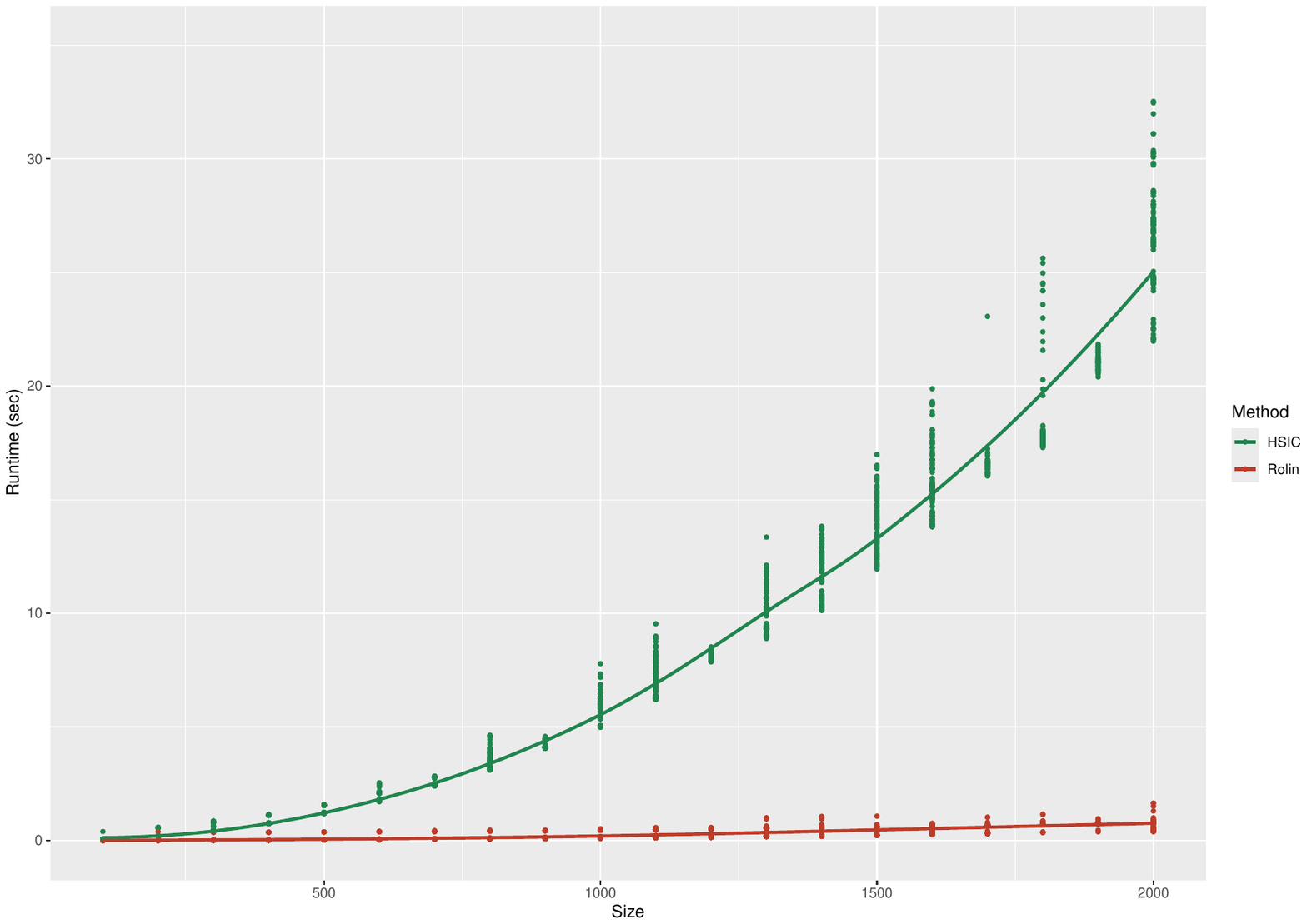}}
 %{pic/RH_compare.pdf}}
    \caption{The runtime performance of Rolin and HSIC.}
    \label{figp4.6}
\end{figure}

\begin{figure}[htb]
    \centerline{\includegraphics[width = 1.0\textwidth, height = 0.52\textwidth]{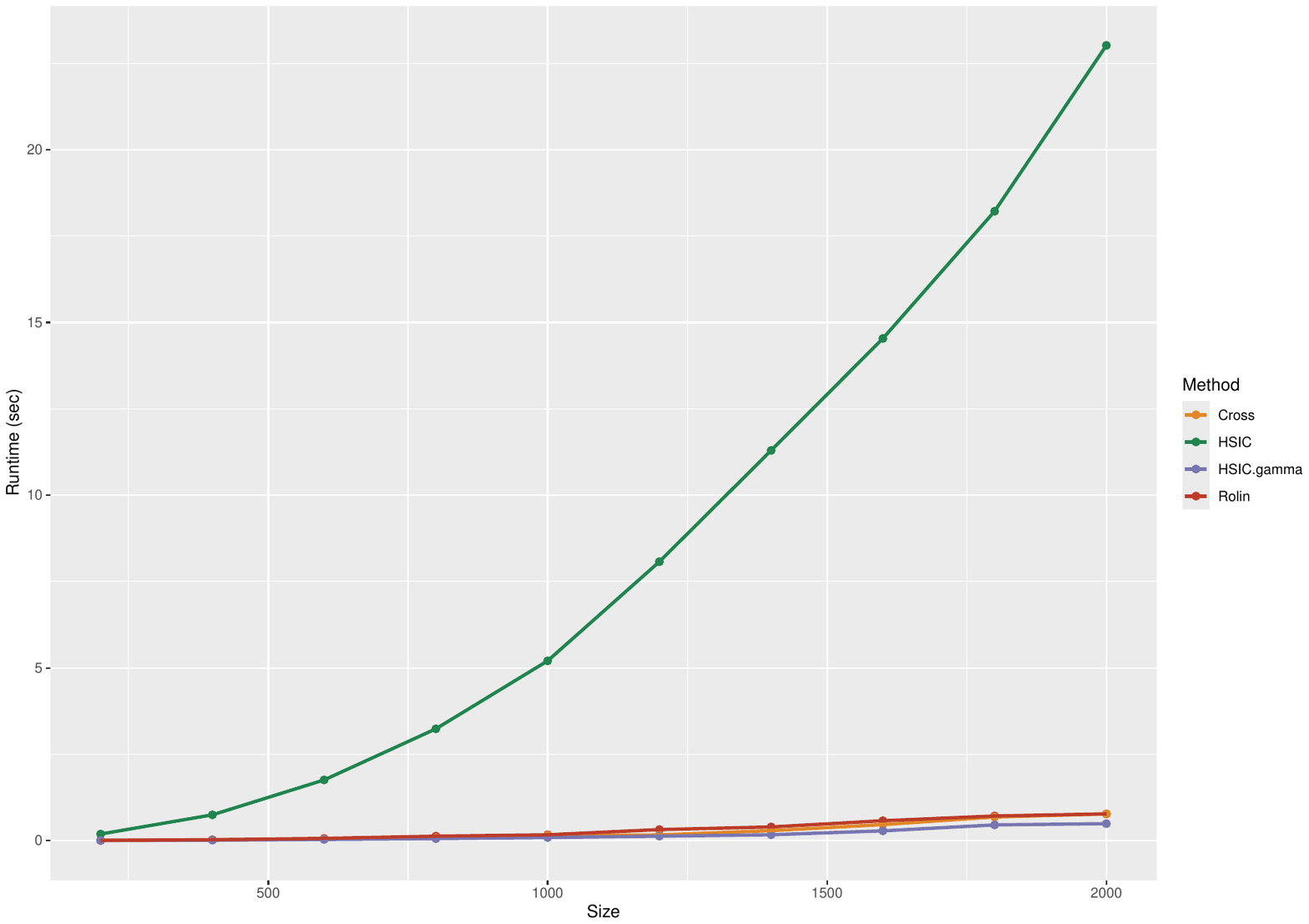}}
 %{pic/RCH_compare.pdf}}
    \caption{The runtime performance of Rolin, Cross, HSIC.gamma and HSIC.}
    \label{figp4.6.1}
\end{figure}

Figure \ref{figp4.6} presents a graphical representation of the runtime for 100 iterations of Rolin alongside HSIC across varying sample sizes. The results prominently illustrate the substantial enhancement in computational efficiency achieved by Rolin. Notably, as the sample size $n$ increases, the curve corresponding to Rolin remains relatively flat, while the HSIC curve exhibits a pronounced acceleration. %This observation underscores the considerable gains in computational efficiency \RED{achievable by omitting the alignment step within the hypothesis testing procedure (Not sure why you said this here)}.
Additionally, Figure \ref{figp4.6.1} displays the averaged runtime over 100 iterations at each sample size, revealing that Rolin demonstrates comparable performance to the other two regularization methods (\textbf{Cross} and \textbf{HSIC.gamma}) in terms of runtime efficiency.

%%% real data 
\subsection{Real data analysis}
In this subsection, we demonstrate the efficacy of the random-lifter independence test by applying it to the single-cell gene expression data from \citet{moignard2015decoding}. 

This dataset consists of 3,934 cells derived from five distinct cell populations across four crucial embryonic stages, spanning from embryonic day 7.0 (E7.0) to day 8.5 (E8.5). It includes gene expression data for 33 transcription factors and nine marker genes, essential for tracking the developmental trajectories of these cells. These genes were quantitatively measured using single-cell quantitative real-time PCR (qRT-PCR), a method that offers precise gene expression measurements through fluorescence signals. These signals directly correspond to the amount of genetic material expressed, providing critical insights into cellular development.

At the specific stage of E8.25, our analysis concentrates on two primary cell populations: 
the GFP$^{+}$ cells (four-somite, 4SG) and Flk1$^{+}$GFP$^{-}$ cells (4SFG$^{-}$), containing 983 and 770 cells, respectively.
Our main goal is to carefully look at and compare the gene-gene expression patterns of these two different cell states.

To get a glance at the correlation between the genes for these two cells, we have generated separate heatmaps as illustrated in  Figure~\ref{figp5}, where the color represents the coefficient values between two corresponding covariates.

We observe two distinct groups of genes. The first group, which includes genes like "Mitf" and "Tbx3" displays a nearly white pattern in the 4SFG$^{-}$ cell heatmap, suggesting a weak or non-existent correlation. However, in the 4SG cell, this group exhibits a clear and significant correlation, as indicated by the presence of colored patterns.
The second group, encompassing genes such as "Gfi1", "HbbbH1", "HoxB2", and "Nfe2" shows a reverse pattern. These genes demonstrate a pronounced and significant correlation in the 4SFG$^{-}$ cell heatmap but exhibit a nearly white correlation effect in the 4SG cell heatmap.

\begin{figure}[!t]
    \centering
    \subfloat[Heat map of the genes in 4SFG$^{-}$ cell.] {
       \includegraphics[width = 1.0\textwidth, height = 0.55\textwidth]{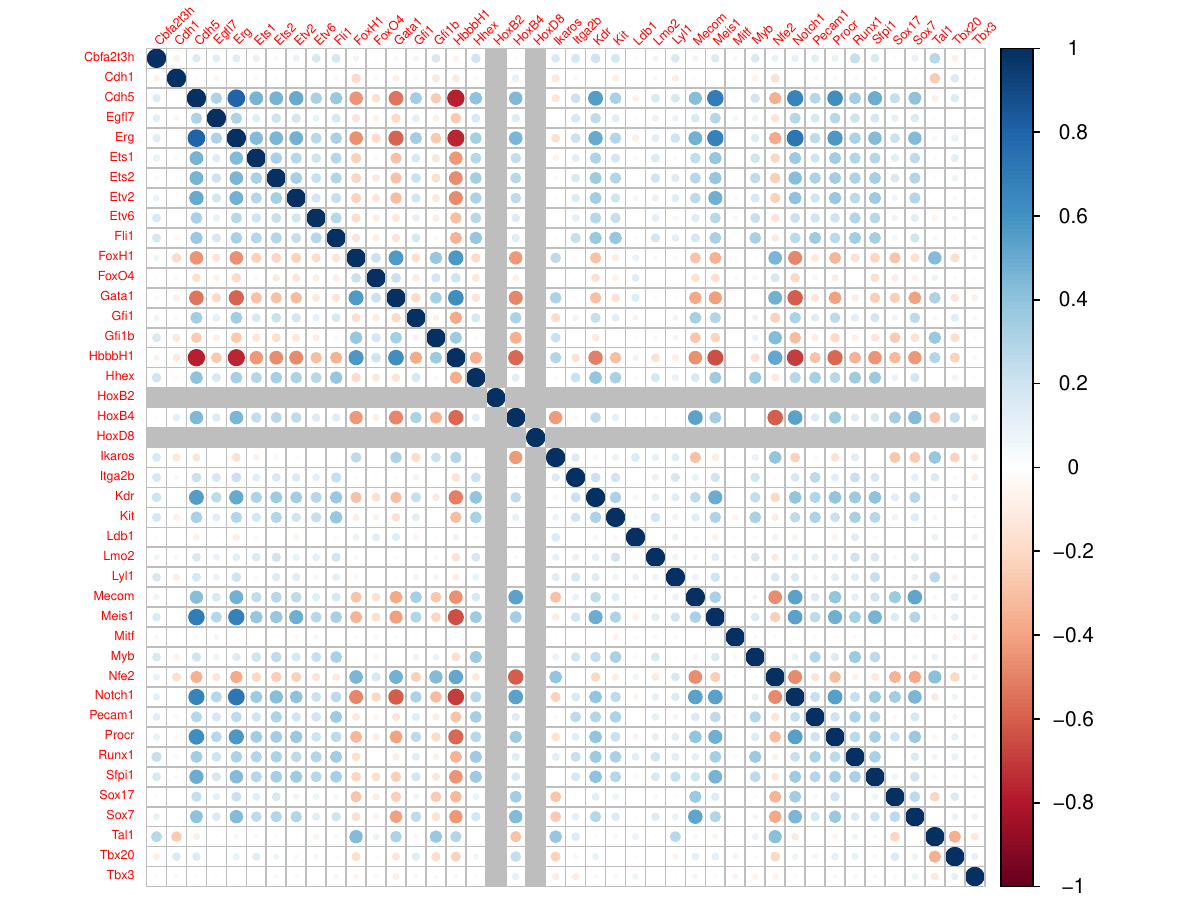}}\\
    \subfloat[Heat map of the genes in 4SG cell.] {
       \includegraphics[width = 1.0\textwidth, height = 0.55\textwidth]{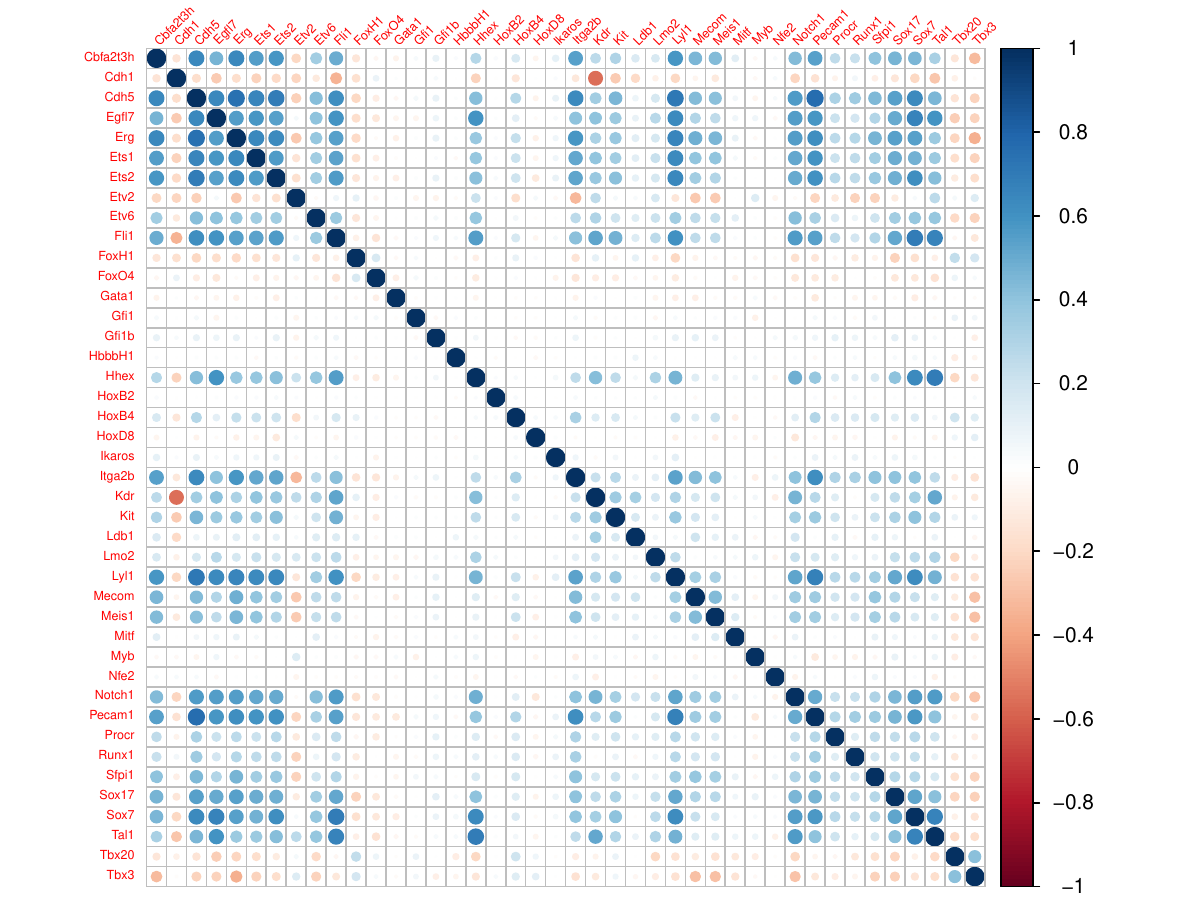}}
    \caption{Heat map of the genes in the cell.}
    \label{figp5}
\end{figure}

\begin{table}[h!]\footnotesize\tabcolsep 11pt
\begin{center}
\caption{The p values of the proposed method and the other two nonparametric methods for testing the independence of two types of genes from the remaining genes under 4SFG$^{-}$ and 4SG cells. 
}
\label{tablelabel}
\vspace{-2mm}
\end{center}
\begin{center}
\begin{tabular}{lccccc}\toprule
Cell & Test Genes & Rolin & Cross &HSIC.gamma & HSIC \\[5pt]
\hline
4SFG$^{-}$ & Group 1 - Remaining genes & 0.4698 & 0.1184 & 0.1101 &0.1825 \\
& Group 2 - Remaining genes & $< 0.0001$ & $< 0.0001$ & $< 0.0001$ & 0.0025 \\
\hline
4SG & Group 1 - Remaining genes & $< 0.0001$ & $< 0.0001$ & $< 0.0001$ &0.0025 \\
& Group 2 - Remaining genes &0.4521 & 0.3305 & 0.3887 &0.3500  \\
\bottomrule
\end{tabular}
\end{center}
 \end{table}

The independence test is performed to examine the relationship between two groups of variables and the remaining variables separately under each of the two cells, and the p-value results are summarized in Table \ref{tablelabel}. 
%\RED{The analysis of transcription factor correlations within 4SFG$^{-}$ and 4SG cells reveals distinct patterns of gene interactions across these cell types. Particularly, non-significant p-values show that the first group of genes exhibits no significant relationships with other genes in 4SFG$^{-}$ cells. This lack of correlation suggests that these genes might be less involved in the primary regulatory mechanisms at this stage or might have stabilized expression that does not respond to changes in other transcription factors. Significant p-values, on the other hand, demonstrate the second group of variables in 4SFG$^{-}$ cells to have strong connections with other genes. This means that they are actively involved in the cell's control systems.

%Interestingly, this pattern is reversed in 4SG cells. Here, the first group of variables exhibits significant relationships, suggesting their active participation in 4SG cell-specific regulatory networks. The second group, on the other hand, does not have any significant relationships, which could mean that they are less involved in the key regulatory pathways or that their expression patterns are more stable within these cells.

%These findings underscore the possibility that regulatory networks involving these genes are differentially regulated between 4SFG$^{-}$ and 4SG cells.}

%This result is also consistent with the other two non-parametric independence test methods. In other words, we obtain consistent results from Rolin that can be much more easily computed.

The study of transcription factor correlations in 4SFG$^{-}$ and 4SG cells shows that these cell types interact with genes differently. In 4SFG$^{-}$ cells, non-significant p-values indicate that we cannot reject the null hypothesis, suggesting that the first group may be independent of the remaining genes in their expression patterns. This implies that the current evidence does not support a statistically significant association between the first group of genes and the remaining genes. Significant p-values, on the other hand, demonstrate that the second group of variables in 4SFG$-$ cells co-expresses with the other genes more frequently than one would anticipate by chance.

Interestingly, this pattern is reversed in 4SG cells. Here, the first group of variables exhibits significant co-expression with the remaining genes, while the second group shows no significant relationships.

These results indicate that the two groups of genes display different co-expression patterns with the remaining genes in the two types of cells. This result is also consistent with the other two non-parametric independence test methods. 

Further functional analyses and experimental validations are crucial to fully understanding the biological relevance of the observed correlations. These steps will help confirm whether the correlation holds biological significance and uncover the underlying mechanisms driving these relationships. We suggest conducting pathway or enrichment analyses to explore these potential connections more deeply. Such analyses can provide insights into whether the co-expressed genes are linked to specific biological pathways or functional categories, thereby enhancing our understanding of their roles and interactions.

\section{Discussion}

This research provides a paradigm shift in hypothesis testing for random objects with complex geometric structures, primarily through the introduction of the Gaussian random fields and random-lifter method technique. The independence measure assesses the disparity between two Gaussian random fields that are defined on the space of the original random variables. The random-lifter method enhances the initial Gaussian random field by incorporating an additional one-dimensional Gaussian random field.
Independence and the zero property of the measure depend on the positive-definiteness of the covariance functions.

To conduct the independence test based on the samples, we introduce a studentized test statistic.  \BLUE{The} cornerstone of our innovation lies in its strategic utilization of the central limit theorem applicable to degenerate U-statistics. Since the random-lifter method compresses the original covariance function, it leads to a phase change in the asymptotic distribution under the appropriate design of the additional one-dimensional Gaussian random field.  

To gain further insights into the proposed method, we examine its theoretical performance compared to the Hilbert-Schmidt Independence Criterion (HSIC \citep{gretton2005kernel}) and show that the power differential between our method and HSIC predominantly hinges on a constant factor, which tends to approach 1 in many scenarios. Furthermore, we establish a more comprehensive upper bound for the minimax properties.
Empirical validation is demonstrated through an array of simulations and their application to a real-world dataset.
Specifically, the results highlight that while our method rivals the current kernel methods in terms of power, but surpasses it in terms of computational efficiency.

For future research, it would be valuable to conduct a more detailed analysis of the bandwidth of the random-lifter method. One potential direction is to explore techniques for learning an optimal bandwidth specific to the data. Additionally, it would be worthwhile to consider a broader range of bandwidth values, including scenarios where the bandwidth is held constant, approaches infinity, or decreases more rapidly toward zero.

Exploring the minimax properties of the method in more general settings would also be beneficial, as this would involve assessing the method's performance under a variety of conditions beyond its current scope. Furthermore, extending the random-lifter method to high-dimensional scenarios and investigating its application to multiple testing represent significant challenges that merit further study. We leave these topics for future research.

%%%%%%%%%%%%%%%%%%%%%%%%%%%%%%%%%%%%%%%%%%%%%%
%\proof
\begin{appendix}
   \section{Technical Proofs}

In the following section, we will present the main proofs of the theorems discussed in the main text. To start, we will establish a crucial result concerning the kernel $\kappa_b$, which plays a fundamental role in our analysis.
\begin{lemma}\label{A1::lem::kerz}
    Let $Z$ be a random variable defined on the set $\mathcal{D}$. We can derive several key properties for the random-lifter terms as follows:
    \begin{align*}
        E[\kappa_b(Z_1, Z_2) \mid Z_1 ] &= b F_1(Z_{1}) g(Z_1) + O_p(b^2); \\
        E [\kappa_b(Z_1, Z_2)] &=   b \int F_1(z) g^2(z) \mathcal{I}\{z \in \mathcal{D}\} d z  + O_p(b^2); \\
    E [\kappa^2_b(Z_1, Z_2)] &=   b \int F_2(z) g^2(z) \mathcal{I}\{z \in \mathcal{D}\} d z  + O_p(b^2); \\
     E [E[\kappa_b(Z_1, Z_2) \mid Z_1 ]^2] &=   b^2 \int F_1^2(z) g^3(z) \mathcal{I}\{z \in \mathcal{D}\} d z  + O_p(b^3)
    \end{align*}
    where $F_i(z) = \int k^i(u) \mathcal{I}\{u \in \mathcal{D}(z)\} d u$, $\mathcal{D}(z) = \{u | z + bu \in \mathcal{D}\}.$
\end{lemma}

\begin{proof}[Proof of Lemma {\rm \ref{A1::lem::kerz}}]
Let's denote the variable transformation of $Z_2$ as $Z_2 = Z_1 + b Z_{12}$. Using Taylor expansion, we can derive the conditional expectation of $\kappa_b(Z_1, Z_2)$ given $Z_1$ as follows:
   \begin{align*}
       & E[\kappa_b(Z_1, Z_2) \mid Z_1 ] \\
    = & \int \kappa_b(Z_1, Z_2) g(Z_2) \mathcal{I}\{Z_2 \in \mathcal{D}\} d Z_2 \\
    = & b \int k(Z_{12}) g(Z_{1} + b Z_{12})  \mathcal{I}\{Z_{1} + b Z_{12} \in \mathcal{D}\} d Z_{12} \\
    = & b \int k(Z_{12}) [g(Z_{1}) + b Z_{12} g^{\prime}(Z_1) + O_p(b^2)] \mathcal{I}\{Z_{12} \in \mathcal{D}(Z_1) \} d Z_{12} \\
    = & b g(Z_{1}) \int k(Z_{12})  \mathcal{I}\{Z_{12} \in \mathcal{D}(Z_1) \} d Z_{12} + O_p(b^2) \\
    = & b F_1(Z_1) g(Z_{1})  + O_p(b^2).
    \end{align*}
    Taking the expectation with respect to $Z_1$, we obtain:
     \begin{align*}
    &E [\kappa_b(Z_1, Z_2)] \\
    =& E\big[ E [\kappa_b(Z_1, Z_2)\mid Z_1] \big] \\
    =& b E [F(Z_1) g(Z_{1})]   + O_p(b^2)\\
    =& b \int F_1(Z_1) g^2(Z_1) \mathcal{I}\{Z_1 \in \mathcal{D}\} d Z_1+ O_p(b^2).
    \end{align*}
    For the expectation of the squared term, we have:
    \begin{align*}
        &E [E[\kappa_b(Z_1, Z_2) \mid Z_1 ]^2] \\
        = & E[b F_1^2(Z_1) g^2(Z_{1})] + O_p(b^3) \\
        = &  b^2 \int F_1^2(Z_1) g^3(Z_{1})\mathcal{I}\{Z_1 \in \mathcal{D}\} d Z_1 + O_p(b^3)
    \end{align*}
    Finally, for the expectation of $\kappa^2_b(Z_1, Z_2)$, we have:
     \begin{align*}
       & E [\kappa^2_b(Z_1, Z_2)] \\
       = & \int \int \kappa^2_b(Z_1, Z_2) g(Z_2) \mathcal{I}\{Z_2 \in \mathcal{D}\} d Z_2 g(Z_1) \mathcal{I}\{Z_1 \in \mathcal{D}\} d Z_1\\
    = & b \int \int k^2(Z_{12}) [g(Z_{1}) + b Z_{12} g^{\prime}(Z_1) + O_p(b^2)] \mathcal{I}\{Z_{12} \in \mathcal{D}(Z_1) \} d Z_{12} g(Z_1) \mathcal{I}\{Z_1 \in \mathcal{D}\} d Z_1\\
       = & b \int F_2(Z_1) g^2(Z_1) \mathcal{I}\{Z_1 \in \mathcal{D}\} d Z_1  + O_p(b^2); 
    \end{align*}
    This concludes the proof.
    
\end{proof}

\begin{remark}\label{rem::A}
    If $Z$ is defined on $\mathbf{R}$, the expressions in Lemma~\ref{A1::lem::kerz} can be simplified as follows:
    \begin{align*}
        E[\kappa_b(Z_1, Z_2) \mid Z_1 ] &= b g(Z_1) + O_p(b^3); \\
        E [\kappa_b(Z_1, Z_2)] &=   b \int  g^2(z) d z  + O_p(b^3); \\
    E [\kappa^2_b(Z_1, Z_2)] &=   b \int k^2(u) du \int g^2(z) d z  + O_p(b^3); \\
     E [E[\kappa_b(Z_1, Z_2) \mid Z_1 ]^2] &=   b^2 \int g^3(z) d z  + O_p(b^3)
    \end{align*}

   The key difference here lies in the omission of the function $F_i(z)$ from the results and a change in the order of the residual term from $O_p(b^2)$ to $O_p(b^3)$, attributed to the symmetry of $k(u)$ at $u = 0$.
\end{remark}

Now, let's delve into the implications of the kernel $\kappa_b(\cdot, \cdot)$ within different intervals. Imagine two intervals of equal length but at different positions, namely $[0, c]$ and $[c_0, c_0 + c]$. We have the relationship between their respective density functions as follows: $g_1(z) = g_2(z + c_0)$, where $z$ belongs to the interval $[0, c]$.

We can define two sets, $\mathcal{D}_1(z)$ and $\mathcal{D}_2(z)$, as:
\begin{align*}
    \mathcal{D}_1(z) = \left\{u \mid \frac{- z}{b} \leq u \leq \frac{c -z}{b} \right\} \text{ and } \mathcal{D}_2(z) = \left\{u \mid \frac{c_0 - z}{b} \leq u \leq \frac{c_0 + c -z}{b} \right\}.
\end{align*}
And it follows that:
\begin{align*}
       % \\
    F_{2i}(z + c_0) = & \int k^i(u) \mathcal{I}\{u \in \mathcal{D}_2(z + c_0)\} d u \\
    = & \int k^i(u) \mathcal{I}\{u \in \mathcal{D}_1(z)\} d u \\
    = & F_{1i}(z). 
\end{align*}
This shows that all four terms in Lemma~\ref{A1::lem::kerz} are equal for both intervals $[0, c]$, and $[c_0, c_0 + c]$.  Consequently, when considering random variables $z$ defined within bounded intervals $[c_0, c_1]$, we only need to focus on the case where $c \geq 0$, specifically the interval $[0, c]$.

Next, we examine the scale transformation of the random variable $Z$. Assuming that the density $g$ is defined on the interval $[0, c]$, we can construct a new random variable, denoted as $\tilde{Z}$, defined on the interval $[0, 1]$. The density of $\tilde{Z}$ is adjusted such that $\tilde{g}(\tilde{z}) = c g(c \tilde{z})$, where $c \tilde{z} = z$, and $\tilde{z} \in [0, 1]$.

With this transformation, we can relate the corresponding functions:

\begin{align*}
    \tilde{F}_i(\tilde{z}) = & \int k^i(\tilde{u}) \mathcal{I}\left\{\frac{-\tilde{z}}{b} \leq \tilde{u} \leq \frac{1 - \tilde{z}}{b} \right\} d\tilde{u} \\
    = & \int k^i(\tilde{u}) \mathcal{I}\left\{\frac{-z}{ b} \leq c\tilde{u} \leq \frac{c - z}{ b} \right\} d\tilde{u} \\
     = & \frac{1}{c} \int k^i\left(\frac{u}{c}\right) \mathcal{I}\left\{\frac{-z}{b} \leq u \leq \frac{c - z}{b} \right\} du.
\end{align*}
Since $F_i(z) = \int k^i(u) \mathcal{I}\{u \in \mathcal{D}(z)\} du$, we can determine two constant sets, denoted as ${B_1^i}$ and ${B_2^i}$, such that:
\begin{align*}
    \frac{B_1^i}{c} F_i(z) \leq  \tilde{F}_i(\tilde{z}) \leq \frac{B_2^i}{c} F_i(z).
\end{align*}
Now, considering the first term in  Lemma~\ref{A1::lem::kerz}, we find that:
\begin{align*}
     E[\kappa_b(\tilde{Z}_1, \tilde{Z}_2) \mid \tilde{Z}_1 ] &= b \tilde{F}_1(\tilde{Z}_{1}) \tilde{g}(\tilde{Z}_1) + O_p(b^2) \\
     & = c b \tilde{F}_1(\tilde{Z}_{1}) g(Z_1) + O_p(b^2).
\end{align*}
Consequently, we can express this relationship as:
\begin{align*}
    B_1^i E[\kappa_b(Z_1, Z_2) \mid Z_1 ] \leq E[\kappa_b(\tilde{Z}_1, \tilde{Z}_2) \mid \tilde{Z}_1 ] \leq B_2^i E[\kappa_b(Z_1, Z_2) \mid Z_1 ].
\end{align*}

This demonstrates that for any random variable $Z$ defined on the interval $[0, c]$, we can introduce a random variable $\tilde{Z}$ defined on $[0, 1]$ through a scale transformation. By adjusting the constants $B_1^i$ and $B_2^i$, we can control the four terms in Lemma~\ref{A1::lem::kerz}. Since the bandwidth parameter $b$ is a tuning parameter, any differences between $Z$ and $\tilde{Z}$ can be attributed to the bandwidth $b$. Therefore, for all scenarios involving bounded intervals, our analysis can focus on random variables defined within the interval $[0, 1]$.

% \section{Asymptotic properties under null hypothesis}
Next, we present the key proofs of the theorems from the main text. To begin, we establish a fundamental result concerning the test statistic:
\begin{lemma}
    Under the null hypothesis, $T_{n, b}$  constitutes a first-order degenerate U-statistic.
    \label{A1::lemma1}
\end{lemma}

% Now we support a brief proof of this result.
\begin{proof}[Proof of Lemma {\rm \ref{A1::lemma1}}]
We can readily verify that:
    \begin{align*}
       &E \big[ h_b(\vesub{W}{1}, \vesub{W}{2}, \vesub{W}{3}, \vesub{W}{4}) | \vesub{W}{1} \big] \\
        =& 2 E \big[\zeta(\vesub{X}{1}, \vesub{X}{2})\eta(\vesub{Y}{1}, \vesub{Y}{2})\kappa_b(Z_1, Z_2) | \vesub{W}{1} \big] - 2 E \big[\zeta(\vesub{X}{1}, \vesub{X}{3})\eta(\vesub{Y}{1}, \vesub{Y}{3})\kappa_b(Z_1, Z_3) | \vesub{W}{1} \big] \\
       = & 0.
    \end{align*}
  Similarly, we can demonstrate that:
    \begin{align*}
        E& \big[ h_b(\vesub{W}{3}, \vesub{W}{2}, \vesub{W}{1}, \vesub{W}{4}) | \vesub{W}{1} \big] = 0,\\
        E& \big[ h_b(\vesub{W}{1}, \vesub{W}{2}, \vesub{W}{4}, \vesub{W}{3}) | \vesub{W}{1} \big] = 0.
    \end{align*}
   By summing up these three terms, we obtain:
    \begin{align*}
        E\big[\Bar{h}_b(\vesub{W}{1}, \vesub{W}{2}, \vesub{W}{3}, \vesub{W}{4}) | \vesub{W}{1} \big] = 0.
    \end{align*}
Thus, we have completed the proof.
\end{proof}

Next, we consider the consistency of the variance estimate. Since $\kappa_{b}(Z_1, Z_2)$ is also related to $n$, the consistency does not always hold. We introduce the following lemma:
\begin{lemma}\label{A1::lemvarest}
    Under the null hypothesis, assume that $E \big[\eta(\vesub{Y}{1}, \vesub{Y}{2}) \big]^{2 + \delta} < \infty $  hold for $\delta > 0$ and $b \rightarrow 0$, $nb \rightarrow \infty$, we have
    \begin{align*}
        E\left[\left|\frac{2}{n(n-1)b} \sum_{1 \leq i < j \leq n} \eta^2(\vesub{Y}{i}, \vesub{Y}{j})\kappa_b^2(Z_i, Z_j) - \frac{1}{b}E[\eta^2(\vesub{Y}{1}, \vesub{Y}{2})\kappa_b^2(Z_1, Z_2)]\right|^{1 + \frac{\delta}{2}}\right] \rightarrow 0.
    \end{align*}
    
\end{lemma}
\begin{proof}[Proof of Lemma {\rm \ref{A1::lemvarest}}]
    We can view the estimate of the variance as a new U statistic of two order. Based on Theorem 1 in \citep{chen1980}, we have
    \begin{align*}
       & E\left[\left|\frac{2}{n(n-1)b} \sum_{1 \leq i < j \leq n} \eta^2(\vesub{Y}{i}, \vesub{Y}{j})\kappa_b^2(Z_i, Z_j) - \frac{1}{b}E[\eta^2(\vesub{Y}{1}, \vesub{Y}{2})\kappa_b^2(Z_1, Z_2)]\right|^{1 + \frac{\delta}{2}}\right] \\
       \leq & \frac{B_1}{b^{1 + \frac{\delta}{2}}} E\left[|\eta^2(\vesub{Y}{1}, \vesub{Y}{2})\kappa_b^2(Z_1, Z_2) - E[\eta^2(\vesub{Y}{1}, \vesub{Y}{2})\kappa_b^2(Z_1, Z_2)]|^{1 +\frac{\delta}{2}} \right] \cdot n^{- \frac{\delta}{2}} \\
        \leq & \frac{B_2}{n^{\frac{\delta}{2}}b^{1 + \frac{\delta}{2}}} \left\{E\left[|\eta^2(\vesub{Y}{1}, \vesub{Y}{2})\kappa_b^2(Z_1, Z_2)|^{1 + \frac{\delta}{2}} \right] + \left[E[\eta^2(\vesub{Y}{1}, \vesub{Y}{2})\kappa_b^2(Z_1, Z_2)]\right]^{1 + \frac{\delta}{2}}  \right\} \\
        \leq & \frac{B_3}{n^{\frac{\delta}{2}}b^{1 + \frac{\delta}{2}}} \left\{E\left[\eta^{2 + \delta}(\vesub{Y}{1}, \vesub{Y}{2})\right] \cdot E\left[\kappa_b^{2+ \delta}(Z_1, Z_2) \right]   \right\}.
    \end{align*}
    where $B_1, B_2, B_3$ are constants. The last inequality follows from Jensen's inequality. Based on similar principles as outlined in Lemma~\ref{A1::lem::kerz}, we can obtain $E[\kappa_b(Z_1, Z_2)]^{2+ \delta} = O_p(b)$. Since $E \big[\eta(\vesub{Y}{1}, \vesub{Y}{2}) \big]^{2 + \delta} < \infty $, we can obtain
    \begin{align*}
        & E\left[\left|\frac{2}{n(n-1)b} \sum_{1 \leq i < j \leq n} \eta^2(\vesub{Y}{i}, \vesub{Y}{j})\kappa_b^2(Z_i, Z_j) - \frac{1}{b}E[\eta^2(\vesub{Y}{1}, \vesub{Y}{2})\kappa_b^2(Z_1, Z_2)]\right|^{1 + \frac{\delta}{2}}\right] \\
        =& O_p\left(\frac{1}{(nb)^{\frac{\delta}{2}}}\right) \rightarrow 0.
    \end{align*}
\end{proof}

%========================================================
\begin{proof}[Proof of Proposition {\rm \ref{Prop::var}}]
We can compute the variance of $T_{n, b}$  based on the properties of U-statistics as follows:
\begin{align*}
    \var(T_{n, b}) = \frac{1}{\CC{n}{4}} \sum_{k = 1}^4 \CC{4}{k} \CC{n-4}{4-k} \sigma_{bk}^2,
\end{align*}
where $\sigma_{bk}^2 = \var\left( E\big[\Bar{h}_b(\vesub{W}{1}, \vesub{W}{2}, \vesub{W}{3}, \vesub{W}{4}) | \vesub{W}{1}, \dots, \vesub{W}{k} \big]  \right)$.

Since $T_{n, b}$ is a first-order degenerate U-statistic, we have $\sigma_{b1} = 0$. To apply the H-decomposition of U-statistics, let's define:
    $$h_{b2}(\vesub{W}{1}, \vesub{W}{2}) = E\big[\Bar{h}_b(\vesub{W}{1}, \vesub{W}{2}, \vesub{W}{3}, \vesub{W}{4}) | \vesub{W}{1}, \vesub{W}{2} \big].$$
Under the null hypothesis, $\ve{X}$ and $\ve{Y}$ are independent.
We observe that:
\begin{align*}
    & E\big[h_b(\vesub{W}{3}, \vesub{W}{2}, \vesub{W}{1}, \vesub{W}{4}) | \vesub{W}{1}, \vesub{W}{2} \big] \\
    = &  \frac{1}{4} \big\{E[\zeta(\vesub{X}{3}, \vesub{X}{2}) \mid \vesub{X}{2}] + E[\zeta(\vesub{X}{1}, \vesub{X}{4}) \mid \vesub{X}{1}] - E[\zeta(\vesub{X}{3}, \vesub{X}{1}) \mid \vesub{X}{1}]- E[\zeta(\vesub{X}{2}, \vesub{X}{4}) \mid \vesub{X}{2}]  \big\}  \\
   & \cdot \big\{ E[\eta(\vesub{Y}{3}, \vesub{Y}{2})\kappa_b(Z_3, Z_2) \mid \vesub{W}{2}] + E[\eta(\vesub{Y}{1}, \vesub{Y}{4})\kappa_b(Z_1, Z_4) \mid \vesub{W}{1}] \\
   & \quad - E[\eta(\vesub{Y}{3}, \vesub{Y}{1})\kappa_b(Z_3, Z_1) \mid \vesub{W}{1}]- E[\eta(\vesub{Y}{2}, \vesub{Y}{4})\kappa_b(Z_2, Z_4) \mid \vesub{W}{2}]  \big\} \\
   = & 0.
\end{align*}
Similarly, we have $E\big[h_b(\vesub{W}{1}, \vesub{W}{2}, \vesub{W}{3}, \vesub{W}{4}) | \vesub{W}{1}, \vesub{W}{2} \big] = E\big[h_b(\vesub{W}{1}, \vesub{W}{2}, \vesub{W}{4}, \vesub{W}{3}) | \vesub{W}{1}, \vesub{W}{2} \big]$.

Since $A_1 = \int F_1(z) g^2(z) d z, A_2 = \int F_1^2(z) g^3(z)  d z, A_3 = \int F_2(z) g^2(z)  d z$, using Lemma~\ref{A1::lem::kerz}, we obtain the expression for $h_{b2}(\vesub{W}{1}, \vesub{W}{2})$:
\begin{align*}
   & h_{b2}(\vesub{W}{1}, \vesub{W}{2})\\ = & \frac{1}{6} \big\{\zeta(\vesub{X}{1}, \vesub{X}{2})  + E[\zeta(\vesub{X}{3}, \vesub{X}{4})] - E[\zeta(\vesub{X}{1}, \vesub{X}{3}) \mid \vesub{X}{1}]- E[\zeta(\vesub{X}{2}, \vesub{X}{4}) \mid \vesub{X}{2}]  \big\}  \\
   & \cdot \big\{ \eta(\vesub{Y}{1}, \vesub{Y}{2})\kappa_b(Z_1, Z_2)  + E[\eta(\vesub{Y}{3}, \vesub{Y}{4})\kappa_b(Z_3, Z_4)] \\
   & \quad - E[\eta(\vesub{Y}{1}, \vesub{Y}{3})\kappa_b(Z_1, Z_3) \mid \vesub{W}{1}]- E[\eta(\vesub{Y}{2}, \vesub{Y}{4})\kappa_b(Z_2, Z_4) \mid \vesub{W}{2}]  \big\} \\
   = &\frac{1}{6} \big\{\zeta(\vesub{X}{1}, \vesub{X}{2})  + E[\zeta(\vesub{X}{3}, \vesub{X}{4})] - E[\zeta(\vesub{X}{1}, \vesub{X}{3}) \mid \vesub{X}{1}]- E[\zeta(\vesub{X}{2}, \vesub{X}{4}) \mid \vesub{X}{2}]  \big\}  \\
   & \cdot \big\{ \eta(\vesub{Y}{1}, \vesub{Y}{2})\kappa_b(Z_1, Z_2)  + A_1 b E[\eta(\vesub{Y}{3}, \vesub{Y}{4})]  - b F_1(Z_1) g(Z_1) E[\eta(\vesub{Y}{1}, \vesub{Y}{3})\mid \vesub{Y}{1}] \\
   & \quad - b F_1(Z_2) g(Z_2) E[\eta(\vesub{Y}{2}, \vesub{Y}{4})\mid \vesub{Y}{2}]  \big\} + O_p(b^3).
\end{align*}

For the variance of $h_{b2}(\vesub{W}{1}, \vesub{W}{2})$, we have
    \begin{align*}
        &\sigma_{b2}^2 \\
        = & E\big[h_{b2}^2(\vesub{W}{1}, \vesub{W}{2}) \big] \\
        = & \frac{1}{36} E \big\{\big[\zeta(\vesub{X}{1}, \vesub{X}{2}) - E[\zeta(\vesub{X}{1}, \vesub{X}{3}) \mid \vesub{X}{1}] - E[\zeta(\vesub{X}{2}, \vesub{X}{4}) \mid \vesub{X}{2}] + E[\zeta(\vesub{X}{3}, \vesub{X}{4})] \big]^2 \\
        &\cdot \big[\eta(\vesub{Y}{1}, \vesub{Y}{2})\kappa_b(Z_1, Z_2) - E[\eta(\vesub{Y}{1}, \vesub{Y}{3})\kappa_b(Z_1, Z_3) \mid \vesub{W}{1}] - E[\eta(\vesub{Y}{2}, \vesub{Y}{4})\kappa_b(Z_2, Z_4) \mid \vesub{W}{2}] \\ 
        & \quad + E[\eta(\vesub{Y}{3}, \vesub{Y}{4})\kappa_b(Z_3, Z_4)] \big]^2  \big\} + O_p(b^3)\\
        = & \frac{b}{36} \big\{E[\zeta^2(\vesub{X}{1}, \vesub{X}{2})] - 2 E\big[\zeta(\vesub{X}{1}, \vesub{X}{2})\zeta(\vesub{X}{1}, \vesub{X}{3}) \big] + \big[E(\zeta(\vesub{X}{1}, \vesub{X}{2}))\big]^2 \big\} \\
        &\cdot \big\{ A_3 E [\eta^2(\vesub{Y}{1}, \vesub{Y}{2})]  - 2A_2 b E\big[\eta(\vesub{Y}{1}, \vesub{Y}{2})\eta(\vesub{Y}{1}, \vesub{Y}{3}) \big] + A_1^2 b \big[ E(\eta(\vesub{Y}{1}, \vesub{Y}{2})) \big]^2  \big\} + O_p(b^3).
    \end{align*}
    Since $\zeta(\cdot, \cdot), \eta(\cdot, \cdot)$ all have finite second moment, $\sigma_{b2}^2$ is an infinitesimal of the same order as $b$. A similar calculation yields $\sigma_{b3}^2 = O_p(b)$ and $ \sigma_{b4}^2 = O_p(b)$.
    
Using the definition of the variance, we can show that:
$$\frac{1}{\CC{n}{4}} \sum_{k = 3}^4 \CC{4}{k} \CC{n-4}{4-k} \sigma_{bk}^2 = O_p\left(\frac{b}{n^3}\right),$$
which implies that the variance is:
\begin{align*}
    \var(T_{n, b}) &= \frac{1}{\CC{n}{4}} \CC{4}{2} \CC{n-4}{2} \sigma_{b2}^2 + O_p\left(\frac{b}{n^3}\right) \\
    &= \frac{72(n-4)(n-5)}{n(n-1)(n-2)(n-3)}\sigma_{b2}^2 + O_p\left(\frac{b}{n^3}\right).
\end{align*}

Now, let's consider the estimate of variance. Note that $\zeta(\vesub{X}{1}, \vesub{X}{2})$ is not related to $n$, 
\begin{align*}
    &(H \zeta H)_{ij} \\=& \frac{(n-1)^2}{n^2} \zeta(\vesub{X}{i}, \vesub{X}{j}) - \frac{(n-1)}{n^2} \left(\sum_{t \neq i} \zeta(\vesub{X}{i}, \vesub{X}{t}) + \sum_{s \neq j} \zeta(\vesub{X}{s}, \vesub{X}{j})\right) + \frac{1}{n^2} \sum_{s \neq i} \sum_{t \neq j} \zeta(\vesub{X}{s}, \vesub{X}{t}),
\end{align*}
 and $ (H \zeta H)_{ij}$ is a consistent estimator of $\zeta(\vesub{X}{i}, \vesub{X}{j}) - E[\zeta(\vesub{X}{i}, \vesub{X}{t}) \mid \vesub{X}{i}] - E[\zeta(\vesub{X}{j}, \vesub{X}{s}) \mid \vesub{X}{j}] + E[\zeta(\vesub{X}{1}, \vesub{X}{2})]$.

 For $\varrho$, we use the conclusion of LEMMA~\ref{A1::lemvarest} and we can get that $ (H \varrho H)_{ij}$ is also a consistent estimator.
 Thus, 
 $$E\big[\big|\frac{1}{n(n-1)}\sum_{i\neq j}M_{ij} - \sigma_{b2}^2\big|^{1 + \frac{\delta}{2}}\big] \rightarrow 0,$$ and $S_{n, b}^2$ is a consistent estimate of $\var(T_{n, b})$.
\end{proof}

%========================Thm 1===========================
For the independence-zero equivalence property, we provide the proof for the Theorem~\ref{Thm::iff}.
\begin{proof}[Proof of Theorem {\rm \ref{Thm::iff}}]

Since there is a one-to-one correspondence between positive definite kernels and RKHSs, we can define the maps $U, Q$ such that
\begin{align*}
    <U(\ve{x}), U(\vesup{x}{\prime})> = & \zeta(\ve{x}, \vesup{x}{\prime}) \\
    <Q(\ve{y}, z), Q(\vesup{y}{\prime}, z^{\prime})> = & \eta(\ve{y}, \vesup{y}{\prime}) \kappa_b(z, z^{\prime})
\end{align*}

By the definition, we have:
$$T_{\zeta, \eta, \kappa_b}(\ve{X}, \ve{Y})=E \big[ U(\ve{X})Q(\ve{Y}, Z) - E(U(\ve{X}))E(Q(\ve{Y}, Z)) \big]^2 \geq 0.$$

Now, let's assume that the probability measure defined by $Z$ is $\lambda$, the probability measure defined by $\ve{W}$ is $\omega^{\prime}$, and the probability measure defined by  $(\ve{Y}, Z)$ is $\nu^{\prime}$. If $\ve{X}$ and $\ve{Y}$ are independent, then $\omega^{\prime} = \mu \otimes \nu^{\prime} = \mu \otimes \nu \otimes \lambda$. This independence leads to the following: 
\begin{align*}
    T_{\zeta, \eta, \kappa_b}(\ve{X}, \ve{Y}) = & \int U(\ve{x}) Q(\ve{y}, z) (\omega^{\prime} - \mu \otimes \nu^{\prime})(d\ve{x}, d(\ve{y}, z)) = 0.
\end{align*}
% which implies $\cov^2_{\zeta, \eta, \kappa_b}(\ve{X}, \ve{Y}) = 0$.

Conversely, if $\cov^2_{\zeta, \eta, \kappa_b}(\ve{X}, \ve{Y})=0$, then 
\begin{align*}
    % &\cov^2_{\zeta, \eta, \kappa_b}(\ve{X}, \ve{Y}) \\
    % =&  \int E\big[ U(\ve{x}) U(\vesup{x}{\prime}) Q(\ve{y}, z) Q(\vesup{y}{\prime}, z^{\prime})\big] (\omega^{\prime} - \mu \otimes \nu^{\prime})^2(d\ve{w},d\vesup{w}{\prime})  \\
    E \left[\int U(\ve{x}) \otimes Q(\ve{y}, z) (\omega^{\prime} - \mu \otimes \nu^{\prime})(d\ve{w}) \right]^2  = 0.
\end{align*}

Since $\zeta(\cdot, \cdot)$ is a positive-definite function,  
$$ \sum _{i=1}^{n}\sum _{j=1}^{n}c_{i}c_{j}\zeta\left(\vesub{x}{i}, \vesub{x}{j}\right)\geq 0$$
holds for any $\vesub{x}{1},\dots ,\vesub{x}{n}\in {\mathcal {X}}$, given 
$ n\in \mathbb {N} ,c_{1},\dots ,c_{n}\in \mathbb {R} $, using the property of positive-definite functions.

By approximating $\mu_1, \mu_2$ with probability measures of finite support, we can establish: $$\int \zeta\left(\ve{x}, \vesup{x}{\prime}\right) \left(\mu_1-\mu_2\right)^2\left(d\ve{x}, d\vesup{x}{\prime}\right)  \geq 0,$$ and the equality holds only when $\mu_1 = \mu_2$. 

Similarly, $\zeta(\ve{y}, \vesup{y}{\prime}), \kappa_b(z, z^{\prime})$ are also positive-definite functions, then for two probability measures $\nu_1, \nu_2$ and $\lambda_1, \lambda_2$ on $\mathcal{Y}$ and $\mathbf{R}$ respectively, we have 
\begin{align*}
    \int \zeta\left(\ve{y}, \vesup{y}{\prime}\right) \left(\nu_1-\nu_2\right)^2\left(d\ve{y}, d\vesup{y}{\prime}\right)   \geq 0, \\
    \int \kappa_b\left(z, z^{\prime}\right) \left(\lambda_1-\lambda_2\right)^2\left(dz, dz^{\prime}\right)   \geq 0,
\end{align*}
and the equality holds only when $\nu_1 = \nu_2$ and $\lambda_1 = \lambda_2$. 

Now, for the product $\zeta\left(\ve{y}, \vesup{y}{\prime}\right) \kappa\left(z, z^{\prime}\right)$, we can show:
\begin{align*}
    & \int \zeta\left(\ve{y}, \vesup{y}{\prime}\right) \kappa\left(z, z^{\prime}\right)\left(\nu_1 \otimes \lambda_1 -\nu_2 \otimes \lambda_2\right)^2\left(d(\ve{y}, z), d(\vesup{y}{\prime}, z^{\prime})\right)  \\
    = & \int \zeta\left(\ve{y}, \vesup{y}{\prime}\right)\left(\nu_1-\nu_2\right)^2\left(d\ve{y}, d\vesup{y}{\prime}\right) \cdot  \int \kappa\left(z, z^{\prime}\right)d\left(\lambda_1-\lambda_2\right)^2\left(d z, d z^{\prime}\right)   \geq 0.
\end{align*}
Thus, we have demonstrated that all these functions are non-negative and that equality holds only when the corresponding measures are equal.

Now, let's define the maps $\beta_U: \mu \rightarrow \int U(\ve{x}) \mu(d\ve{x})$, $\beta_Q: \nu^{\prime} \rightarrow \int Q(\ve{y}, z) \nu^{\prime}(d(\ve{y}, z))$. 
If $\beta_U(\mu_1) = \beta_U(\mu_2)$, then we can show that:
\begin{align*}
    E[\beta_U(\mu_1 - \mu_2)]^2 &= \int E[U(\ve{x})U(\vesup{x}{\prime})]  \left(\mu_1-\mu_2\right)^2\left(d\ve{x}, d\vesup{x}{\prime}\right) \\
    &= \int \zeta(\ve{x}, \vesup{x}{\prime})  \left(\mu_1-\mu_2\right)^2\left(d\ve{x}, d\vesup{x}{\prime}\right) = 0,
\end{align*}
which implies $\mu_1 = \mu_2$. Thus, $\beta_U$,$\beta_Q$ are injective on the sets of probability measures on $\mathcal{X}$ and $\mathcal{Y} \times \mathbf{R}$, respectively. 

We further define a map $\beta_{U \otimes Q}(\omega^{\prime}): \omega \rightarrow \int U(\ve{x}) \otimes Q(\ve{y}, z) \omega^{\prime} (d\ve{w})$. We can show:
\begin{align*}
    T_{\zeta, \eta, \kappa_b}(\ve{X}, \ve{Y}) 
    &=E\left[\beta_{U \otimes Q}(\omega^{\prime} - \mu \otimes \nu^{\prime}) \right]^2 = 0.
\end{align*}

Let $\omega^{\prime} \in \mathcal{B}(\mathcal{X} \times \mathcal{Y} \times \mathbf{R})$ 
 Now, assuming that $K$ is a Gaussian filed defined on $\mathcal{X}$, if we have $\beta_{U \otimes Q}(\omega^{\prime})=0$, then we can define a bounded linear map $T_K:  \mathcal{X} \times \mathcal{Y} \times \mathbf{R}  \rightarrow  \mathcal{Y} \times \mathbf{R}$ as
$$
T_K(U, Q)= E(U K) Q.
$$
For a Borel set $B \subseteq  \mathcal{Y} \times \mathbf{R}$, we define
$$
\nu_K(B)=\int E(U(\ve{x}) K) \mathbf{1}_B(\ve{y}, z) \omega^{\prime}(d\ve{w}),
$$
We can show that:
\begin{align*}
    \beta_Q\left(\nu_K\right)&=\int E(U(\ve{x}) K) Q(\ve{y}, z)  \omega^{\prime} (d\ve{w})\\ 
    &=\int T_K(U(\ve{x}), Q(\ve{y}, z))  \omega^{\prime} (d\ve{w})\\
    &=T_K\left(\beta_{U \otimes Q}(\omega^{\prime})\right)\\
    &=0.
\end{align*}
This implies that $\nu_K=0$ is due to the injectivity of $\beta_Q$. We have: 
$$
\int E(U(\ve{x}) K) \mathbf{1}_B(\ve{y}, z) \omega^{\prime}(d\ve{w}) = E \left[K\int U(\ve{x}) \mathbf{1}_B(\ve{y}, z) \omega^{\prime}(d\ve{w}) \right]=0 .
$$
Since this holds for each $K \in \mathcal{X}$, for every Borel set $B \subseteq \mathcal{Y} \times \mathbf{R}$, we get:
$$
\int U(\ve{x}) \mathbf{1}_B(\ve{y}, z) \omega^{\prime}(d\ve{w}) =0 .
$$
For a Borel set $A \subseteq \mathcal{X}$, we define $\mu_B(A)=\omega^{\prime}(A \times B)$.
Since $\beta_U\left(\mu_B\right)=\int U(\ve{x}) \mathbf{1}_B(\ve{y}, z) \omega^{\prime}(d\ve{w})=0$, we can conclude that  $\mu_B=0$ by the injectivity of $\beta_U$. Thus, $\omega^{\prime}(A \times B)=0$ for every pair of Borel sets $A$ and $B$, which implies that $\beta_{U \otimes Q}(\omega^{\prime})$ is injective. As $E \left[\beta_{U \otimes Q}(\omega^{\prime} - \mu \otimes \nu^{\prime}) \right]^2 = 0$, it follows that $\omega^{\prime} - \mu \otimes \nu^{\prime}=0$  and consequently,  $\ve{X}$ and $\ve{Y}$ are independent.

This concludes the proof of Theorem {\rm \ref{Thm::iff}}.
\end{proof}

%========================Thm 2===========================
\begin{proof}[Proof of Theorem {\rm \ref{Thm::randomH0}}]
   
      Let's begin by examining the structure of $T_{n, b}$ through the lens of the H-decomposition. We express $T_{n, b}$ as
    \begin{align*}
        T_{n, b} = & \frac{12}{n(n-1)} \sum_{i < j} h_{b2}(\vesub{W}{i}, \vesub{W}{j}) + R_n^{(2)},
    \end{align*}
   where the residual term, $R_n^{(2)}$ is bounded as
    \begin{align*}
        \var \big(n \sigma^{-1}_{b2} R_n^{(2)} \big) = & n^2 \sum_{k = 3}^4 \CC{4}{k} \frac{1}{\CC{n}{k}} \frac{\sigma_{bk}}{\sigma_{b2}} = O_p\left(\frac{1}{n}\right).
    \end{align*}

   To further understand the behavior and properties of $T_{n, b}$, we delve into its martingale structure. Define
    \begin{align*}
        \Gamma_i = \frac{\sqrt{2}}{(n-1) \sigma_{b2}}\sum_{j = 1}^{i - 1} h_{b2}(\vesub{W}{i}, \vesub{W}{j}).
    \end{align*}
  This definition allows us to see that the sequence ${D_i = \sum_{j = 2}^i \Gamma_j, i = 2, \cdots }$ forms a martingale, with the sequence ${\Gamma_i, i = 2, \cdots }$ itself being a martingale difference sequence.

To confirm that a degenerate U-statistic adheres to the martingale central limit theorem (MCLT), we need to establish two main conditions:
\begin{enumerate}
    \item[(i)] $\sum_i E\left[\Gamma_{i}^2 \mid \mathcal{F}_{i-1}\right] \stackrel{p}{\longrightarrow} 1$,
    \item[(ii)] $\forall \varepsilon>0, \sum_i E\left[\Gamma_{i}^2 I(\Gamma_{i} > \varepsilon)\right] \longrightarrow 0$.
\end{enumerate}  
    Utilizing Lemma B.4 from \citet{fan1996}, we analyze the normality of our statistic through the definition of  
    \begin{align*}
        G_b(\vesub{W}{1}, \vesub{W}{2}) = & E \big[h_{b2}(\vesub{W}{1}, \vesub{W}{3})h_{b2}(\vesub{W}{2}, \vesub{W}{3}) \mid \vesub{W}{1}, \vesub{W}{2} \big].
    \end{align*}
    Our goal is to demonstrate that
        \begin{align}
       & \frac{E\big[G_b^2(\vesub{W}{1}, \vesub{W}{2}) \big] + \frac{1}{n} E\big[h_{b2}^4(\vesub{W}{1}, \vesub{W}{2}) \big]}{\big(E\big[h_{b2}^2(\vesub{W}{1}, \vesub{W}{2}) \big] \big)^2}  \to 0.
       \label{Gb}
    \end{align}
     
      This statistical formulation can be considered as a specific instantiation of the Lindeberg-Levy central limit theorem or the martingale central limit theorem.  Let's dissect the main equation, focusing on ensuring that both critical terms approach zero, fulfilling the necessary conditions for statistical normality.
    The equation in question can be deconstructed into two pivotal terms. The first term, $E\big[G_b^2(\vesub{W}{1}, \vesub{W}{2}) \big] / \big(E\big[h_{b2}^2(\vesub{W}{1}, \vesub{W}{2}) \big] \big)^2$ is used to verify condition (i) in MCLT, and the second term, $\frac{1}{n} E\big[h_{b2}^4(\vesub{W}{1}, \vesub{W}{2}) \big] / \big(E\big[h_{b2}^2(\vesub{W}{1}, \vesub{W}{2}) \big] \big)^2$ is for condition (ii). These trends are crucial for satisfying the two conditions required for our statistical analysis.
    
     We modify the second term and provide a more general condition.
    % We commence with ensuring that the second term diminishes to zero, as indicated by:
    %  $$\frac{\frac{1}{n} E\big[h_{b2}^4(\vesub{W}{1}, \vesub{W}{2}) \big]}{\big(E\big[h_{b2}^2(\vesub{W}{1}, \vesub{W}{2}) \big] \big)^2} \rightarrow 0.$$
     % To support this, 
     We delve into the expected magnitude of $\Gamma_{i}$, considering both scenarios where exceeds a threshold $\varepsilon$ and where it doesn't. 
     \begin{align*}
     E[|\Gamma_{i}|^{2 + \delta}] = & E [|\Gamma_{i}|^{2 + \delta}  \mathcal{I}\{\Gamma_{i} > \varepsilon\}] + E [|\Gamma_{i}|^{2 + \delta}  \mathcal{I}\{\Gamma_{i} \leq \varepsilon\}]\\
     \geq & E [|\Gamma_{i}|^{2 + \delta}  \mathcal{I}\{\Gamma_{i} > \varepsilon\}] \\
      \geq & \varepsilon^{\delta} E\left[\Gamma_{i}^2 \mathcal{I}\{\Gamma_{i} > \varepsilon\}\right],
     \end{align*}
     This analysis leads us to conclude that it's crucial to demonstrate that $\sum_{i = 2}^n E[|\Gamma_i|^{2 + \delta}] \rightarrow 0$.

     Employing Rosenthal's inequality, we calculate 
     \begin{align*}
         &E[|\Gamma_i|^{2 + \delta}]\\
         = & \left(\frac{2}{(n-1) \sigma_{b2}} \right)^{2 + \delta} E\left[\sum_{j = 1}^{i - 1} h_{b2}(\vesub{W}{i}, \vesub{W}{j}) \right]^{2 + \delta} \\
         = & \left(\frac{\sqrt{2}}{(n-1) \sigma_{b2}} \right)^{2 + \delta} E\left[ E\left[ |\sum_{j = 1}^{i - 1} h_{b2}(\vesub{W}{i}, \vesub{W}{j})|^{2 + \delta} \mid \vesub{W}{i} \right] \right] \\
         \leq & \left(\frac{\sqrt{2}}{(n-1) \sigma_{b2}} \right)^{2 + \delta} E\left[ B_1\left(\sum_{j = 1}^{i - 1}E\left[ | h_{b2}(\vesub{W}{i}, \vesub{W}{j})|^{2 + \delta} \mid \vesub{W}{i} \right] \right) \right. \\
         & \left.+ B_2 \left(\sum_{j = 1}^{i - 1}E\left[ | h_{b2}(\vesub{W}{i}, \vesub{W}{j})|^2 \mid \vesub{W}{i} \right] \right)^{1 + \frac{\delta}{2}}  \right] \\
         = & \left(\frac{\sqrt{2}}{(n-1) \sigma_{b2}} \right)^{2 + \delta} \left\{ B_1 (i - 1) E\left[ h_{b2}(\vesub{W}{i}, \vesub{W}{j})\right]^{2 + \delta}  + B_2 E\left[(i - 1)E\left[ | h_{b2}(\vesub{W}{i}, \vesub{W}{j})|^2 \mid \vesub{W}{i} \right] \right]^{1 + \frac{\delta}{2}}  \right\} \\
          \leq & \left(\frac{\sqrt{2}}{(n-1) \sigma_{b2}} \right)^{2 + \delta} \left\{ B_1 (i - 1) E\left[ h_{b2}(\vesub{W}{i}, \vesub{W}{j})\right]^{2 + \delta}  + B_2 (i - 1)^{1 + \frac{\delta}{2}} E\left[ h_{b2}(\vesub{W}{i}, \vesub{W}{j})\right]^{2 + \delta}  \right\},
     \end{align*}
     where $B_1, B_2$ are some constants.

     By adding them up, we have
     \begin{align*}
         \sum_{i = 2}^n E[|\Gamma_i|^{2 + \delta}] \leq B_3 \left(\frac{\sqrt{2}}{(n-1) \sigma_{b2}} \right)^{2 + \delta} \left\{ \sum_{i = 2}^n [(i - 1)^{1 + \frac{\delta}{2}} + (i - 1)] E\left[ h_{b2}(\vesub{W}{i}, \vesub{W}{j})\right]^{2 + \delta}\right\},
     \end{align*}
     where $B_3$ is a constant. Given that the order of $\sum_{i = 2}^n [(i - 1)^{1 + \frac{\delta}{2}} + (i - 1)]$ is $n^{2 + \frac{\delta}{2}}$, it becomes crucial to ensure that the ratio of the expectation of $h_{b2}^{2 + \delta}$ over a particular normalization factor approaches zero:
     \begin{align}
         \frac{E\big[h_{b2}^{2 + \delta}(\vesub{W}{1}, \vesub{W}{2}) \big]}{n^{\frac{\delta}{2}}\big(E\big[h_{b2}^2(\vesub{W}{1}, \vesub{W}{2}) \big] \big)^{1 + \frac{\delta}{2}}} \rightarrow 0.
         \label{cond2}
     \end{align}
     
    We further examine the expectation of $h_{b2}^2 (\vesub{W}{1}, \vesub{W}{2})$. Under the null hypothesis, the expectation is intricately composed of the interactions between $\zeta$ and $\eta$, parameters that capture the statistical properties of our model, scaled by $b$, the bandwidth parameter. This detailed examination reveals that:
    \begin{align*}
        & E \big[h_{b2}^2 (\vesub{W}{1}, \vesub{W}{2}) \big] \\
        =&\frac{b}{36} \big\{E[\zeta^2(\vesub{X}{1}, \vesub{X}{2})] - 2 E\big[\zeta(\vesub{X}{1}, \vesub{X}{2})\zeta(\vesub{X}{1}, \vesub{X}{3}) \big] + \big[E(\zeta(\vesub{X}{1}, \vesub{X}{2}))\big]^2 \big\} \\
        &\cdot \big\{ A_3 E [\eta^2(\vesub{Y}{1}, \vesub{Y}{2})]  - 2A_2 b E\big[\eta(\vesub{Y}{1}, \vesub{Y}{2})\eta(\vesub{Y}{1}, \vesub{Y}{3}) \big] + A_1^2 b \big[ E(\eta(\vesub{Y}{1}, \vesub{Y}{2})) \big]^2  \big\} + O_p(b^3),
    \end{align*}
    where $A_2, A_3$ are defined in Proposition~\ref{Prop::var}. It is worth noting that $E[g^3(z)] < \infty$, $F_1(z) < 1$, $\int k^2_b(z) dz < \infty$, ensuring that $A_1, A_2, A_3$ are finite constants. Consequently, we can deduce that $E\big[h_{b2}^2(\vesub{W}{1}, \vesub{W}{2}) \big] = O_p(b)$.
    
    We delve into the expectations of $\eta^{2 + \delta}$ and and its interaction with $\kappa^{2 + \delta}_b$, based on similar principles as outlined in Lemma~\ref{A1::lem::kerz}. The expectation is given by:
    % Clearly, the following equation holds
     \begin{align*}
        E \big[\eta^{2 + \delta}(\vesub{Y}{1}, \vesub{Y}{2})\kappa^{2 + \delta}_b(Z_1, Z_2) \big] = A_{2 + \delta} b E[\eta^{2 + \delta}(\vesub{Y}{1}, \vesub{Y}{2})],
    \end{align*}
    where $ A_{2 + \delta} = \int F_{2 + \delta}(z) g^2(z) \mathcal{I}\{z \in \mathcal{D}\} d z$. Since $\int k^{2 + \delta}(z) dz < \infty$, ensuring that  $$E \big[\eta^{2 + \delta}(\vesub{Y}{1}, \vesub{Y}{2})\kappa^{2 + \delta}_b(Z_1, Z_2) \big] = O_p(b).$$
    
    Turning our attention to the crucial numerator term in the equation (\ref{cond2}), we analyze  the expectation of the ($2 + \delta$)-th power of $h_{b2} (\vesub{W}{1}, \vesub{W}{2})$ to establish its behavior. The calculation unfolds as follows:
    \begin{align*}
          &   E\big[h_{b2}^{2 + \delta} (\vesub{W}{1}, \vesub{W}{2}) \big] \\
        = & \frac{1}{6^{2 + \delta}} E \big[ \zeta(\vesub{X}{1}, \vesub{X}{2})  + E[\zeta(\vesub{X}{3}, \vesub{X}{4})] - E[\zeta(\vesub{X}{1}, \vesub{X}{3}) \mid \vesub{X}{1}]- E[\zeta(\vesub{X}{2}, \vesub{X}{4}) \mid \vesub{X}{2}]   \big]^{2 + \delta}  \\
       & \cdot E \big[ \eta(\vesub{Y}{1}, \vesub{Y}{2})\kappa_b(Z_1, Z_2)  + A_2 b E[\eta(\vesub{Y}{3}, \vesub{Y}{4})]  - b g(Z_1) F(Z_1) E[\eta(\vesub{Y}{1}, \vesub{Y}{3})\mid \vesub{Y}{1}] \\
       & \quad - b \phi(Z_2) E[\eta(\vesub{Y}{2}, \vesub{Y}{4})\mid \vesub{Y}{2}] + O_p(b^3) \big\}  \big]^{2 + \delta} \\
        = & \frac{1}{6^{2 + \delta}} E \big[ \zeta(\vesub{X}{1}, \vesub{X}{2})  + E[\zeta(\vesub{X}{3}, \vesub{X}{4})] - E[\zeta(\vesub{X}{1}, \vesub{X}{3}) \mid \vesub{X}{1}]- E[\zeta(\vesub{X}{2}, \vesub{X}{4}) \mid \vesub{X}{2}]   \big]^{2 + \delta}  \\
       & \cdot  E \big[\eta(\vesub{Y}{1}, \vesub{Y}{2})\kappa_b(Z_1, Z_2) + O_p(b) \big]^{2 + \delta} \\
       \leq & \frac{2^{\delta}}{6^{2 + \delta}} E \big[ \zeta(\vesub{X}{1}, \vesub{X}{2})  + E[\zeta(\vesub{X}{3}, \vesub{X}{4})] - E[\zeta(\vesub{X}{1}, \vesub{X}{3}) \mid \vesub{X}{1}]- E[\zeta(\vesub{X}{2}, \vesub{X}{4}) \mid \vesub{X}{2}]   \big]^{2 + \delta}  \\
       & \cdot  E \big[\eta^{2 + \delta}(\vesub{Y}{1}, \vesub{Y}{2})\kappa_b^{2 + \delta}(Z_1, Z_2) + O_p(b^{2 + \delta}) \big].
    \end{align*}
    The last inequality is obtained from Jensen's inequality. 

    % For LEMMA~\ref{A1::lemvarest}, we have $E\big[\kappa_b^{2 + \delta}(Z_1, Z_2) \big] = O_p(b)$.
    Furthermore, since $E \big[\zeta^{2 + \delta}(\vesub{X}{1}, \vesub{X}{2}) \big] < \infty$, and $E \big[\eta^{2 + \delta}(\vesub{Y}{1}, \vesub{Y}{2}) \big] < \infty $, we can conclude that $E\big[h_{b2}^{2 + \delta}(\vesub{W}{1}, \vesub{W}{2}) \big] = O_p(b)$. Thus,
    \begin{align*}
         \frac{E\big[h_{b2}^{2 + \delta}(\vesub{W}{1}, \vesub{W}{2}) \big]}{n^{\frac{\delta}{2}}\big(E\big[h_{b2}^2(\vesub{W}{1}, \vesub{W}{2}) \big] \big)^{1 + \frac{\delta}{2}}} = O_p\left(\frac{b}{n^{\frac{\delta}{2}} b^{1 + \frac{\delta}{2}}}\right) = O_p\left(\frac{1}{(nb)^{\frac{\delta}{2}} }\right)\rightarrow 0.
     \end{align*}

    For the first term $E\big[G_b^2(\vesub{W}{1}, \vesub{W}{2}) \big] / \big(E\big[h_{b2}^2(\vesub{W}{1}, \vesub{W}{2}) \big] \big)^2$, we show that it also tends toward zero. 
    Let $\Tilde{\zeta}(\vesub{X}{1}, \vesub{X}{2}) = \zeta(\vesub{X}{1}, \vesub{X}{2})  + E[\zeta(\vesub{X}{3}, \vesub{X}{4})] - E[\zeta(\vesub{X}{1}, \vesub{X}{3}) \mid \vesub{X}{1}]- E[\zeta(\vesub{X}{2}, \vesub{X}{4}) \mid \vesub{X}{2}] $, $Z_{2} = Z_{1} + b Z_{12}$, $Z_{3} = Z_{1} + b Z_{13}$, $Z_{4} = Z_{1} + b Z_{14}$, we obtain
       \begin{align*}     
       & E\big[ E\big[\kappa_b(Z_1, Z_3)\kappa_b(Z_2, Z_3) \mid Z_1, Z_2 \big]^2 \big] \\
        = & E\big[\kappa_b(Z_1, Z_3)\kappa_b(Z_1, Z_4)\kappa_b(Z_2, Z_3)\kappa_b(Z_2, Z_4) \big]  \\
        = & \int \kappa_b(Z_1, Z_3)\kappa_b(Z_1, Z_4)\kappa_b(Z_2, Z_3)\kappa_b(Z_2, Z_4) g(Z_1) g(Z_2) g(Z_3) g(Z_4) \mathcal{I}\{Z_1 \in \mathcal{D}\} \\
        & \mathcal{I}\{Z_2 \in \mathcal{D}\} \mathcal{I}\{Z_3 \in \mathcal{D}\} \mathcal{I}\{Z_4 \in \mathcal{D}\} d Z_1 d Z_2 d Z_3 d Z_4 \\
        = & b^3 \int k\left(Z_{13}\right)  k\left(Z_{14}\right) k\left(Z_{12} + Z_{13}\right) k\left(Z_{12} + Z_{14}\right) g(Z_1) g(Z_{1} + b Z_{12}) g(Z_{1} + b Z_{13})  \\
        & g(Z_{1} + b Z_{14}) \mathcal{I}\{Z_1 \in \mathcal{D}\} \mathcal{I}\{Z_{12} \in \mathcal{D}(Z_1)\} \mathcal{I}\{Z_{13} \in \mathcal{D}(Z_1)\} \mathcal{I}\{Z_{14} \in \mathcal{D}(Z_1)\}  d Z_1 d Z_{12} d Z_{13} d Z_{14} \\
        = & b^3 \int  k\left(Z_{13}\right)  k\left(Z_{14}\right) k\left(Z_{12} + Z_{13}\right) k\left(Z_{12} + Z_{14}\right)g^4(Z_1)  \mathcal{I}\{Z_1 \in \mathcal{D}\} \mathcal{I}\{Z_{12} \in \mathcal{D}(Z_1)\} \\
        &  \mathcal{I}\{Z_{13} \in \mathcal{D}(Z_1)\} \mathcal{I}\{Z_{14} \in \mathcal{D}(Z_1)\}  d Z_{12} d Z_{13} d Z_{14}d Z_1  + O_p(b^6) \\
         = & b^3 B_4 + O_p(b^6),
    \end{align*}
    where $B_4$ is a constant. This equation leads 
   \begin{align*}     
       & E\big[G_b^2(\vesub{W}{1}, \vesub{W}{2}) \big]\\   = & E \big[h_{b2}(\vesub{W}{1}, \vesub{W}{3})h_{b2}(\vesub{W}{2}, \vesub{W}{3})h_{b2}(\vesub{W}{1}, \vesub{W}{4})h_{b2}(\vesub{W}{2}, \vesub{W}{4}) \big] \\
        = & O_p(b^3).
    \end{align*}
    By combining the above results, we have
    \begin{align*}
       & \frac{E\big[G_b^2(\vesub{W}{1}, \vesub{W}{2}) \big] }{\big(E\big[h_{b2}^2(\vesub{W}{1}, \vesub{W}{2}) \big] \big)^2} = O_p\left(\frac{b^3 }{b^2} \right) \to 0.
    \end{align*}
    
    Drawing on Lemma B.4 as cited in \citet{fan1996} and the principles of the martingale limit theorem, we reach a pivotal conclusion regarding the distribution of our test statistic, $nT_{n, b}$. Specifically, we find that:
    \begin{align*}
        \frac{nT_{n, b}}{6\sqrt{2}\sigma_{b2}} \to N(0, 1).
    \end{align*}
    Guided by Proposition \ref{Prop::var},  $E[|S_{n, b}^2 - \var(T_{n, b})|^{1 + \frac{\delta}{2}}] \rightarrow 0$. Consequently, by employing Slutsky's lemma, we conclude that
    \begin{align*}
        \frac{T_{n, b}}{S_{n, b}} \to N(0, 1).
    \end{align*}
This concludes the proof.
\end{proof}

%========================Thm===========================
To complete the proof of Theorem \ref{random:rate}, 
We continue to use the notation $G_b(W_1, W_2)$ introduced in Theorem~\ref{Thm::randomH0} and induce some lemmas stating the following.

\begin{lemma}\label{A1::lemma3}
    For $1 \leq i < j \leq n$ and $1 \leq k < l \leq n$, if $(i ,j) \neq (k, l)$, then we have
    $$ E \big[ G(\vesub{W}{i}, \vesub{W}{j})G(\vesub{W}{k}, \vesub{W}{l})\big] = 0.$$
\end{lemma}

\begin{proof}[Proof of Lemma \ref{A1::lemma3}]
    Since $E \big[ h_b(\vesub{W}{1}, \vesub{W}{2}, \vesub{W}{3}, \vesub{W}{4}) | \vesub{W}{1} \big] = 0$, then for $i \neq j \neq k \neq l$, it holds that
    \begin{align*}
         E \big[ G(\vesub{W}{i}, \vesub{W}{j})G(\vesub{W}{k}, \vesub{W}{l})\big] = & \big[ E \big[ G(\vesub{W}{1}, \vesub{W}{2})\big] \big]^2\\
         = & \big[ E [ h_{b2} (\vesub{W}{1}, \vesub{W}{3})  h_{b2} (\vesub{W}{2}, \vesub{W}{3})  ] \big]^2 \\
         = & \big[ E [ E[ h_{b2} (\vesub{W}{1}, \vesub{W}{3})  h_{b2} (\vesub{W}{2}, \vesub{W}{3}) \mid \vesub{W}{3}]  \big]^2 \\
         = &  0.
    \end{align*}
    If only one element of $(i, j)$ and $(k, l)$ is the same, Without loss of generality, we assume $i = k$ and $j \neq l$, then 
       \begin{align*}
         &E \big[ G(\vesub{W}{i}, \vesub{W}{j})G(\vesub{W}{k}, \vesub{W}{l})\big]\\
         % = & E \big[ E_{\ve{W}} [ h_{b2} (\vesub{W}{i}, \ve{W})  h_{b2} (\vesub{W}{j}, \ve{W})  ] E_{\ve{W}} [ h_{b2} (\vesub{W}{i}, \ve{W})  h_{b2} (\vesub{W}{l}, \ve{W})  ]\big]\\
         = & E \big[  h_{b2} (\vesub{W}{1}, \vesub{W}{3})  h_{b2} (\vesub{W}{2}, \vesub{W}{3})  h_{b2} (\vesub{W}{1}, \vesub{W}{5})  h_{b2} (\vesub{W}{4}, \vesub{W}{5})  \big] \\
         = &  E \big[ E[ h_{b2} (\vesub{W}{1}, \vesub{W}{3})h_{b2} (\vesub{W}{1}, \vesub{W}{5}) \mid \vesub{W}{3}, \vesub{W}{5}] E[h_{b2} (\vesub{W}{2}, \vesub{W}{3}) \mid \vesub{W}{3}]    E[h_{b2} (\vesub{W}{4}, \vesub{W}{5}) \mid \vesub{W}{5}]  \big] \\
         = & 0
    \end{align*}
    Thus, the proof is complete.
\end{proof}

\begin{lemma}\label{random:lemma}
    $\forall 0 < \delta < 1$, assume that $E\left[\zeta(\vesub{X}{1}, \vesub{X}{2}) \right]^{2 + 2\delta} < \infty, E\left[\eta(\vesub{Y}{1}, \vesub{Y}{2}) \right]^{2 + 2\delta} < \infty$, then under the null hypothesis is true, then we have
    \begin{align*}
        &\sup_{t \in \mathbf{R}}  \left| \pr \left( \frac{nT_{n, b}}{6 \sqrt{2} \sigma_{b2}}   \le t \right)  - \Phi(t)  \right| \\
        \lesssim & \Big\{\frac{1}{(nb)^{\delta}} \frac{E [\Tilde{\zeta}(\vesub{X}{1}, \vesub{X}{2}) \eta(\vesub{Y}{1}, \vesub{Y}{2})]^{2+2\delta} }{ \big[ E [\Tilde{\zeta}(\vesub{X}{1}, \vesub{X}{2}) \eta(\vesub{Y}{1}, \vesub{Y}{2})]^{2} \big]^{1+\delta}}\\
      & +b^{\frac{1+\delta}{2}} \frac{\left\{E \left[g_x(\vesub{X}{1}, \vesub{X}{2}, \vesub{X}{3}, \vesub{X}{4}) g_y(\vesub{Y}{1}, \vesub{Y}{2}, \vesub{Y}{3}, \vesub{Y}{4})\right]  \right\} ^{\frac{1+\delta}{2}}}{\big[ E |\Tilde{\zeta}(\vesub{X}{1}, \vesub{X}{2}) \eta(\vesub{Y}{1}, \vesub{Y}{2})|^{2} \big]^{1+\delta}}  \Big\}^{\frac{1}{3+2\delta}}.
    \end{align*}
\end{lemma}

\begin{proof}[Proof of lemma \ref{random:lemma}]
    Recalling Theorem \ref{Thm::randomH0}, it is clear that as $n \to \infty$, 
    \begin{align*}
        \frac{nT_{n, b}}{6\sqrt{2}\sigma_{b2}} & \to N(0, 1), \\
        nT_{n, b} & = \frac{12}{ (n-1)} \sum h_{b2}(\vesub{W}{i}, \vesub{W}{j}) + n R_n^{(2)},
    \end{align*}
    where $\var \big[n\sigma_{b2}^{-1} R_n^{(2)} \big] = O(\frac{1}{n})$. Let $M_{ni} = \frac{\sqrt{2}}{\sigma_{b2} (n-1)} \sum_{j < i} h_{b2}(\vesub{W}{i}, \vesub{W}{j})$, $\mathcal{F}_{i}$ be the sigma algebra generated by $\vesub{W}{1}, \dots, \vesub{W}{i}$, then $\{\Gamma_{ni}\}$ is a square-integrable martingale-difference sequence (MDS) with filtration $\{\mathcal{F}_{i}\}$ and satisfies the Lindeberg condition. From the Berry-Essen theorem, it follows that
     \begin{align*}
        \sup_{x \in \mathbf{R}} \left| \pr \left(\sum_{i = 1}^n \Gamma_{ni} \le t \right) - \Phi(t) \right| \le B_{\delta} (L_{n, 2\delta} + N_{n, 2\delta})^{\frac{1}{3+2\delta}},
    \end{align*}
    where $B_{\delta}$ is a $\delta$-dependent constant, and
    \begin{align*}
        L_{n, 2\delta} = \sum_{i = 1}^n E  \left[ \Gamma_{ni} \right]^{2+2\delta},\ \ N_{n, 2\delta} = E\left[ 
\left| \sum_{i = 1}^n E (\Gamma_{ni}^2 | \mathcal{F}_{i - 1}) - 1 \right|^{1+\delta} \right].
    \end{align*}
    
 By Rosenthal’s inequality for the sum of independent random variables, there exists a constant $B_{\delta,1} > 0$ such that
 \begin{align*}
       & L_{n, 2\delta} \\ \le & \frac{B_{\delta, 1}}{\sigma_{b2}^{2+2\delta} (n-1)^{2+2\delta} }\sum_{i = 1}^n E\Big\{ \sum_{j < i}E \big[ (h_{b2}(\vesub{W}{j}, \vesub{W}{i}))^{2+2\delta} | \vesub{W}{i} \big]  + \big( \sum_{j < i} E \big[(h_{b2}(\vesub{W}{j}, \vesub{W}{i}))^{2} | \vesub{W}{i} \big] \big)^{1+\delta}
        \Big\}  \\
        \le & \frac{B_{\delta, 1}}{\sigma_{b2}^{2+2\delta}(n-1)^{2+2\delta}} \Big\{ \sum_{i= 1}^n (i - 1) E \big[\big( h_{b2}(\vesub{W}{1}, \vesub{W}{2}) \big)^{2+\delta} \big] + \sum_{i = 1}^n (i - 1)^{1+ \delta}  E \big[\big( h_{b2}(\vesub{W}{1}, \vesub{W}{2}) \big)^{2+2\delta} \big]
        \Big\} \\
        \lesssim & \frac{n^2 + n^{2+2\delta}}{n^{2+2\delta}} \frac{ E \big[\big( h_{b2}(\vesub{W}{1}, \vesub{W}{2}) \big)^{2+2\delta} \big]}{\sigma_{b2}^{2+2\delta}} \\
        \lesssim & \frac{1}{(nb)^{\delta}} \frac{E [\Tilde{\zeta}(\vesub{X}{1}, \vesub{X}{2}) \eta(\vesub{Y}{1}, \vesub{Y}{2})]^{2+2\delta} }{ \big[ E [\Tilde{\zeta}(\vesub{X}{1}, \vesub{X}{2}) \eta(\vesub{Y}{1}, \vesub{Y}{2})]^{2} \big]^{1+\delta}}.
    \end{align*}

    We next consider the second term, $N_{n, 2\delta}$. According to the von-Bahr Essen Inequality derived in Theorem 9.3.a in \citep{lin2011probability}, we have
    \begin{align*}
        N_{n, 2\delta} = & E \left[ \left| \sum_{i = 1}^n E (M_{ni}^2 | \mathcal{F}_{i - 1}) - 1  \right|^{1+\delta} \right] \\
        \le & 2^{\delta} E \left[ \left| \frac{2}{\sigma_{b2}^2(n-1)^2} \sum_{i = 1}^n E \big[\sum_{j < i} (h_{b2} (\vesub{W}{j}, \vesub{W}{i}))^2 |\vesub{W}{1},\cdots, \vesub{W}{i-1}  \big] - 1  \right|^{1+\delta} \right] \\
        & + 2^{\delta} E \left[ \left| \frac{4}{\sigma_{b2}^2(n-1)^2} \sum_{i = 1}^n E \big[\sum_{j < l < i} h_{b2} (\vesub{W}{j}, \vesub{W}{i})  h_{b2} (\vesub{W}{l}, \vesub{W}{i}) |\vesub{W}{1},\cdots, \vesub{W}{i-1}  \big]  \right|^{1+\delta} \right] \\
        \le & \frac{2^{1+\delta}}{(1-B_{\delta, 2})\sigma_{b2}^{2+2\delta}(n-1)^{2+2\delta}} \sum_{j = 1}^n E \Big\{ \big[ \sum_{i = j+1}^{n} E \big[h_{b2}^2(\vesub{W}{j}, \vesub{W}{i}) | \vesub{W}{j} \big] - (n - j)\sigma_{b2}^2 \big]^{1+\delta} \Big\} \\
        & + 2^{\delta} E \left[ \left| \frac{4}{\sigma_{b2}^2(n-1)^2} \sum_{i = 1}^n E \big[\sum_{j < l < i} h_{b2} (\vesub{W}{j}, \vesub{W}{i})  h_{b2} (\vesub{W}{l}, \vesub{W}{i}) |\vesub{W}{1},\cdots, \vesub{W}{i-1}  \big]  \right|^{1+\delta} \right],
        % \lesssim & \frac{1}{(nb)^{\delta}} \frac{E |\Tilde{\zeta}(\vesub{X}{1}, \vesub{X}{2}) \eta(\vesub{Y}{1}, \vesub{Y}{2})|^{2+2\delta} }{ \big( E |\Tilde{\zeta}(\vesub{X}{1}, \vesub{X}{2}) \eta(\vesub{Y}{1}, \vesub{Y}{2})|^{2} \big)^{1+\delta}}.
    %   & + E \big[ \big( E_{\vesub{W}{3}}\Tilde{\zeta}(\vesub{X}{1}, \vesub{X}{3})\eta(\vesub{Y}{1}, \vesub{Y}{3})\Tilde{\zeta}(\vesub{X}{2}, \vesub{X}{3})\eta(\vesub{Y}{2}, \vesub{Y}{3}) \big)^2 \big]^{\frac{1+\delta}{2}} b^{\frac{1+\delta}{2}}.
    \end{align*}
    where $B_{\delta, 2}$ is a constant less than 1.
 
     By Jensen's inequality, the first term can be bounded by
    \begin{align*}
       & \frac{2^{1+\delta}}{(1-B_{\delta, 2})\sigma_{b2}^{2+2\delta}(n-1)^{2+2\delta}} \sum_{j = 1}^n E \Big\{ \big[ \sum_{i = j + 1}^n E \big[h_{b2}^2(\vesub{W}{j}, \vesub{W}{i}) | \vesub{W}{j} \big] - (n - j)\sigma_{b2}^2 \big]^{1+\delta} \Big\} \\
      \le &\frac{2^{1+2\delta}}{(1-B_{\delta, 2})\sigma_{b2}^{2+2\delta}(n-1)^{2+2\delta}} \sum_{j = 1}^n (n - j)^{1 + \delta}   \Big\{  E\big[\big( E \big[h_{b2}^2(\vesub{W}{j}, \vesub{W}{i}) | \vesub{W}{i} \big]\big)^{1+\delta}\big] + \big(E \big[h_{b2}^2(\vesub{W}{j}, \vesub{W}{i}) \big] \big)^{1+\delta} \Big\} \\
        \le &\frac{2^{2+2\delta}}{(1-B_{\delta, 2})\sigma_{b2}^{2+2\delta}(n-1)^{2+2\delta}} \sum_{j = 1}^n (n-j)^{1 + \delta} E \big[\big( h_{b2}(\vesub{W}{j}, \vesub{W}{i}) \big)^{2+2\delta} \big]  \big]  \\
    \lesssim & \frac{1}{(nb)^{\delta}} \frac{E [\Tilde{\zeta}(\vesub{X}{1}, \vesub{X}{2}) \eta(\vesub{Y}{1}, \vesub{Y}{2})]^{2+2\delta} }{ \big[ E [\Tilde{\zeta}(\vesub{X}{1}, \vesub{X}{2}) \eta(\vesub{Y}{1}, \vesub{Y}{2})]^{2} \big]^{1+\delta}}.
    \end{align*}
    
    By Jensen's inequality and Lemma \ref{A1::lemma3}, the second term can be bounded by
      \begin{align*}
        & 2^{\delta} E \left[ \left| \frac{4}{\sigma_{b2}^2(n-1)^2} \sum_{i = 1}^n E \big[\sum_{j < l < i} h_{b2} (\vesub{W}{j}, \vesub{W}{i})  h_{b2} (\vesub{W}{l}, \vesub{W}{i}) |\vesub{W}{1},\cdots, \vesub{W}{i-1}  \big]  \right|^{1+\delta} \right] \\
        \le & 2^{\delta} \left\{E \left[  \frac{4}{\sigma_{b2}^2(n-1)^2} \sum_{i = 1}^n E \big[\sum_{j < l < i} h_{b2} (\vesub{W}{j}, \vesub{W}{i})  h_{b2} (\vesub{W}{l}, \vesub{W}{i}) |\vesub{W}{1},\cdots, \vesub{W}{i-1}  \big]  \right]^{2}  \right\} ^{\frac{1+\delta}{2}}\\
         = & 2^{\delta} \left\{E \left[  \frac{4}{\sigma_{b2}^2(n-1)^2} \sum_{1 \leq j < l < n} (n-l) E \big[ h_{b2} (\vesub{W}{j}, \ve{W})  h_{b2} (\vesub{W}{l}, \ve{W}) \mid \vesub{W}{j}, \vesub{W}{l} \big]   \right]^{2} \right\} ^{\frac{1+\delta}{2}}\\
        = & \frac{2^{2 + 3\delta}}{\sigma_{b2}^{2+2\delta}(n-1)^{2+2\delta}}  \left\{E \left[  \sum_{1 \leq i < j < n}(n-j)^2 E \big[ (h_{b2} (\vesub{W}{j}, \ve{W})  h_{b2} (\vesub{W}{l}, \ve{W}) )^2 \mid \vesub{W}{j}, \vesub{W}{l}\big] \right] \right\} ^{\frac{1+\delta}{2}}\\
         \lesssim & b^{-(1+\delta)}  \left\{E \left[   h_{b2} (\vesub{W}{1}, \vesub{W}{2})  h_{b2} (\vesub{W}{1}, \vesub{W}{3}) h_{b2} (\vesub{W}{2}, \vesub{W}{4})  h_{b2} (\vesub{W}{3}, \vesub{W}{4})\right] \right\} ^{\frac{1+\delta}{2}}  \\
        \lesssim &  b^{\frac{1+\delta}{2}} \frac{\left\{E \left[g_x(\vesub{X}{1}, \vesub{X}{2}, \vesub{X}{3}, \vesub{X}{4}) g_y(\vesub{Y}{1}, \vesub{Y}{2}, \vesub{Y}{3}, \vesub{Y}{4})\right]  \right\} ^{\frac{1+\delta}{2}}}{\big( E |\Tilde{\zeta}(\vesub{X}{1}, \vesub{X}{2}) \eta(\vesub{Y}{1}, \vesub{Y}{2})|^{2} \big)^{1+\delta}} 
    \end{align*}
    Thus,
    \begin{align*}
        &\sup_{t \in \mathbf{R}}  \left| \pr \left( \frac{nT_{n, b}}{6 \sqrt{2} \sigma_{b2}}   \le t \right)  - \Phi(t)  \right| \\
        \lesssim & \Big\{\frac{1}{(nb)^{\delta}} \frac{E [\Tilde{\zeta}(\vesub{X}{1}, \vesub{X}{2}) \eta(\vesub{Y}{1}, \vesub{Y}{2})]^{2+2\delta} }{ \big[ E [\Tilde{\zeta}(\vesub{X}{1}, \vesub{X}{2}) \eta(\vesub{Y}{1}, \vesub{Y}{2})]^{2} \big]^{1+\delta}}\\
      & +b^{\frac{1+\delta}{2}} \frac{\left\{E \left[g_x(\vesub{X}{1}, \vesub{X}{2}, \vesub{X}{3}, \vesub{X}{4}) g_y(\vesub{Y}{1}, \vesub{Y}{2}, \vesub{Y}{3}, \vesub{Y}{4})\right]  \right\}^{\frac{1+\delta}{2}}}{\big[ E |\Tilde{\zeta}(\vesub{X}{1}, \vesub{X}{2}) \eta(\vesub{Y}{1}, \vesub{Y}{2})|^{2} \big]^{1+\delta}}  \Big\}^{\frac{1}{3+2\delta}}.
    \end{align*}
\end{proof}

\begin{proof}[Proof of Theorem \ref{random:rate}]
% Now we focus on the proof of Theorem \ref{random:H1}.
Observe that
\begin{align*}
        &\sup_{t \in \mathbf{R}}  \left| \pr\left( \frac{T_{n, b}}{S_{n, b}} \le t \right) - \Phi(t) \right| \\
        \le & 2\sup_{t \in \mathbf{R}} \left| \pr\left( \frac{nT_{n, b}}{6\sqrt{2} \sigma_{b2}}  \le t \right) - \Phi(t) \right| + \sup_{t \in \mathbf{R}} \left| \Phi(t) - \Phi(t\sqrt{1+ r}) \right| \\
        & + \sup_{t \in \mathbf{R}} \left| \Phi(t) - \Phi(t\sqrt{1 - r}) \right|  + 2  \pr \left( \left|\frac{n^2S_{n, b}^2}{72 \sigma_{b2}^2} - 1\right| \ge r \right) 
\end{align*}
The aforementioned inequality holds for $\forall 0 < r < 1$. Borrowing the fact from Lemma 75 \citep{Gao2021} that, there exists a positive constant $B < \infty$ such that
\begin{align*}
        & \sup_{t \in \mathbf{R}} \left| \Phi(t) - \Phi(t\sqrt{1+ r}) \right|  \leq B r, \\
        & \sup_{t \in \mathbf{R}} \left| \Phi(t) - \Phi(t\sqrt{1 - r}) \right| \leq B r.
\end{align*}
In addition, $S_{n, b} \rightarrow 72 \sigma_{b2}^2/n^2$, by Markov's inequality, we have
\begin{align*}
         \pr \left( \left|\frac{n^2 S_{n, b}^2}{72 \sigma_{b2}^2} - 1\right| \ge r \right) \leq& \frac{n^{2+2\delta} E\left|S_{n, b}^2 - \frac{72 \sigma_{b2}^2}{n^2}\right|^{1+\delta}}{r^{1 + \delta} (72\sigma_{b2})^{2+2\delta}} \\
       \lesssim & \frac{1}{(nb)^{\delta}} \frac{E [\Tilde{\zeta}(\vesub{X}{1}, \vesub{X}{2}) \eta(\vesub{Y}{1}, \vesub{Y}{2})]^{2+2\delta} }{ \big[ E [\Tilde{\zeta}(\vesub{X}{1}, \vesub{X}{2}) \eta(\vesub{Y}{1}, \vesub{Y}{2})]^{2} \big]^{1+\delta}}.
\end{align*}
Then, let $r = b^{\frac{1+\delta}{2(3+2\delta)}}$, it holds that
\begin{align*}
     & \sup _{t \in \mathbf{R}} \left|\pr\left(\frac{T_{n, b}}{S_{n, b}}\le t\right)-\Phi(t)\right| \\
 \lesssim  & \Big\{\frac{1}{(nb)^{\delta}} \frac{E [\Tilde{\zeta}(\vesub{X}{1}, \vesub{X}{2}) \eta(\vesub{Y}{1}, \vesub{Y}{2})]^{2+2\delta} }{ \big[ E [\Tilde{\zeta}(\vesub{X}{1}, \vesub{X}{2}) \eta(\vesub{Y}{1}, \vesub{Y}{2})]^{2} \big]^{1+\delta}}\\
      & +b^{\frac{1+\delta}{2}} \frac{\left\{E \left[g_x(\vesub{X}{1}, \vesub{X}{2}, \vesub{X}{3}, \vesub{X}{4}) g_y(\vesub{Y}{1}, \vesub{Y}{2}, \vesub{Y}{3}, \vesub{Y}{4})\right]  \right\} ^{\frac{1+\delta}{2}}}{\big[ E |\Tilde{\zeta}(\vesub{X}{1}, \vesub{X}{2}) \eta(\vesub{Y}{1}, \vesub{Y}{2})|^{2} \big]^{1+\delta}}  \Big\}^{\frac{1}{3+2\delta}}.
  \end{align*}

\end{proof}

%========================================================
%========================Thm 4===========================
% \section{Power Analysis}
\begin{proof}[Proof of Theorem {\rm \ref{random:H1}}]
    For simplicity, we define $h_{b1}(\vesub{W}{1}) = E\big[ \Bar{h}_b(\vesub{W}{1}, \vesub{W}{2}, \vesub{W}{3}, \vesub{W}{4}) | \vesub{W}{1} \big]$. Under the alternative hypothesis, the formulation for $h_{b1}(\vesub{W}{1})$ is detailed as:
    \begin{align*}
          & h_{b1}(\vesub{W}{1})\\
          = & \frac{1}{4} E \Big[ \big( \zeta(\vesub{X}{1}, \vesub{X}{2}) + \zeta(\vesub{X}{3}, \vesub{X}{4}) - \zeta(\vesub{X}{1}, \vesub{X}{3}) - \zeta(\vesub{X}{2}, \vesub{X}{4}) \big)   \\
    & \cdot \big(b  g(Z_1) F(Z_1) \eta(\vesub{Y}{1}, \vesub{Y}{2}) + A_1 b \eta(\vesub{Y}{3}, \vesub{Y}{4})  - b  g(Z_1) F(Z_1) \eta(\vesub{Y}{1}, \vesub{Y}{3}) - A_1 b \eta(\vesub{Y}{2}, \vesub{Y}{4})   \big) \big| \vesub{W}{1}
    \Big] \\
    = & \frac{b}{2} \Big\{ E\big[ \zeta(\vesub{X}{1}, \vesub{X}{2}) \eta(\vesub{Y}{1}, \vesub{Y}{2}) \big] g(Z_1) F(Z_1) +  E\big[ \zeta(\vesub{X}{3}, \vesub{X}{4}) \eta(\vesub{Y}{1}, \vesub{Y}{2}) \big] g(Z_1) F(Z_1) \\
    & -  E\big[ \zeta(\vesub{X}{1}, \vesub{X}{3}) \eta(\vesub{Y}{1}, \vesub{Y}{2}) \big] g(Z_1) F(Z_1) -  E\big[ \zeta(\vesub{X}{3}, \vesub{X}{4}) \eta(\vesub{Y}{1}, \vesub{Y}{3}) \big] g(Z_1) F(Z_1) \\
    & +  E\big[ \zeta(\vesub{X}{1}, \vesub{X}{2}) \eta(\vesub{Y}{3}, \vesub{Y}{4}) \big] A_1 +  E\big[ \zeta(\vesub{X}{3}, \vesub{X}{4}) \eta(\vesub{Y}{3}, \vesub{Y}{4}) \big] A_1 \\
    & -  E\big[ \zeta(\vesub{X}{1}, \vesub{X}{2}) \eta(\vesub{Y}{2}, \vesub{Y}{4}) \big] A_1 -  E\big[ \zeta(\vesub{X}{3}, \vesub{X}{4}) \eta(\vesub{Y}{2}, \vesub{Y}{4}) \big] A_1 \big| \vesub{W}{1}
    \Big\}.
    \end{align*}
    This expression is bounded by the parameter $b$, given that the second moments of $\zeta(\vesub{X}{1}, \vesub{X}{2})$ and $\eta(\vesub{Y}{1}, \vesub{Y}{2})$ are finite:
    $E \big[\zeta(\vesub{X}{1}, \vesub{X}{2}) \big]^2 < \infty,E \big[\eta(\vesub{Y}{1}, \vesub{Y}{2}) \big]^2 < \infty $.
    
    Utilizing the U-statistic convergence theorem, the normalized sum of deviations from the mean,
    \begin{align*}
        \frac{1}{\sqrt{n} \sigma_{b1}} \sum_{i = 1}^n \big[h_{b1}(\vesub{W}{i}) - T_{\zeta, \eta, \kappa_b}(\ve{X}, \ve{Y})) \big] \to N(0, 1),
    \end{align*}
    where $\sigma_{b1}^2 = \var(h_{b1}(\vesub{W}{1}))$. 
    
    The desired result follows.
\end{proof}

Let's analyze the power function of the random-lifter method, comparing it with the Hilbert-Schmidt Independence Criterion (HSIC) method using identical kernels $\zeta$ and $\eta$.  This comparison is centered around the test statistic $T_n$,  with $\gamma$ representing the expectation of $T_n$, formulated as:
\begin{align*}
    \gamma = E\left[\zeta(\vesub{X}{1}, \vesub{X}{2}) \eta(\vesub{Y}{1}, \vesub{Y}{2})  \right] + E\left[\zeta(\vesub{X}{1}, \vesub{X}{2})\right] E\left[\eta(\vesub{Y}{1}, \vesub{Y}{2}) \right]   - 2 E\left[\zeta(\vesub{X}{1}, \vesub{X}{2}) \eta(\vesub{Y}{1}, \vesub{Y}{2}) \right]
\end{align*}
Under similar conditions as Theorem~\ref{random:H1}, we observe that $\frac{\sqrt{n}(T_n - \gamma)}{\sigma_1}$ is distributed according to a standard normal distribution, where $\sigma_1^2$ denotes the variance of the first component in the H-decomposition. We still use the notation of  $A_1 = \int F_1(z) g^2(z)  d z, A_2 = \int F_1^2(z) g^3(z)  d z$.

We first provide some useful lemmas.
\begin{lemma}\label{A3::lem::exp}
    For $Z$ defined on $\mathcal{D}$, the expectation of $T_{n, b}$ is given by:
    \begin{align*}
        E[T_{n, b}] = b \int F_1(z) g^2(z) \mathcal{I}\{z \in \mathcal{D} \} d z \cdot \gamma + O_p(b^2) = A_1 b \gamma +O_p(b^2).
    \end{align*}
\end{lemma}

\begin{proof}[Proof of Lemma {\rm \ref{A3::lem::exp}}]
     We first consider the relationship between $\cov^2_{\zeta, \eta, \kappa_b}(\ve{X}, \ve{Y})$ and $\gamma$, and it follows that
    \begin{align*}
         &E[T_{n, b}]\\
        = & E\left[\zeta(\vesub{X}{1}, \vesub{X}{2}) \eta(\vesub{Y}{1}, \vesub{Y}{2}) \kappa_b(Z_1, Z_2) \right] + E\left[\zeta(\vesub{X}{1}, \vesub{X}{2})\right] E\left[\eta(\vesub{Y}{1}, \vesub{Y}{2}) \kappa_b(Z_1, Z_2) \right]\\
        & - 2 E\left[\zeta(\vesub{X}{1}, \vesub{X}{2}) \eta(\vesub{Y}{1}, \vesub{Y}{2}) \kappa_b(Z_1, Z_2) \right] \\
        = &  E\left[ \kappa_b(Z_1, Z_2) \right] \cdot \left\{ E\left[\zeta(\vesub{X}{1}, \vesub{X}{2}) \eta(\vesub{Y}{1}, \vesub{Y}{2})  \right] + E\left[\zeta(\vesub{X}{1}, \vesub{X}{2})\right] E\left[\eta(\vesub{Y}{1}, \vesub{Y}{2}) \right] \right. \\
        & \left.   - 2 E\left[\zeta(\vesub{X}{1}, \vesub{X}{2}) \eta(\vesub{Y}{1}, \vesub{Y}{2}) \right] \right\} \\
  = & b \int F_1(z) g^2(z) \mathcal{I}\{z \in \mathcal{D} \} d z \cdot \gamma + O_p(b^2).
    \end{align*}
\end{proof}
Similar to the conclusion above, if $\mathcal{D} = \mathbf{R}$, we can obtain
    \begin{align*}
        E[T_{n, b}] = b \int g^2(z) d z \cdot \gamma + O_p(b^3).
    \end{align*}

Next, we consider the variance of the first component of the H-decomposition. For simplicity, we denote $\zeta_{ij} = \zeta(\vesub{X}{i}, \vesub{X}{j})$, the same notation is also used for $\eta$.

\begin{lemma}\label{A3::lem::var}
     For $Z$ defined on $\mathcal{D}$, we have
     \begin{align*}
        & \sigma_{b1}^2 \\= & \frac{b^2}{4} \left\{  A_2 E\left[E\left[\zeta_{12}\eta_{12} + \zeta_{34}\eta_{12} - \zeta_{13}\eta_{12} - \zeta_{34}\eta_{13} \mid \vesub{W}{1} \right]^2 \right] \right. \\
         &  + A_1^2 \left[ E\left[E\left[\zeta_{12}\eta_{34} + \zeta_{34}\eta_{34} - \zeta_{12}\eta_{24} - \zeta_{34}\eta_{24} \mid \vesub{W}{1} \right]^2 \right] \right. \\
         &  - \left(E\left[\zeta_{12}\eta_{12} + \zeta_{34}\eta_{12} - \zeta_{13}\eta_{12} - \zeta_{34}\eta_{13}\right]\right)^2 - \left(E\left[\zeta_{12}\eta_{34} + \zeta_{34}\eta_{34} - \zeta_{12}\eta_{24} - \zeta_{34}\eta_{24}\right]\right)^2 \\
         & + 2 \text{cov}(E\left[\zeta_{12}\eta_{34} + \zeta_{34}\eta_{34} - \zeta_{12}\eta_{24} - \zeta_{34}\eta_{24} \mid \vesub{W}{1} \right], \\
         & \left. \left. \quad \quad E\left[\zeta_{12}\eta_{12} + \zeta_{34}\eta_{12} - \zeta_{13}\eta_{12} - \zeta_{34}\eta_{13} \mid \vesub{W}{1} \right]) \right] \right\} + O_p(b^3).
     \end{align*}
\end{lemma}

\begin{proof}[Proof of Lemma {\rm \ref{A3::lem::var}}]
    For the first component of the H-decomposition, we have
    \begin{align*}
       & h_{b1}(\vesub{W}{1} ) \\= &  \frac{b}{2} \left\{F_1(Z_1) g(Z_1) \cdot E\left[\zeta_{12}\eta_{12} + \zeta_{34}\eta_{12} - \zeta_{13}\eta_{12} - \zeta_{34}\eta_{13} \mid \vesub{W}{1} \right] \right.\\
        & \left. + \int F_1(z) g^2(z) \mathcal{I}\{z \in \mathcal{D} \} d z \cdot E\left[\zeta_{12}\eta_{34} + \zeta_{34}\eta_{34} - \zeta_{12}\eta_{24} - \zeta_{34}\eta_{24} \mid \vesub{W}{1} \right]  \right\} + O_p(b^2).
    \end{align*}
    The result can be obtained by computing $\sigma_{b1}^2 = \var[h_{b1}(\vesub{W}{1} )]$.
\end{proof}

Let 
\begin{align*}
    H_1(\vesub{W}{1}) = & E\left[\zeta_{12}\eta_{12} + \zeta_{34}\eta_{12} - \zeta_{13}\eta_{12} - \zeta_{34}\eta_{13} \mid \vesub{W}{1} \right], \\
     H_2(\vesub{W}{1}) = & E\left[\zeta_{12}\eta_{34} + \zeta_{34}\eta_{34} - \zeta_{12}\eta_{24} - \zeta_{34}\eta_{24} \mid \vesub{W}{1} \right],\\
     \tilde{H}_1=& E[H_1(\vesub{W}{1})]^2,\\
     \tilde{H}_2=& E[H_2(\vesub{W}{1})]^2 + 2 E[H_1(\vesub{W}{1}) H_2(\vesub{W}{1})] - (E[H_1(\vesub{W}{1})])^2\\ &- (E[H_2(\vesub{W}{1})])^2- 2 E[H_1(\vesub{W}{1})] E[H_2(\vesub{W}{1})],
\end{align*}
then, it follows that 
\begin{align*}
    \sigma_{1}^2 = \tilde{H}_1 +\tilde{H}_2,  
\end{align*}
and 
\begin{align*}
    \sigma_{b1}^2 = b^2 \left(A_2  \tilde{H}_1 + A_1^2  \tilde{H}_2 \right)  + O_p(b^3).
\end{align*}
These lemmas lay the groundwork for understanding the statistical behavior of the test statistic under the influence of the random-lifter approach, emphasizing the method's adaptability to different kernel configurations and its implications for the expected value and variance. Additionally, the variance part is split into two terms, each augmented by a different coefficient.

\begin{proof}[Proof of Theorem {\rm \ref{random:H12}}]
    We begin by examining the power function of the random-lifter independence test (Rolin). For Theorem~\ref{random:H1},  the power function is essentially the probability of correctly rejecting the null hypothesis when it is false, which can be represented as follows:
    \begin{align*}
       K_{ROLIN} = & \pr\left(\frac{T_{n, b}}{S_{n, b}} > \Phi(1 - \alpha)\right) \\
        = & \pr\left(\sqrt{n} \frac{T_{n, b} - T_{\zeta, \eta, \kappa_b}(\ve{X}, \ve{Y})}{\sigma_{b1}} > \sqrt{n} \frac{S_{n, b}\Phi(1 - \alpha) - T_{\zeta, \eta, \kappa_b}(\ve{X}, \ve{Y})}{\sigma_{b1}} \right) \\
        = & \Phi\left(\sqrt{n} \frac{T_{\zeta, \eta, \kappa_b}(\ve{X}, \ve{Y}) - S_{n, b}\Phi(1 - \alpha)}{\sigma_{b1}}\right).
    \end{align*}

   Turning our attention to the HSIC method, which also exhibits asymptotic normality under the alternative hypothesis, the power function, $K_{HSIC}$, is given by:
     \begin{align*}
       K_{HSIC} = & \pr\left(n T_n > q_{1 - \alpha} \right) \\
        = & \pr\left(\sqrt{n} \frac{T_{n} - \gamma}{\sigma_{1}} > \sqrt{n} \frac{\frac{1}{n} q_{1 - \alpha} - \gamma}{\sigma_{1}} \right) \\
        = & \Phi\left(\sqrt{n} \frac{\gamma - \frac{1}{n} q_{1 - \alpha} }{\sigma_{1}} \right).
    \end{align*}
    where $q_{1 - \alpha}$ denotes the HSIC threshold for significance level $1 - \alpha$. Then, we have $$\frac{\sqrt{n} \gamma}{\sigma_1} = \Phi^{-1}(p_{HSIC}) + O_p(\frac{1}{\sqrt{n}}).$$

    In the proof of Proposition~\ref{Prop::var}, 
    % since
    % \begin{align*}
    %    & S_{n, b}^2 \\ = & \frac{1}{n^2} E \left\{ \left[ \zeta(\vesub{X}{i}, \vesub{X}{j}) - E[\zeta(\vesub{X}{i}, \vesub{X}{t}) \mid \vesub{X}{i}] - E[\zeta(\vesub{X}{j}, \vesub{X}{s}) \mid \vesub{X}{j}] + E[\zeta(\vesub{X}{1}, \vesub{X}{2})] \right] \right. \\
    %     &  \cdot \left[ \eta(\vesub{Y}{i}, \vesub{Y}{j}) \kappa_b(Z_i, Z_j) - E[\eta(\vesub{X}{i}, \vesub{X}{t}) \kappa_b(Z_i, Z_t) \mid \vesub{W}{i}] - E[\eta(\vesub{Y}{j}, \vesub{Y}{s}) \kappa_b(Z_j, Z_s) \mid \vesub{W}{j}]  \right. \\
    %     & \left. \left. + E[\eta(\vesub{Y}{1}, \vesub{Y}{2})\kappa_b(Z_1, Z_2)] \right] \right\}.
    % \end{align*}
    we have 
    $$\frac{\sqrt{n}S_{n, b}}{\sigma_{b1}} = O_p(\frac{1}{\sqrt{n b}}).$$

    Combining Lemma~\ref{A3::lem::exp} and Lemma~ \ref{A3::lem::var}, we conclude that:
    \begin{align*}
        & K_{ROLIN} \\ 
        = & \Phi\left(\sqrt{n} \frac{T_{\zeta, \eta, \kappa_b}(\ve{X}, \ve{Y}) - S_{n, b}\Phi(1 - \alpha)}{\sigma_{b1}}\right) \\
        = & \Phi\left(\frac{T_{\zeta, \eta, \kappa_b}(\ve{X}, \ve{Y})\sigma_{1}}{\sigma_{b1} \gamma} \Phi^{-1}(K_{HSIC}) + O_p(\frac{1}{\sqrt{nb}}) \right) \\
        = & \Phi\left(\frac{A_1 \sqrt{\tilde{H}_1 + \tilde{H}_2}  }{\sqrt{A_2  \tilde{H}_1 + A_1^2  \tilde{H}_2}} \Phi^{-1}(K_{HSIC}) + O_p(\frac{1}{\sqrt{nb}}) \right) \\
        = & \Phi\left(\frac{ \sqrt{\tilde{H}_1 + \tilde{H}_2}  }{\sqrt{\frac{A_2}{A_1^2}  \cdot \tilde{H}_1 +   \tilde{H}_2}} \Phi^{-1}(K_{HSIC}) + O_p(\frac{1}{\sqrt{nb}}) \right).
    \end{align*}

    This concludes the proof.
\end{proof}

Now, we turn to find the minimax rate of the proposed statistic.
Let $<f_1, f_2> = \int f_1(x)f_2(x)dx$ be the $L_2$ inner product of $f_1$ and $f_2$.

\begin{lemma}\label{lem::var}
Under the alternative hypothesis and assumption~\ref{ass1}, then
$$
 \var(T_{n, b}) \leq B_0 b^2\left(\frac{1}{n}+\frac{\left\|\zeta\right\|_{\infty}\left\|\eta\right\|_{\infty}}{n^2 b}\right)
$$
is held with a constant $B_0$.
\end{lemma}

\begin{proof}[Proof of Lemma {\rm \ref{lem::var}}]
Based on the variance of U-statistic \cite{lee2019u}, we have
\begin{align*}
    \var(T_{n, b}) = \frac{1}{\CC{n}{4}} \sum_{k=1}^4 \CC{4}{k} \CC{n-4}{4-k} \sigma_{bk}^2
\end{align*}
and we can find a constant $B_1$ such that
\begin{align*}
    \frac{1}{\CC{n}{4}} \CC{4}{k} \CC{n-4}{4-k} =& \frac{4 !(n-4) !}{n !} \CC{4}{k} \frac{(n-4) !}{(n-8+k)(4-k) !} \\
    \le &  \CC{4}{k} \frac{4 !}{(4-k) !} \frac{1}{(n-3)^4} \frac{(n-4) !}{(n-8+k) !} \\
    \le &\CC{4}{k} \frac{4 !}{(4-k) !} \frac{1}{(n-3)^4} (n-3)^{4-k} \le \frac{B_1}{n^k}
\end{align*}
For the law of total variance, we have
$$
\var\left(T_{n, b}\right) \leq B_2 \left(\frac{\sigma_{b1}^2}{n}+\frac{\sigma_{b4}^2}{n^2}\right)
$$
hold with constant $B_2$.

We first consider $\sigma_{b1}$,
 \begin{align*}
      &    E\big[ \Bar{h}_b(\vesub{W}{1}, \vesub{W}{2}, \vesub{W}{3}, \vesub{W}{4}) | \vesub{W}{1} \big] \\
    = & \frac{b}{2} \Big\{ E\big[ \zeta(\vesub{X}{1}, \vesub{X}{2}) \eta(\vesub{Y}{1}, \vesub{Y}{2}) \big] g(Z_1) F(Z_1) +  E\big[ \zeta(\vesub{X}{3}, \vesub{X}{4}) \eta(\vesub{Y}{1}, \vesub{Y}{2}) \big] g(Z_1) F(Z_1) \\
    & -  E\big[ \zeta(\vesub{X}{1}, \vesub{X}{3}) \eta(\vesub{Y}{1}, \vesub{Y}{2}) \big] g(Z_1) F(Z_1) -  E\big[ \zeta(\vesub{X}{3}, \vesub{X}{4}) \eta(\vesub{Y}{1}, \vesub{Y}{3}) \big] g(Z_1) F(Z_1) \\
    & +  E\big[ \zeta(\vesub{X}{1}, \vesub{X}{2}) \eta(\vesub{Y}{3}, \vesub{Y}{4}) \big] A_1 +  E\big[ \zeta(\vesub{X}{3}, \vesub{X}{4}) \eta(\vesub{Y}{3}, \vesub{Y}{4}) \big] A_1 \\
    & -  E\big[ \zeta(\vesub{X}{1}, \vesub{X}{2}) \eta(\vesub{Y}{2}, \vesub{Y}{4}) \big] A_1 -  E\big[ \zeta(\vesub{X}{3}, \vesub{X}{4}) \eta(\vesub{Y}{2}, \vesub{Y}{4}) \big] A_1 \big| \vesub{W}{1}
    \Big\},
    \end{align*}
    then
    \begin{align*}
 \sigma_{b1}^2 &  \leq b^2 B_3\left[M_1+M_2+M_3+M_4+M_5+M_6\right],
 \end{align*}
 where $B_3$ is a constant and
 \begin{align*}
M_1  & =\var\left[E\left[\zeta(\vesub{X}{1}, \vesub{X}{2}) \eta(\vesub{Y}{1}, \vesub{Y}{2})  g(Z_1) F(Z_1) \big| \vesub{W}{1}\right]\right] \\
M_2 & =\var\left[E\left[\zeta(\vesub{X}{3}, \vesub{X}{4}) \eta(\vesub{Y}{1}, \vesub{Y}{2})  g(Z_1) F(Z_1) \big| \vesub{W}{1}\right]\right] \\
M_3 & =\var\left[E\left[\zeta(\vesub{X}{1}, \vesub{X}{3}) \eta(\vesub{Y}{1}, \vesub{Y}{2}) g(Z_1) F(Z_1) \big| \vesub{W}{1}\right]\right] \\
M_4 & =\var\left[E\left[\zeta(\vesub{X}{3}, \vesub{X}{4}) \eta(\vesub{Y}{1}, \vesub{Y}{3})  g(Z_1) F(Z_1) \big| \vesub{W}{1}\right]\right] \\
M_5 & =\var\left[E\left[\zeta(\vesub{X}{1}, \vesub{X}{2}) \eta(\vesub{Y}{3}, \vesub{Y}{4}) A_1 \big| \vesub{W}{1}\right]\right] \\
M_6 & =\var\left[E\left[\zeta(\vesub{X}{1}, \vesub{X}{2}) \eta(\vesub{Y}{2}, \vesub{Y}{3}) A_1 \big| \vesub{W}{1}\right]\right].
\end{align*}
Next, we calculate the expressions of the six terms. For $M_1$, we have
\begin{align*}
M_1 & \leq E\left\{\left[E\left[\zeta(\vesub{X}{1}, \vesub{X}{2}) \eta(\vesub{Y}{1}, \vesub{Y}{2}) g(Z_1) F(Z_1) \big| \vesub{W}{1}\right]\right]^2\right\} \\
& \leq A_2 E\left[\zeta(\vesub{X}{1}, \vesub{X}{2}) \eta(\vesub{Y}{1}, \vesub{Y}{2}) \zeta(\vesub{X}{1}, \vesub{X}{3}) \eta(\vesub{Y}{1}, \vesub{Y}{3})\right] \\
 & \leq A_2 \|f\|_{\infty}^2 \int \zeta(\vesub{X}{1}, \vesub{X}{2}) \eta(\vesub{Y}{1}, \vesub{Y}{2}) \zeta(\vesub{X}{1}, \vesub{X}{3}) \eta(\vesub{Y}{1}, \vesub{Y}{3}) f(\vesub{X}{1}, \vesub{Y}{1}) \prod_{i=1}^3 d \vesub{x}{i} d \vesub{y}{i} \\
 & =A_2 \|f\|_{\infty}^2 \int\left[\int \zeta(\vesub{X}{1}, \ve{x}) d \ve{x}\right]^2\left[\int \zeta(\vesub{Y}{1}, \ve{y}) d \ve{y}\right]^2 f(\vesub{X}{1}, \vesub{Y}{1}) d \vesub{X}{1} d \vesub{Y}{1} \\
 & =A_2\|f\|_{\infty}^2.
\end{align*}
For $M_2$, we have
\begin{align*}
M_2 & \leq E\left\{\left[E\left[\zeta(\vesub{X}{3}, \vesub{X}{4}) \eta(\vesub{Y}{1}, \vesub{Y}{2}) \cdot g(Z_1) F(Z_1) \big| \vesub{W}{1}\right]\right]^2\right\} \\
& \leq A_2\left[E[\zeta(\vesub{X}{3}, \vesub{X}{4})]\right]^2 \cdot E\left\{\left[E\left[\eta(\vesub{Y}{1}, \vesub{Y}{2}) \big| \vesub{Y}{1}\right]\right]^2\right\} \\
& \leq A_2\left[E[\zeta(\vesub{X}{3}, \vesub{X}{4})]\right]^2 \cdot E\left[\eta(\vesub{Y}{1}, \vesub{Y}{2}) \eta(\vesub{Y}{1}, \vesub{Y}{3})\right] \\
& =A_2\left\|f_x\right\|_{\infty}^2 \int \eta(\vesub{Y}{1}, \vesub{Y}{2}) \eta(\vesub{Y}{1}, \vesub{Y}{3}) f_y(\vesub{Y}{1}) f_y(\vesub{Y}{2}) f_y(\vesub{Y}{3}) d \vesub{Y}{1} d \vesub{Y}{2} d \vesub{Y}{3} \\
& =A_2 \left\|f_x\right\|_{\infty}^2  \int\left[\int \eta(\vesub{Y}{1}, \vesub{Y}{2}) f_y(\vesub{Y}{2}) d \vesub{Y}{2}\right]\left[\int \eta(\vesub{Y}{1}, \vesub{Y}{3}) f_y(\vesub{Y}{3}) d \vesub{Y}{3}\right] f_y(\vesub{Y}{1}) d \vesub{Y}{1} \\
& \leq A_3 \left\|f_x\right\|_{\infty}^2 \left\|f_y\right\|_{\infty}^2.
\end{align*}
For $M_3$, we have
\begin{align*}
M_3  \leq & E\left\{\left[E\left[\zeta(\vesub{X}{1}, \vesub{X}{3}) \eta(\vesub{Y}{1}, \vesub{Y}{2})  g(Z_1) F(Z_1) \big| \vesub{W}{1}\right]\right]^2\right\} \\
 \leq & A_2 E\left[\zeta(\vesub{X}{1}, \vesub{X}{3}) \zeta(\vesub{X}{1}, \vesub{X}{4})\eta(\vesub{Y}{1}, \vesub{Y}{2}) \eta(\vesub{Y}{1}, \vesub{Y}{5})\right] \\
 = & A_2 \int \zeta(\vesub{X}{1}, \vesub{X}{3}) \zeta(\vesub{X}{1}, \vesub{X}{4})\eta(\vesub{Y}{1}, \vesub{Y}{2}) \eta(\vesub{Y}{1}, \vesub{Y}{5}) f(\vesub{X}{1}, \vesub{Y}{1}) f_x(\vesub{X}{3}) f_x(\vesub{X}{4}) f_y(\vesub{Y}{2}) \\
& f_y(\vesub{Y}{5}) d \vesub{X}{1} d \vesub{X}{3} d \vesub{X}{4} d \vesub{Y}{1} d \vesub{Y}{2} d \vesub{Y}{5} \\
 = & A_2 \int\left[\int \zeta(\vesub{X}{1}, \vesub{X}{3}) f_x(\vesub{X}{3}) d \vesub{X}{3}\right]\left[\int \zeta(\vesub{X}{1}, \vesub{X}{4})f_x(\vesub{X}{4}) d \vesub{X}{4}\right] \\
& \cdot {\left[\int \eta(\vesub{Y}{1}, \vesub{Y}{2}) f_y(\vesub{Y}{2}) d \vesub{Y}{2}\right] f(\vesub{X}{1}, \vesub{Y}{2}) d \vesub{X}{1} d \vesub{Y}{1} } \\
 \leq & A_2 \left\|f_x\right\|_{\infty}^2 \left\|f_y\right\|_{\infty}^2.
\end{align*}
For $M_4$, we have
\begin{align*}
M_4  \leq &E\left\{\left[E\left[\zeta(\vesub{X}{3}, \vesub{X}{4}) \eta(\vesub{Y}{1}, \vesub{Y}{3}) g(Z_1) F(Z_1) \big| \vesub{W}{1}\right)\right]^2\right\} \\
\leq & A_2 E\left[\zeta(\vesub{X}{3}, \vesub{X}{4})  \zeta(\vesub{X}{2}, \vesub{X}{5}) \eta(\vesub{Y}{1}, \vesub{Y}{3}) \eta(\vesub{Y}{1}, \vesub{Y}{2})\right] \\
= & A_2 \int  \zeta(\vesub{X}{2}, \vesub{X}{5}) \zeta(\vesub{X}{3}, \vesub{X}{4}) \eta(\vesub{Y}{1}, \vesub{Y}{3}) \eta(\vesub{Y}{1}, \vesub{Y}{2}) f_y(\vesub{Y}{1}) f(\vesub{X}{2}, \vesub{Y}{2}) \\
&  f(\vesub{X}{3}, \vesub{Y}{3}) f_x(\vesub{X}{4}) f_x(\vesub{x}{5}) d \vesub{X}{2} d \vesub{X}{3} d \vesub{X}{4} d \vesub{Y}{1} d \vesub{Y}{2} d \vesub{Y}{3} \\
 \leq& A_2\|f\|_{\infty}^2 \int \left[\int  \zeta(\vesub{X}{2}, \vesub{X}{5}) d \vesub{X}{2}\right] \left[\int \zeta(\vesub{X}{3}, \vesub{X}{4}) d \vesub{X}{3}\right] \left[\int \eta(\vesub{Y}{1}, \vesub{Y}{3}) d \vesub{Y}{3}\right] \\
&  \cdot \left[\int \eta(\vesub{Y}{1}, \vesub{Y}{2}) d \vesub{Y}{2}\right] f_y(\vesub{Y}{1}) f_x(\vesub{X}{4}) f_x(\vesub{X}{5}) d \vesub{Y}{1} d \vesub{X}{4} d \vesub{X}{5} \\
= & A_2 \|f\|_{\infty}^2.
\end{align*}
For $M_5$, we have
\begin{align*}
M_5 & \leq E\left\{\left[E\left[\zeta(\vesub{X}{1}, \vesub{X}{2}) \eta(\vesub{Y}{3}, \vesub{Y}{4}) A_1 \big| \vesub{W}{1}\right]\right]^2\right\} \\
& \leq A_1^2 E\left\{\left[E\left[\zeta(\vesub{X}{1}, \vesub{X}{2}) \big| \vesub{X}{2}\right]\right]^2\right\}\left[E[\eta(\vesub{Y}{3}, \vesub{Y}{4})]\right]^2\\
& =A_1^2 E\left[\zeta(\vesub{X}{1}, \vesub{X}{2}) \zeta(\vesub{X}{1}, \vesub{X}{3})\right]\left[E[\eta(\vesub{Y}{3}, \vesub{Y}{4})]\right]^2 \\
& \leq A_1^2\left\|f_x\right\|_{\infty}^2\left\|f_y\right\|_{\infty}^2.
\end{align*}
For $M_6$, similar to $M_4$ we have
\begin{align*}
M_6 & \leq E\left\{\left[E\left[\zeta(\vesub{X}{1}, \vesub{X}{2}) \eta(\vesub{Y}{2}, \vesub{Y}{3}) A_1 \big| \vesub{W}{1}\right)\right]^2\right\} \\
& \leq A_1^2 E\left[\zeta(\vesub{X}{1}, \vesub{X}{2}) \zeta(\vesub{X}{1}, \vesub{X}{4})\eta(\vesub{Y}{2}, \vesub{Y}{3}) \eta(\vesub{Y}{4}, \vesub{Y}{5})\right] \\
& \leq A_1^2\|f\|_{\infty}^2 
\end{align*}

We now turn to $\sigma_{b4}^2$, note that
\begin{align*}
&\sigma_{b4}^2\\
\leq & B_4 \var\left[\zeta(\vesub{X}{1}, \vesub{X}{2})\left[\eta(\vesub{Y}{1}, \vesub{Y}{2}) \kappa\left(Z_1, Z_2\right)+\eta(\vesub{Y}{3}, \vesub{Y}{4}) \kappa\left(Z_3, Z_4\right) -2 \eta(\vesub{Y}{1}, \vesub{Y}{3}) \kappa\left(Z_1, Z_3\right)\right]\right] \\
\leq & B_4\left\{\var\left[\zeta(\vesub{X}{1}, \vesub{X}{2}) \eta(\vesub{Y}{1}, \vesub{Y}{2}) \kappa\left(Z_1, Z_2\right) \right]+\var\left[\zeta(\vesub{X}{1}, \vesub{X}{2}) \eta(\vesub{Y}{3}, \vesub{Y}{4}) \kappa\left(Z_3, Z_4\right)\right] \right. \\
& \left.+\var\left[\zeta(\vesub{X}{1}, \vesub{X}{2}) \eta(\vesub{Y}{1}, \vesub{Y}{3}) \kappa\left(Z_1, Z_3\right)\right]\right\} \\
\leq & B_4 b\left\{ E\left[\zeta^2(\vesub{X}{1}, \vesub{X}{2}) \eta^2(\vesub{Y}{1}, \vesub{Y}{2})\right]+E\left[\zeta^2(\vesub{X}{1}, \vesub{X}{2}) \eta^2(\vesub{Y}{3}, \vesub{Y}{4})\right] +E\left[\zeta^2(\vesub{X}{1}, \vesub{X}{2}) \eta^2(\vesub{Y}{1}, \vesub{Y}{3})\right]\right\},
\end{align*}
where $B_4$ is a constant. For $ E\left[\zeta^2(\vesub{X}{1}, \vesub{X}{2}) \eta^2(\vesub{Y}{1}, \vesub{Y}{2})\right]$
\begin{align*}
& E\left[\zeta^2(\vesub{X}{1}, \vesub{X}{2}) \eta^2(\vesub{Y}{1}, \vesub{Y}{2})\right] \\
\leq& \left\|\zeta\right\|_{\infty} \left\|\eta\right\|_{\infty} E\left[\zeta(\vesub{X}{1}, \vesub{X}{2}) \eta(\vesub{Y}{1}, \vesub{Y}{2})\right] \\
\leq& \left\|\zeta\right\|_{\infty} \left\|\eta\right\|_{\infty}\|f\|_{\infty} \left[\int \zeta(\vesub{X}{1}, \vesub{X}{2}) d \vesub{X}{2}\right]\left[\int \eta(\vesub{Y}{1}, \vesub{Y}{2}) d \vesub{Y}{2}\right] f(\vesub{X}{1}, \vesub{Y}{1}) d \vesub{X}{1} d \vesub{Y}{1} \\
 =&\left\|\zeta\right\|_{\infty}\left\|\eta\right\|_{\infty}\left\|f\right\|_{\infty}
\end{align*}
Similarly,
\begin{align*}
& E\left[\zeta^2(\vesub{X}{1}, \vesub{X}{2}) \eta^2(\vesub{Y}{3}, \vesub{Y}{4})\right] \leq\left\|\zeta\right\|_{\infty}\left\|\eta\right\|_{\infty}\left\|f_x\right\|_{\infty}\left\|f_y\right\|_{\infty} \\
& E\left[\zeta^2(\vesub{X}{1}, \vesub{X}{2}) \eta^2(\vesub{Y}{1}, \vesub{Y}{3})\right] \leq\left\|\zeta\right\|_{\infty}\left\|\eta\right\|_{\infty}\|f\|_{\infty}.
\end{align*}
Thus, $\sigma_{b4}^2 \leq B_4  b \left\|\zeta\right\|_{\infty}\left\|\eta\right\|_{\infty} \left(\left\|{f}\right\|_{\infty}+\left\| f_x\right\|_{\infty}\left\| f_y\right\|_{\infty} \right)$.
Combining the results above, we can find a constant $B_0$ such that
\begin{align*}
 \var\left(T_{n, b}\right) 
& \leq B_0 b^2\left(\frac{1}{n}+\frac{\left\|\zeta\right\|_{\infty}\left\|\eta\right\|_{\infty}}{n^2 b}\right).
\end{align*}
 This concludes the proof.
\end{proof}

\begin{proof}[Proof of Theorem {\rm \ref{thm::minimax1}}]
For Proposition \ref{Prop::var}, 
there exists $N_0 > 0$, when $n > N_0$, $S_{n, b}^2 \leq B_1 b/n^2$ with a constant $B_1$.
Based on Chebyshev's inequality, we have
\begin{align*}
    \pr_{H_1}\left(\frac{\left|T_{n, b}-E (T_{n, b})\right|}{\sqrt{ \var(T_{n, b})}} \geq \frac{1}{\sqrt{\beta}}\right) \leq \beta,
\end{align*}
which implies $\pr_{H_1}\left(T_{n, b} \leq E\left(T_{n, b}\right)-\sqrt{\frac{\var(T_{n, b})}{\beta}}\right) \leq \beta$.

For Lemma \ref{lem::var}, we have  
\begin{align*}
    E\left(T_{n, b}\right) \geq B(\beta) \left(\frac{b}{\sqrt{n}} + \frac{\sqrt{b\left\|\zeta\right\|_{\infty}\left\|\eta\right\|_{\infty}}}{n}\right) \geq \sqrt{\frac{ \var(T_{n, b})}{\beta}}+ \Phi(\alpha) S_{n, b}
\end{align*}
holds for sufficiently large $n$, then
\begin{align*}
   \pr_{H_1}\left(T_{n, b} \leq \Phi(\alpha) S_{n, b}\right) \leq \pr_{H_1}\left(T_{n, b} \leq E\left(T_{n, b}\right)-\sqrt{\frac{\var(T_{n, b})}{\beta}}\right) \leq \beta
\end{align*}
is also held for sufficiently large $n$.

\end{proof}

\begin{lemma}\label{lem::expect}
Under the alternative hypothesis and assumption \ref{ass1}, then
$$
E\left[T_{n, b}\right]=A_1 b\left\langle \psi, \psi *\left(\zeta \otimes \eta\right)\right\rangle
$$
where $*$ denotes the convolution operator concerning the Lebesgue measure.
\end{lemma}

The result in Lemma~\ref{lem::expect} can also be writen as 
$$E\left[T_{n, b}\right]=\frac{A_1 b}{2}\left(\|\psi\|_2^2+\left\|\psi *\left(\zeta \otimes \eta\right)\right\|_2^2-\left\|\psi - \psi *\left(\zeta \otimes \eta\right)\right\|_2^2\right).$$

\begin{proof}[Proof of Lemma {\rm \ref{lem::expect}}]
    \begin{align*}
& E\left[T_{n, b}\right]\\= & E\left[\Bar{h}_b(\vesub{W}{1}, \vesub{W}{2}, \vesub{W}{3}, \vesub{W}{4}) \right] \\
= & \frac{1}{4} E\left[\left(\zeta(\vesub{X}{1}, \vesub{X}{2})+\zeta(\vesub{X}{3}, \vesub{X}{4}) -\zeta(\vesub{X}{1}, \vesub{X}{3}) -\zeta(\vesub{X}{2}, \vesub{X}{4}) \right) \right. \\
& \left.\cdot A_1  b\left(\eta(\vesub{Y}{1}, \vesub{Y}{2}) +\eta(\vesub{Y}{3}, \vesub{Y}{4}) -\eta(\vesub{Y}{1}, \vesub{Y}{3}) -\eta(\vesub{Y}{2}, \vesub{Y}{4}) \right)\right] \\
= & A_1 b \left\{E \left[\zeta(\vesub{X}{1}, \vesub{X}{2})  \eta(\vesub{Y}{1}, \vesub{Y}{2})\right] +E\left[\zeta(\vesub{X}{1}, \vesub{X}{2})  \eta(\vesub{Y}{1}, \vesub{Y}{2}) \right]-2 E\left[\zeta(\vesub{X}{1}, \vesub{X}{2})  \eta(\vesub{Y}{1}, \vesub{Y}{3}) \right]\right\} \\
= & A_1 b\left\{\int \zeta(\vesub{X}{1}, \vesub{X}{2})  \eta(\vesub{Y}{1}, \vesub{Y}{2})  f(\vesub{X}{1}, \vesub{Y}{1})  f(\vesub{X}{2}, \vesub{Y}{2})  d \vesub{X}{1} d \vesub{X}{2} d \vesub{Y}{1} d \vesub{Y}{2} \right. \\
& +\int \zeta(\vesub{X}{1}, \vesub{X}{2})  \eta(\vesub{Y}{1}, \vesub{Y}{2})  f_x(\vesub{X}{1}) f_x(\vesub{X}{2})   f_y(\vesub{Y}{1}) f_y(\vesub{Y}{2})   d \vesub{X}{1} d \vesub{X}{2} d \vesub{Y}{1} d \vesub{Y}{2}  \\
& \left.-\int \zeta(\vesub{X}{1}, \vesub{X}{2})  \eta(\vesub{Y}{1}, \vesub{Y}{2})  f(\vesub{X}{1}, \vesub{Y}{1}) f_x(\vesub{X}{2}) f_y(\vesub{Y}{2}) d \vesub{X}{1} d \vesub{X}{2} d \vesub{Y}{1} d \vesub{Y}{2} \right\} \\
= & A_1 b \int \zeta(\vesub{X}{1}, \vesub{X}{2})  \eta(\vesub{Y}{1}, \vesub{Y}{2}) \left[f(\vesub{X}{1}, \vesub{Y}{1}) -f_x(\vesub{X}{1}) f_y(\vesub{Y}{1})\right] \\
& {\left[f\left(\vesub{X}{2}, \vesub{Y}{2}\right) -f_x(\vesub{X}{2}) f_y(\vesub{Y}{2})\right] d \vesub{X}{1} d \vesub{X}{2} d \vesub{Y}{1} d \vesub{Y}{2}  } \\
= & A_1 b \int \psi(\vesub{X}{1}, \vesub{Y}{1})\left[\int \psi(\vesub{X}{2}, \vesub{Y}{2}) \zeta\left(\vesub{X}{1}-\vesub{X}{2}\right) \eta\left(\vesub{Y}{1}-\vesub{Y}{2}\right) d \vesub{X}{2} d \vesub{Y}{2}\right] d \vesub{X}{1} d \vesub{Y}{1} \\
= & A_1 b \int \psi(\vesub{X}{1}, \vesub{Y}{1})\left[\psi *\left(\zeta \otimes \eta\right)\right](\vesub{X}{1}, \vesub{Y}{1}) d \vesub{X}{1} d \vesub{Y}{1} \\
= & A_1 b\left\langle \psi, \psi *\left(\zeta \otimes \eta\right)\right\rangle
\end{align*}
\end{proof}

\begin{lemma}\label{lem::var2}
Under the alternative hypothesis and assumption \ref{ass1}, then
$$
 \sigma_{b1}^2 \leq B_0^{\prime} b^2 \| \psi *\left(\zeta \otimes \eta\right) \|_2^2
$$
 holds with a constant $B_0^{\prime}$.
\end{lemma}

\begin{proof}[Proof of Lemma {\rm \ref{lem::var2}}]
    For the proof of Lemma \ref{lem::var}, we know
    \begin{align*}
 \sigma_{b1}^2 \leq& B_1 b^2\left\{\var\left[E\left[ \zeta(\vesub{X}{1}, \vesub{X}{2})\left( \eta(\vesub{Y}{1}, \vesub{Y}{2})- \eta(\vesub{Y}{1}, \vesub{Y}{3})\right)  g(Z_1) F(Z_1)\big| \vesub{W}{1}\right]\right]\right. \\
&+\var\left[E\left[ \zeta(\vesub{X}{1}, \vesub{X}{2})\left( \eta(\vesub{Y}{2}, \vesub{Y}{3})- \eta(\vesub{Y}{3}, \vesub{Y}{4}) \right)  A_1\big| \vesub{W}{1}\right]\right. \\
&\left.+\var\left[E\left[\left( \zeta(\vesub{X}{2}, \vesub{X}{3})- \zeta(\vesub{X}{3}, \vesub{X}{4})\right)  \eta(\vesub{Y}{1}, \vesub{Y}{2})  g(Z_1) F(Z_1)\big| \vesub{W}{1}\right]\right]\right\}
\end{align*}
Note that
\begin{align*}
&\left[\psi *\left(\zeta \otimes \eta\right)\right](\vesub{X}{1}, \vesub{Y}{1})\\
= & \int \psi(\vesub{X}{2}, \vesub{Y}{2}) \zeta(\vesub{X}{1}, \vesub{X}{2}) \eta(\vesub{Y}{1}, \vesub{Y}{2}) d\vesub{X}{2} d\vesub{Y}{2} \\
= & \int\left[f(\vesub{X}{2},\vesub{Y}{2})-f_x(\vesub{X}{2}) f_y(\vesub{Y}{2})\right] \zeta(\vesub{X}{1}, \vesub{X}{2}) \eta(\vesub{Y}{1}, \vesub{Y}{2}) d\vesub{X}{2} d\vesub{Y}{2} \\
= & \int \zeta(\vesub{X}{1}, \vesub{X}{2}) \eta(\vesub{Y}{1}, \vesub{Y}{2}) f(\vesub{X}{2},\vesub{Y}{2}) d\vesub{X}{2} d\vesub{Y}{2}-\int \zeta(\vesub{X}{1}, \vesub{X}{2}) \eta(\vesub{Y}{1}, \vesub{Y}{2}) f_x(\vesub{X}{2})  f_y(\vesub{Y}{2}) d\vesub{X}{2} d\vesub{Y}{2} \\
= & \int \zeta(\vesub{X}{1}, \vesub{X}{2})\left(\eta(\vesub{Y}{1}, \vesub{Y}{2})-\eta(\vesub{Y}{2}, \vesub{Y}{3})\right) f(\vesub{X}{2},\vesub{Y}{2}) f_y(\vesup{y}{\prime \prime}) d\vesub{X}{2} d\vesub{Y}{2} d \vesub{Y}{3}\\
= & E\left[\zeta(\vesub{X}{1}, \vesub{X}{2})\left(\eta(\vesub{Y}{1}, \vesub{Y}{2})-\eta(\vesub{Y}{1}, \vesub{Y}{3})\right) \mid \vesub{X}{1}, \vesub{Y}{1} \right].
\end{align*}

Let $\Psi(\ve{x}, \ve{y}) = \left[\psi *\left(\zeta \otimes \eta\right)\right](\ve{x}, \ve{y})$, then  
\begin{align*}
& E\left[ \zeta(\vesub{X}{1}, \vesub{X}{2})\left( \eta(\vesub{Y}{1}, \vesub{Y}{2})-\eta(\vesub{Y}{1}, \vesub{Y}{3})\right) \big| \vesub{W}{1}\right]=\Psi(\vesub{X}{1}, \vesub{Y}{1}) \\
& E\left[ \zeta(\vesub{X}{1}, \vesub{X}{2})\left( \eta(\vesub{Y}{2}, \vesub{Y}{3})- \eta(\vesub{Y}{3}, \vesub{Y}{4})\right) \big| \vesub{W}{1}\right]=E\left[\Psi(\vesub{X}{1}, \vesub{Y}{3}) \big| \vesub{X}{1}\right] \\
& E\left[\left( \zeta(\vesub{X}{2}, \vesub{X}{3})- \zeta(\vesub{X}{3}, \vesub{X}{4})\right)  \eta(\vesub{Y}{1}, \vesub{Y}{2}) \big| \vesub{W}{1}\right]=E\left[\Psi(\vesub{X}{3}, \vesub{Y}{1}) \big| \vesub{Y}{1}\right]
\end{align*}

For the law of total variance, we have
$$\sigma_{b1}^2 \leq B_1 b^2\left\{A_2 \var\left(\Psi(\vesub{X}{1}, \vesub{Y}{1})\right)+A_1^2  \var\left(\Psi(\vesub{X}{1}, \vesub{Y}{3})\right)+A_2 \var\left(\Psi(\vesub{X}{3}, \vesub{Y}{1})\right)\right\}$$
We can bound the three terms by
\begin{align*}
& \var\left(\Psi(\vesub{X}{1}, \vesub{Y}{1})\right) \leq\|f\|_{\infty}\|G\|_2^2 \\
& \var\left(\Psi(\vesub{X}{1}, \vesub{Y}{3})\right) \leq\left\|f_x \otimes f_y\right\|_{\infty}\|\Psi\|_2^2 \\
& \var\left(\Psi(\vesub{X}{3}, \vesub{Y}{1})\right) \leq\left\|f_x \otimes f_y\right\|_{\infty}\|\Psi\|_2^2.
\end{align*}
Thus,
$$\sigma_{b1}^2 \leq B_0^{\prime} b^2 \| \psi *\left(\zeta \otimes \eta\right) \|_2^2,$$
is hold with a constant $B_0^{\prime}$.
\end{proof}

\begin{proof}[Proof of Theorem {\rm \ref{thm::minimax2}}]
For Lemma \ref{lem::var2}, we know there exists a constant $B_0^{\prime}$ such that
$$\sigma_{b1}^2 \leq B_0^{\prime} b^2 \| \psi *\left(\zeta \otimes \eta\right) \|_2^2,$$
Similar to the proof of Lemma \ref{lem::var}, we have
$$\var(T_{n, b}) \leq \frac{B_1^{\prime} b^2}{n}\left(\| \psi *\left(\zeta \otimes \eta\right) \|_2^2 + \frac{1}{nb} \|\zeta \|_{\infty} \|\eta\|_{\infty} \right),$$
where $B_1^{\prime}$ is a constant.
Since $\sqrt{a+b} \leq \sqrt{a} + \sqrt{b}$ and $2\sqrt{ab} \leq c a + \frac{b}{c}$ for $c > 0$, then
\begin{align*}
    2\sqrt{\frac{\var(T_{n, b})}{\beta}} &\leq 2 \sqrt{\frac{B_1^{\prime} b^2}{n \beta} \| \psi *\left(\zeta \otimes \eta\right) \|_2^2} + 2 \sqrt{\frac{B_1^{\prime} b}{n^2 \beta}\|\zeta \|_{\infty} \|\eta\|_{\infty}} \\
    & \leq A_2 b  \| \psi *\left(\zeta \otimes \eta\right) \|_2^2 + \frac{B_2^{\prime} b}{n} + \frac{B_2^{\prime} \sqrt{b}}{n} \sqrt{\|\zeta \|_{\infty} \|\eta\|_{\infty}},
\end{align*}
is hold with constant $B_2^{\prime}$. For Lemma \ref{lem::expect}, we have
$$
E\left[T_{n, b}\right]=\frac{A_1 b}{2}\left(\|\psi\|_2^2+\left\|\psi *\left(\zeta \otimes \eta\right)\right\|_2^2-\left\|\psi-\psi *\left(\zeta \otimes \eta\right)\right\|_2^2\right).
$$
We modify the condition in Theorem \ref{thm::minimax1}, 
\begin{align*}
    &\frac{A_1 b}{2}\left(\|\psi\|_2^2+\left\|\psi *\left(\zeta \otimes \eta\right)\right\|_2^2-\left\|\psi-\psi *\left(\zeta \otimes \eta\right)\right\|_2^2\right) \\
    \geq & \frac{A_1}{2} b  \| \psi *\left(\zeta \otimes \eta\right) \|_2^2 + \frac{B_2^{\prime} b}{2n} + \frac{B_2^{\prime} \sqrt{b}}{2n} \sqrt{\|\zeta \|_{\infty} \|\eta\|_{\infty}} + \frac{B_2^{\prime} \sqrt{b}}{n},
\end{align*}
which implies
$$\|\psi\|_2^2 \geq \|\psi - \psi*\left(\zeta \otimes \eta\right)\|_2^2 + \frac{B(\beta)}{n\sqrt{b}} \sqrt{\|\zeta \|_{\infty} \|\eta\|_{\infty}}.$$
\end{proof}

\section{Additional simulations}

In this section, we examine two additional scenarios. The initial four examples demonstrate the dependency relationship, wherein both $X$ and $Y$ are categorical variables. Subsequent examples elucidate the dependency interactions between one continuous and one discrete variable. The results are summarized in Figure~\ref{example1_4}, \ref{example5_8}.

\begin{enumerate}
    \item[\BLUE{Example 1:}] $X$ is sampled from the binomial distribution $B(1, 0.5)$, $\epsilon_1$ and $\epsilon_2$ are independently drawn from binomial distributions $B(2, 0.5)$ and $B(3, 0.5)$, respectively. We have
    $$Y = X * \epsilon_1 + (1 - X) * \epsilon_2.$$

    \item[\BLUE{Example 2:}] $X$ and $\epsilon$ are independently derived from binomial distribution $B(2, 0.5)$ and $B(10, 0.5)$, respectively. $Y$ is given by
    $$Y = \sqrt{X} + \epsilon.$$

    \item[\BLUE{Example 3:}] $X$ and $\epsilon$ are independently drawn from Poisson distributions, $Pois(1)$ and $Pois(2)$, respectively. Let
    $$Y = 0.2*(X + X^2) + \epsilon.$$

    \item[\BLUE{Example 4:}] $X$ follows a Poisson distribution $Pois(1)$, $\epsilon$ is sampled independently from a binomial distribution $B(3, 0.5)$. $Y$ is calculated as 
    $$Y = X - \epsilon^2.$$

    \item[\BLUE{Example 5:}] $X$ is obtained from a binomial distribution $B(1, 0.5)$, and $Y$ is modeled as a normal distribution with a mean equal to $X$ and an identity variance, expressed as 
    $$Y \sim N(X, I_p).$$

    \item[\BLUE{Example 6:}] $X$ is sampled from the binomial distribution $B(3, 0.5)$, $\epsilon$ is derived from the standard normal distribution. Let 
    $$Y = X * \epsilon.$$
    
    \item[\BLUE{Example 7:}] $X$ is sourced from a binomial distribution $B(1, 0.5)$, $\epsilon_1$ and $\epsilon_2$ are independently from chi-square distributions $\chi(3)$ and $\chi(5/3)$, respectively. We have
    $$Y = X * \epsilon_1 + (1 - X) * \epsilon_2.$$
    
    \item[\BLUE{Example 8:}] $X$ originates from a Poisson distribution $pois(1)$ and $\epsilon_1$ is derived from a Student-t distribution with 2 degrees of freedom,
    and $\epsilon_2$ is independently drawn from a centered normal distribution with a variance of 2. Then
    $$Y = X * \epsilon_1 + \epsilon_2.$$
\end{enumerate}

\begin{figure}[h!]
    \centerline{\includegraphics[width = 1.0\textwidth, height = 0.67\textwidth]{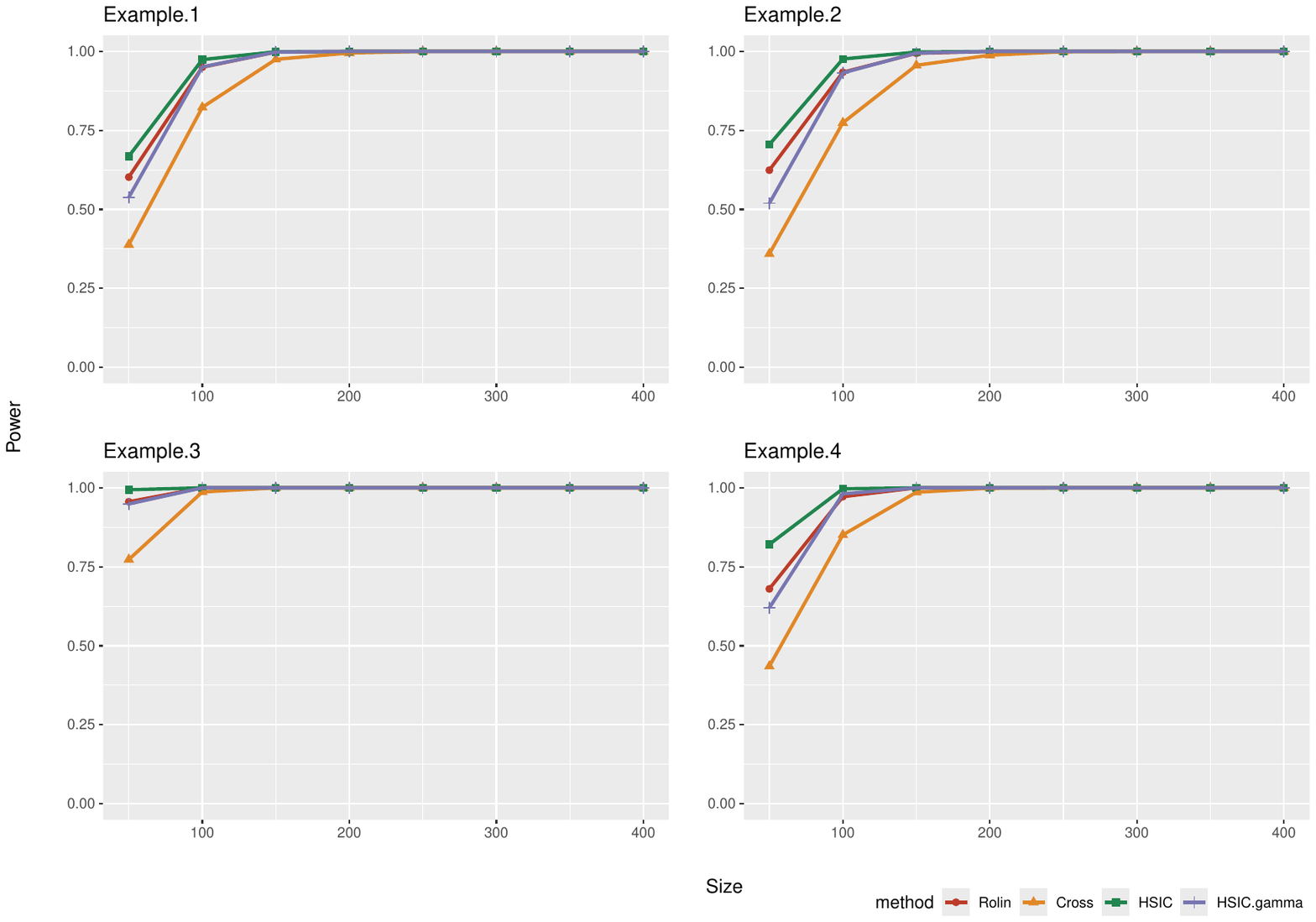}}
    \caption{The power performance in Examples 1-4.}
    \label{example1_4}
\end{figure}

\begin{figure}[h!]
    \centerline{\includegraphics[width = 1.0\textwidth, height = 0.67\textwidth]{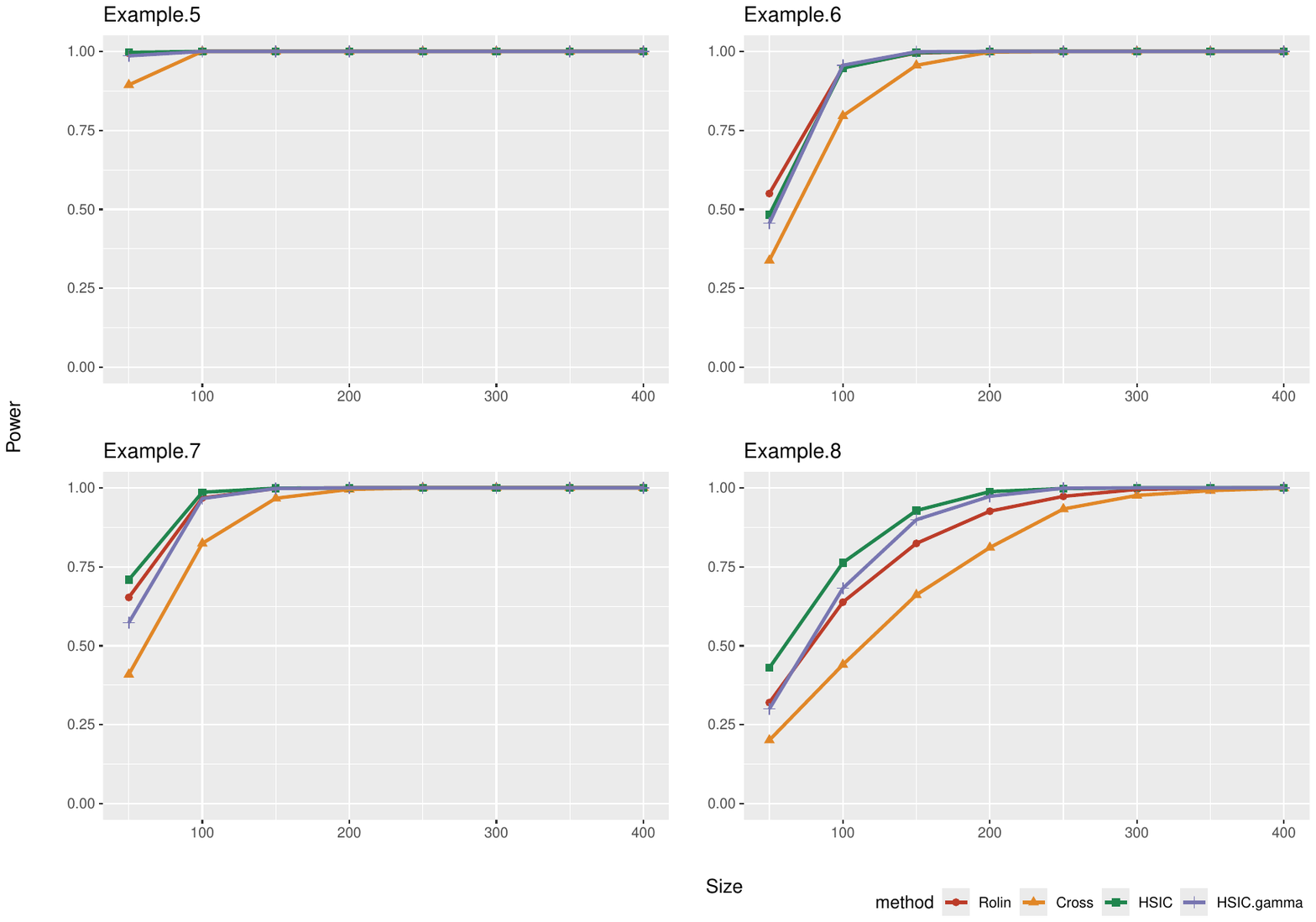}}
    \caption{The power performance in Examples 5-8.}
    \label{example5_8}
\end{figure}

Below, we provide four additional examples illustrating the asymptotic distribution of the proposed statistic under the null hypothesis in discrete cases.

\begin{enumerate}
    \item[\BLUE{Example 9:}] $X$ and $Y$ are both discrete variable, where
    \begin{align*}
        X  \sim B(5, 0.5), \quad
        Y  \sim B(10, 0.5).
    \end{align*}

    \item[\BLUE{Example 10:}] $X$ and $Y$ are both discrete variable, where
    \begin{align*}
        X  \sim B(6, 0.5), \quad
        Y  \sim pois(2).
    \end{align*}

    \item[\BLUE{Example 11:}] $X$ is discrete and $Y$ is continuous. That is 
    \begin{align*}
        X  \sim B(5, 0.5), \quad
        Y  \sim N(0, 1).
    \end{align*}

    \item[\BLUE{Example 12:}] $X$ is discrete and $Y$ is continuous. That is 
    \begin{align*}
        X  \sim pois(2.5), \quad
        Y  \sim N(0, 1).
    \end{align*}   
\end{enumerate}
The results are shown in Figure \ref{example9_12}.

\begin{figure}[h!]
\centerline{\includegraphics[width = 1.0\textwidth, height = 0.67\textwidth]{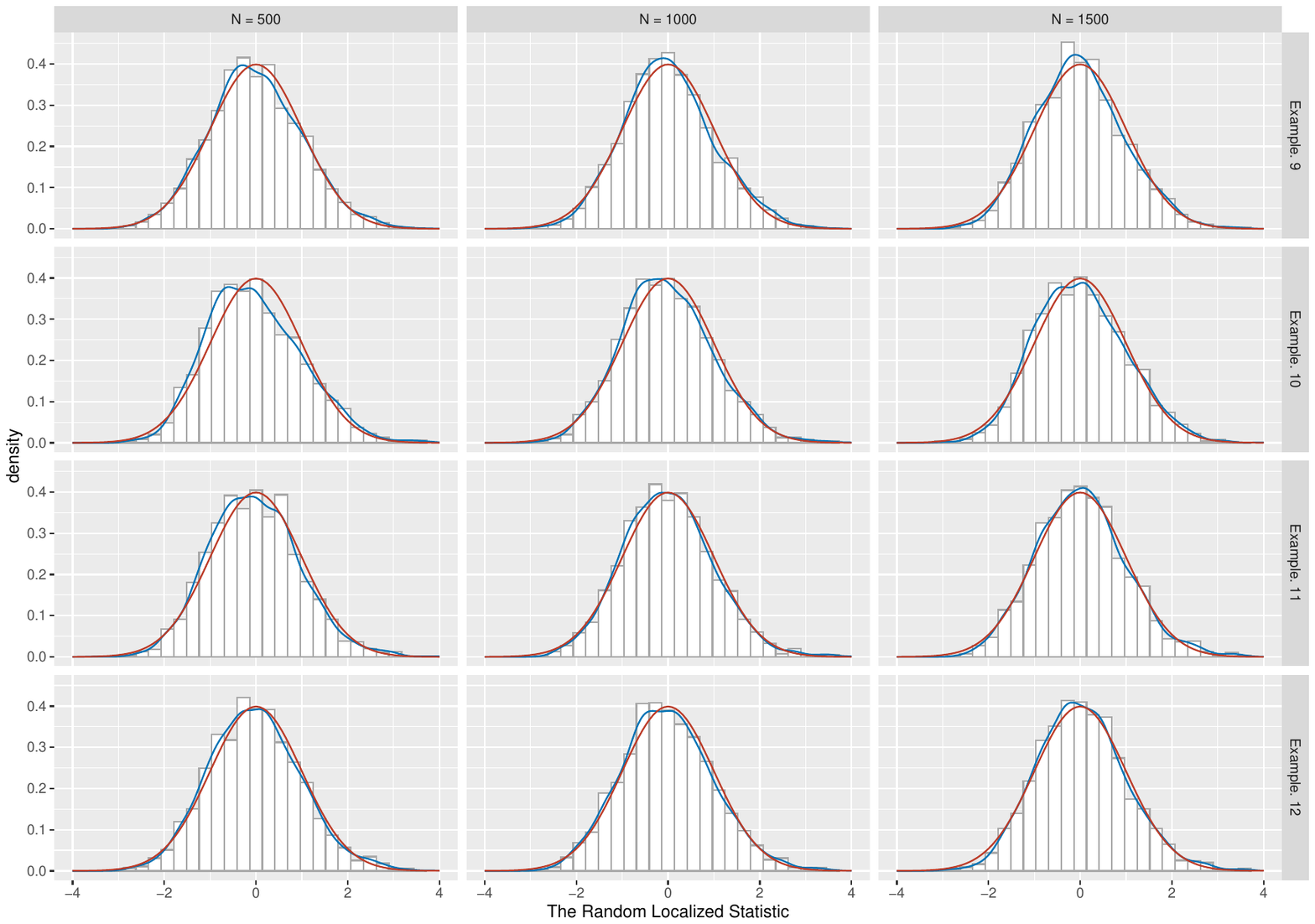}}
    \caption{The histogram and density plot in Examples 9-12.}
    \label{example9_12}
\end{figure}
\end{appendix}
\bibliographystyle{imsart-nameyear}
\bibliography{aos-RL.bib}

\begin{thebibliography}{34}
% BibTex style file: imsart-nameyear.bst, 2017-11-03
% Default style options (sort=1,type=nameyear).
% Used options (sort=1,type=nameyear).

\bibitem[\protect\citeauthoryear{Albert et~al.}{2022}]{Melisande2022}
\begin{barticle}[author]
\bauthor{\bsnm{Albert},~\bfnm{M{\'e}lisande}\binits{M.}},
  \bauthor{\bsnm{Laurent},~\bfnm{B{\'e}atrice}\binits{B.}},
  \bauthor{\bsnm{Marrel},~\bfnm{Amandine}\binits{A.}} \AND
  \bauthor{\bsnm{Meynaoui},~\bfnm{Anouar}\binits{A.}}
(\byear{2022}).
\btitle{{Adaptive test of independence based on HSIC measures}}.
\bjournal{The Annals of Statistics}
\bvolume{50}
\bpages{858 -- 879}.
\bdoi{10.1214/21-AOS2129}
\end{barticle}
\endbibitem

\bibitem[\protect\citeauthoryear{Deb and Sen}{2023}]{deb2023}
\begin{barticle}[author]
\bauthor{\bsnm{Deb},~\bfnm{Nabarun}\binits{N.}} \AND
  \bauthor{\bsnm{Sen},~\bfnm{Bodhisattva}\binits{B.}}
(\byear{2023}).
\btitle{Multivariate Rank-Based Distribution-Free Nonparametric Testing Using
  Measure Transportation}.
\bjournal{Journal of the American Statistical Association}
\bvolume{118}
\bpages{192--207}.
\bdoi{10.1080/01621459.2021.1923508}
\end{barticle}
\endbibitem

\bibitem[\protect\citeauthoryear{Fan and Li}{1996}]{fan1996}
\begin{barticle}[author]
\bauthor{\bsnm{Fan},~\bfnm{Yanqin}\binits{Y.}} \AND
  \bauthor{\bsnm{Li},~\bfnm{Qi}\binits{Q.}}
(\byear{1996}).
\btitle{Consistent Model Specification Tests: Omitted Variables and
  Semiparametric Functional Forms}.
\bjournal{Econometrica}
\bvolume{64}
\bpages{865--890}.
\end{barticle}
\endbibitem

\bibitem[\protect\citeauthoryear{Fukumizu et~al.}{2007}]{NIPS2007_Fukumizu}
\begin{binproceedings}[author]
\bauthor{\bsnm{Fukumizu},~\bfnm{Kenji}\binits{K.}},
  \bauthor{\bsnm{Gretton},~\bfnm{Arthur}\binits{A.}},
  \bauthor{\bsnm{Sun},~\bfnm{Xiaohai}\binits{X.}} \AND
  \bauthor{\bsnm{Sch\"{o}lkopf},~\bfnm{Bernhard}\binits{B.}}
(\byear{2007}).
\btitle{Kernel Measures of Conditional Dependence}.
In \bbooktitle{Advances in Neural Information Processing Systems}
(\beditor{\bfnm{J.}\binits{J.}~\bsnm{Platt}},
  \beditor{\bfnm{D.}\binits{D.}~\bsnm{Koller}},
  \beditor{\bfnm{Y.}\binits{Y.}~\bsnm{Singer}} \AND
  \beditor{\bfnm{S.}\binits{S.}~\bsnm{Roweis}}, eds.)
\bvolume{20}.
\bpublisher{Curran Associates, Inc.}
\end{binproceedings}
\endbibitem

\bibitem[\protect\citeauthoryear{Gao and Shao}{2023}]{gao2023two}
\begin{barticle}[author]
\bauthor{\bsnm{Gao},~\bfnm{Hanjia}\binits{H.}} \AND
  \bauthor{\bsnm{Shao},~\bfnm{Xiaofeng}\binits{X.}}
(\byear{2023}).
\btitle{Two sample testing in high dimension via maximum mean discrepancy}.
\bjournal{Journal of Machine Learning Research}
\bvolume{24}
\bpages{1--33}.
\end{barticle}
\endbibitem

\bibitem[\protect\citeauthoryear{Gao et~al.}{2021}]{Gao2021}
\begin{barticle}[author]
\bauthor{\bsnm{Gao},~\bfnm{Lan}\binits{L.}},
  \bauthor{\bsnm{Fan},~\bfnm{Yingying}\binits{Y.}},
  \bauthor{\bsnm{Lv},~\bfnm{Jinchi}\binits{J.}} \AND
  \bauthor{\bsnm{Shao},~\bfnm{Qi-Man}\binits{Q.-M.}}
(\byear{2021}).
\btitle{{Asymptotic distributions of high-dimensional distance correlation
  inference}}.
\bjournal{The Annals of Statistics}
\bvolume{49}
\bpages{1999 -- 2020}.
\bdoi{10.1214/20-AOS2024}
\end{barticle}
\endbibitem

\bibitem[\protect\citeauthoryear{Gretton et~al.}{2005a}]{gretton2005measuring}
\begin{binproceedings}[author]
\bauthor{\bsnm{Gretton},~\bfnm{Arthur}\binits{A.}},
  \bauthor{\bsnm{Bousquet},~\bfnm{Olivier}\binits{O.}},
  \bauthor{\bsnm{Smola},~\bfnm{Alex}\binits{A.}} \AND
  \bauthor{\bsnm{Sch{\"o}lkopf},~\bfnm{Bernhard}\binits{B.}}
(\byear{2005}a).
\btitle{Measuring statistical dependence with Hilbert-Schmidt norms}.
In \bbooktitle{International conference on algorithmic learning theory}
\bpages{63--77}.
\bpublisher{Springer}.
\end{binproceedings}
\endbibitem

\bibitem[\protect\citeauthoryear{Gretton et~al.}{2005b}]{gretton2005kernel}
\begin{barticle}[author]
\bauthor{\bsnm{Gretton},~\bfnm{Arthur}\binits{A.}},
  \bauthor{\bsnm{Herbrich},~\bfnm{Ralf}\binits{R.}},
  \bauthor{\bsnm{Smola},~\bfnm{Alexander}\binits{A.}},
  \bauthor{\bsnm{Bousquet},~\bfnm{Olivier}\binits{O.}} \AND
  \bauthor{\bsnm{Sch{\"o}lkopf},~\bfnm{Bernhard}\binits{B.}}
(\byear{2005}b).
\btitle{Kernel Methods for Measuring Independence}.
\bjournal{Journal of Machine Learning Research}
\bvolume{6}
\bpages{2075--2129}.
\end{barticle}
\endbibitem

\bibitem[\protect\citeauthoryear{Gretton et~al.}{2005c}]{pmlr-vR5-gretton05a}
\begin{binproceedings}[author]
\bauthor{\bsnm{Gretton},~\bfnm{Arthur}\binits{A.}},
  \bauthor{\bsnm{Smola},~\bfnm{Alexander}\binits{A.}},
  \bauthor{\bsnm{Bousquet},~\bfnm{Olivier}\binits{O.}},
  \bauthor{\bsnm{Herbrich},~\bfnm{Ralf}\binits{R.}},
  \bauthor{\bsnm{Belitski},~\bfnm{Andrei}\binits{A.}},
  \bauthor{\bsnm{Augath},~\bfnm{Mark}\binits{M.}},
  \bauthor{\bsnm{Murayama},~\bfnm{Yusuke}\binits{Y.}},
  \bauthor{\bsnm{Pauls},~\bfnm{Jon}\binits{J.}},
  \bauthor{\bsnm{Sch\"olkopf},~\bfnm{Bernhard}\binits{B.}} \AND
  \bauthor{\bsnm{Logothetis},~\bfnm{Nikos}\binits{N.}}
(\byear{2005}c).
\btitle{Kernel Constrained Covariance for Dependence Measurement}.
In \bbooktitle{Proceedings of the Tenth International Workshop on Artificial
  Intelligence and Statistics}
(\beditor{\bfnm{Robert~G.}\binits{R.~G.}~\bsnm{Cowell}} \AND
  \beditor{\bfnm{Zoubin}\binits{Z.}~\bsnm{Ghahramani}}, eds.).
\bseries{Proceedings of Machine Learning Research}
\bvolume{R5}
\bpages{112--119}.
\bpublisher{PMLR}
\bnote{Reissued by PMLR on 30 March 2021.}
\end{binproceedings}
\endbibitem

\bibitem[\protect\citeauthoryear{Gretton et~al.}{2007}]{gretton2007kernel}
\begin{barticle}[author]
\bauthor{\bsnm{Gretton},~\bfnm{Arthur}\binits{A.}},
  \bauthor{\bsnm{Fukumizu},~\bfnm{Kenji}\binits{K.}},
  \bauthor{\bsnm{Teo},~\bfnm{Choon}\binits{C.}},
  \bauthor{\bsnm{Song},~\bfnm{Le}\binits{L.}},
  \bauthor{\bsnm{Sch{\"o}lkopf},~\bfnm{Bernhard}\binits{B.}} \AND
  \bauthor{\bsnm{Smola},~\bfnm{Alex}\binits{A.}}
(\byear{2007}).
\btitle{A kernel statistical test of independence}.
\bjournal{Advances in neural information processing systems}
\bvolume{20}.
\end{barticle}
\endbibitem

\bibitem[\protect\citeauthoryear{Hongjian~Shi and Han}{2022}]{shi2022}
\begin{barticle}[author]
\bauthor{\bsnm{Hongjian~Shi},~\bfnm{Mathias~Drton}\binits{M.~D.}} \AND
  \bauthor{\bsnm{Han},~\bfnm{Fang}\binits{F.}}
(\byear{2022}).
\btitle{Distribution-Free Consistent Independence Tests via Center-Outward
  Ranks and Signs}.
\bjournal{Journal of the American Statistical Association}
\bvolume{117}
\bpages{395--410}.
\bdoi{10.1080/01621459.2020.1782223}
\end{barticle}
\endbibitem

\bibitem[\protect\citeauthoryear{Kim, Balakrishnan and
  Wasserman}{2020}]{kim2020}
\begin{barticle}[author]
\bauthor{\bsnm{Kim},~\bfnm{Ilmun}\binits{I.}},
  \bauthor{\bsnm{Balakrishnan},~\bfnm{Sivaraman}\binits{S.}} \AND
  \bauthor{\bsnm{Wasserman},~\bfnm{Larry}\binits{L.}}
(\byear{2020}).
\btitle{{Robust multivariate nonparametric tests via projection averaging}}.
\bjournal{The Annals of Statistics}
\bvolume{48}
\bpages{3417 -- 3441}.
\bdoi{10.1214/19-AOS1936}
\end{barticle}
\endbibitem

\bibitem[\protect\citeauthoryear{Kim, Balakrishnan and
  Wasserman}{2022}]{kim2022}
\begin{barticle}[author]
\bauthor{\bsnm{Kim},~\bfnm{Ilmun}\binits{I.}},
  \bauthor{\bsnm{Balakrishnan},~\bfnm{Sivaraman}\binits{S.}} \AND
  \bauthor{\bsnm{Wasserman},~\bfnm{Larry}\binits{L.}}
(\byear{2022}).
\btitle{{Minimax optimality of permutation tests}}.
\bjournal{The Annals of Statistics}
\bvolume{50}
\bpages{225 -- 251}.
\bdoi{10.1214/21-AOS2103}
\end{barticle}
\endbibitem

\bibitem[\protect\citeauthoryear{Lee}{2019}]{lee2019u}
\begin{bbook}[author]
\bauthor{\bsnm{Lee},~\bfnm{A~J}\binits{A.~J.}}
(\byear{2019}).
\btitle{U-statistics: Theory and Practice}.
\bpublisher{Routledge}, \baddress{New York}.
\end{bbook}
\endbibitem

\bibitem[\protect\citeauthoryear{Li and Yuan}{2019}]{li2019optimality}
\begin{bmisc}[author]
\bauthor{\bsnm{Li},~\bfnm{Tong}\binits{T.}} \AND
  \bauthor{\bsnm{Yuan},~\bfnm{Ming}\binits{M.}}
(\byear{2019}).
\btitle{On the Optimality of Gaussian Kernel Based Nonparametric Tests against
  Smooth Alternatives}.
\end{bmisc}
\endbibitem

\bibitem[\protect\citeauthoryear{Lin and Bai}{2011}]{lin2011probability}
\begin{bbook}[author]
\bauthor{\bsnm{Lin},~\bfnm{Zhengyan}\binits{Z.}} \AND
  \bauthor{\bsnm{Bai},~\bfnm{Zhidong}\binits{Z.}}
(\byear{2011}).
\btitle{Probability inequalities}.
\bpublisher{Springer Science \& Business Media}.
\end{bbook}
\endbibitem

\bibitem[\protect\citeauthoryear{Lopez-Paz, Hennig and
  Sch{\"o}lkopf}{2013}]{lopez2013randomized}
\begin{barticle}[author]
\bauthor{\bsnm{Lopez-Paz},~\bfnm{David}\binits{D.}},
  \bauthor{\bsnm{Hennig},~\bfnm{Philipp}\binits{P.}} \AND
  \bauthor{\bsnm{Sch{\"o}lkopf},~\bfnm{Bernhard}\binits{B.}}
(\byear{2013}).
\btitle{The randomized dependence coefficient}.
\bjournal{Advances in neural information processing systems}
\bvolume{26}.
\end{barticle}
\endbibitem

\bibitem[\protect\citeauthoryear{Lyons}{2013}]{lyons2013}
\begin{barticle}[author]
\bauthor{\bsnm{Lyons},~\bfnm{Russell}\binits{R.}}
(\byear{2013}).
\btitle{{Distance covariance in metric spaces}}.
\bjournal{The Annals of Probability}
\bvolume{41}
\bpages{3284 -- 3305}.
\bdoi{10.1214/12-AOP803}
\end{barticle}
\endbibitem

\bibitem[\protect\citeauthoryear{Moignard et~al.}{2015}]{moignard2015decoding}
\begin{barticle}[author]
\bauthor{\bsnm{Moignard},~\bfnm{Victoria}\binits{V.}},
  \bauthor{\bsnm{Woodhouse},~\bfnm{Steven}\binits{S.}},
  \bauthor{\bsnm{Haghverdi},~\bfnm{Laleh}\binits{L.}},
  \bauthor{\bsnm{Lilly},~\bfnm{Andrew~J}\binits{A.~J.}},
  \bauthor{\bsnm{Tanaka},~\bfnm{Yosuke}\binits{Y.}},
  \bauthor{\bsnm{Wilkinson},~\bfnm{Adam~C}\binits{A.~C.}},
  \bauthor{\bsnm{Buettner},~\bfnm{Florian}\binits{F.}},
  \bauthor{\bsnm{Macaulay},~\bfnm{Iain~C}\binits{I.~C.}},
  \bauthor{\bsnm{Jawaid},~\bfnm{Wajid}\binits{W.}},
  \bauthor{\bsnm{Diamanti},~\bfnm{Evangelia}\binits{E.}} \betal{et~al.}
(\byear{2015}).
\btitle{Decoding the regulatory network of early blood development from
  single-cell gene expression measurements}.
\bjournal{Nature biotechnology}
\bvolume{33}
\bpages{269--276}.
\end{barticle}
\endbibitem

\bibitem[\protect\citeauthoryear{Pan et~al.}{2019}]{pan2019ball}
\begin{barticle}[author]
\bauthor{\bsnm{Pan},~\bfnm{Wenliang}\binits{W.}},
  \bauthor{\bsnm{Wang},~\bfnm{Xueqin}\binits{X.}},
  \bauthor{\bsnm{Zhang},~\bfnm{Heping}\binits{H.}},
  \bauthor{\bsnm{Zhu},~\bfnm{Hongtu}\binits{H.}} \AND
  \bauthor{\bsnm{Zhu},~\bfnm{Jin}\binits{J.}}
(\byear{2019}).
\btitle{Ball covariance: A generic measure of dependence in Banach space}.
\bjournal{Journal of the American Statistical Association}
\bvolume{115}
\bpages{307--317}.
\end{barticle}
\endbibitem

\bibitem[\protect\citeauthoryear{Pfister et~al.}{2018}]{pfister2018kernel}
\begin{barticle}[author]
\bauthor{\bsnm{Pfister},~\bfnm{Niklas}\binits{N.}},
  \bauthor{\bsnm{B{\"u}hlmann},~\bfnm{Peter}\binits{P.}},
  \bauthor{\bsnm{Sch{\"o}lkopf},~\bfnm{Bernhard}\binits{B.}} \AND
  \bauthor{\bsnm{Peters},~\bfnm{Jonas}\binits{J.}}
(\byear{2018}).
\btitle{Kernel-based tests for joint independence}.
\bjournal{Journal of the Royal Statistical Society: Series B (Statistical
  Methodology)}
\bvolume{80}
\bpages{5--31}.
\end{barticle}
\endbibitem

\bibitem[\protect\citeauthoryear{Rigollet and
  Tsybakov}{2007}]{rigollet2007linear}
\begin{barticle}[author]
\bauthor{\bsnm{Rigollet},~\bfnm{Ph}\binits{P.}} \AND
  \bauthor{\bsnm{Tsybakov},~\bfnm{Alexander~B}\binits{A.~B.}}
(\byear{2007}).
\btitle{Linear and convex aggregation of density estimators}.
\bjournal{Mathematical Methods of Statistics}
\bvolume{16}
\bpages{260--280}.
\end{barticle}
\endbibitem

\bibitem[\protect\citeauthoryear{Schweizer and
  Wolff}{1981}]{schweizer1981nonparametric}
\begin{barticle}[author]
\bauthor{\bsnm{Schweizer},~\bfnm{Berthold}\binits{B.}} \AND
  \bauthor{\bsnm{Wolff},~\bfnm{Edward~F}\binits{E.~F.}}
(\byear{1981}).
\btitle{On nonparametric measures of dependence for random variables}.
\bjournal{The annals of statistics}
\bvolume{9}
\bpages{879--885}.
\end{barticle}
\endbibitem

\bibitem[\protect\citeauthoryear{Scott}{1979}]{scott1979optimal}
\begin{barticle}[author]
\bauthor{\bsnm{Scott},~\bfnm{David~W}\binits{D.~W.}}
(\byear{1979}).
\btitle{On optimal and data-based histograms}.
\bjournal{Biometrika}
\bvolume{66}
\bpages{605--610}.
\end{barticle}
\endbibitem

\bibitem[\protect\citeauthoryear{Shekhar, Kim and
  Ramdas}{2023}]{shekhar2023permutation}
\begin{barticle}[author]
\bauthor{\bsnm{Shekhar},~\bfnm{Shubhanshu}\binits{S.}},
  \bauthor{\bsnm{Kim},~\bfnm{Ilmun}\binits{I.}} \AND
  \bauthor{\bsnm{Ramdas},~\bfnm{Aaditya}\binits{A.}}
(\byear{2023}).
\btitle{A permutation-free kernel independence test}.
\bjournal{Journal of Machine Learning Research}
\bvolume{24}
\bpages{1--68}.
\end{barticle}
\endbibitem

\bibitem[\protect\citeauthoryear{Silverman}{2018}]{silverman2018density}
\begin{bbook}[author]
\bauthor{\bsnm{Silverman},~\bfnm{Bernard~W}\binits{B.~W.}}
(\byear{2018}).
\btitle{Density estimation for statistics and data analysis}.
\bpublisher{Routledge}.
\end{bbook}
\endbibitem

\bibitem[\protect\citeauthoryear{Sriperumbudur
  et~al.}{2010}]{JMLR:v11:sriperumbudur10a}
\begin{barticle}[author]
\bauthor{\bsnm{Sriperumbudur},~\bfnm{Bharath~K.}\binits{B.~K.}},
  \bauthor{\bsnm{Gretton},~\bfnm{Arthur}\binits{A.}},
  \bauthor{\bsnm{Fukumizu},~\bfnm{Kenji}\binits{K.}},
  \bauthor{\bsnm{Sch{{\"o}}lkopf},~\bfnm{Bernhard}\binits{B.}} \AND
  \bauthor{\bsnm{Lanckriet},~\bfnm{Gert R.~G.}\binits{G.~R.~G.}}
(\byear{2010}).
\btitle{Hilbert Space Embeddings and Metrics on Probability Measures}.
\bjournal{Journal of Machine Learning Research}
\bvolume{11}
\bpages{1517-1561}.
\end{barticle}
\endbibitem

\bibitem[\protect\citeauthoryear{Sz{\'e}kely, Rizzo and
  Bakirov}{2007}]{szekely2007measuring}
\begin{barticle}[author]
\bauthor{\bsnm{Sz{\'e}kely},~\bfnm{G{\'a}bor~J}\binits{G.~J.}},
  \bauthor{\bsnm{Rizzo},~\bfnm{Maria~L}\binits{M.~L.}} \AND
  \bauthor{\bsnm{Bakirov},~\bfnm{Nail~K}\binits{N.~K.}}
(\byear{2007}).
\btitle{Measuring and testing dependence by correlation of distances}.
\bjournal{The annals of statistics}
\bvolume{35}
\bpages{2769--2794}.
\end{barticle}
\endbibitem

\bibitem[\protect\citeauthoryear{Székely and Rizzo}{2013}]{SZEKELY2013193}
\begin{barticle}[author]
\bauthor{\bsnm{Székely},~\bfnm{Gábor~J.}\binits{G.~J.}} \AND
  \bauthor{\bsnm{Rizzo},~\bfnm{Maria~L.}\binits{M.~L.}}
(\byear{2013}).
\btitle{The distance correlation t-test of independence in high dimension}.
\bjournal{Journal of Multivariate Analysis}
\bvolume{117}
\bpages{193-213}.
\bdoi{https://doi.org/10.1016/j.jmva.2013.02.012}
\end{barticle}
\endbibitem

\bibitem[\protect\citeauthoryear{Wang et~al.}{2021}]{wang2021nonparametric}
\begin{bmisc}[author]
\bauthor{\bsnm{Wang},~\bfnm{Xueqin}\binits{X.}},
  \bauthor{\bsnm{Zhu},~\bfnm{Jin}\binits{J.}},
  \bauthor{\bsnm{Pan},~\bfnm{Wenliang}\binits{W.}},
  \bauthor{\bsnm{Zhu},~\bfnm{Junhao}\binits{J.}} \AND
  \bauthor{\bsnm{Zhang},~\bfnm{Heping}\binits{H.}}
(\byear{2021}).
\btitle{Nonparametric Statistical Inference via Metric Distribution Function in
  Metric Spaces}.
\end{bmisc}
\endbibitem

\bibitem[\protect\citeauthoryear{Wen et~al.}{2023}]{jiang2023}
\begin{btechreport}[author]
\bauthor{\bsnm{Wen},~\bfnm{Canhong}\binits{C.}},
  \bauthor{\bsnm{Wang},~\bfnm{Xueqin}\binits{X.}},
  \bauthor{\bsnm{Gao},~\bfnm{Zhe}\binits{Z.}} \AND
  \bauthor{\bsnm{Jiang},~\bfnm{Yunlu}\binits{Y.}}
(\byear{2023}).
\btitle{High-dimensional robust test of independence via Grothendieck's
  covariance}
\btype{Technical Report}.
\end{btechreport}
\endbibitem

\bibitem[\protect\citeauthoryear{Xiru}{1980}]{chen1980}
\begin{barticle}[author]
\bauthor{\bsnm{Xiru},~\bfnm{Chen}\binits{C.}}
(\byear{1980}).
\btitle{On limiting properties of U-statistics and von Mises statistics}.
\bjournal{Scientia Sinica (in Chinese)}
\bvolume{10}
\bpages{522-532}.
\end{barticle}
\endbibitem

\bibitem[\protect\citeauthoryear{Zhang}{2019}]{zhang2019bet}
\begin{barticle}[author]
\bauthor{\bsnm{Zhang},~\bfnm{Kai}\binits{K.}}
(\byear{2019}).
\btitle{BET on Independence}.
\bjournal{Journal of the American Statistical Association}
\bvolume{114}
\bpages{1620--1637}.
\end{barticle}
\endbibitem

\bibitem[\protect\citeauthoryear{Zhang et~al.}{2018}]{zhang2018large}
\begin{barticle}[author]
\bauthor{\bsnm{Zhang},~\bfnm{Qinyi}\binits{Q.}},
  \bauthor{\bsnm{Filippi},~\bfnm{Sarah}\binits{S.}},
  \bauthor{\bsnm{Gretton},~\bfnm{Arthur}\binits{A.}} \AND
  \bauthor{\bsnm{Sejdinovic},~\bfnm{Dino}\binits{D.}}
(\byear{2018}).
\btitle{Large-scale kernel methods for independence testing}.
\bjournal{Statistics and Computing}
\bvolume{28}
\bpages{113--130}.
\end{barticle}
\endbibitem

\end{thebibliography}

\end{document}